\documentclass[twoside]{article}
\pdfoutput=1
\usepackage[a4paper]{geometry}
\usepackage[latin1]{inputenc} % ou \usepackage[utf8]{inputenc}
\usepackage[T1]{fontenc} % ou \usepackage[OT1]{fontenc}
\usepackage{RR}
\usepackage{hyperref}
\usepackage{enumitem}
\usepackage{amsmath,amssymb,euscript}
\usepackage{dsfont}

\usepackage{graphicx}
\usepackage{color}

\usepackage{ntheorem}

\usepackage{float}

\floatstyle{ruled}
\newfloat{structure}{Hhtb}{los}
\floatname{structure}{Structure}

\usepackage{subfig}
\usepackage{algorithmic}
\usepackage{algorithm}
\usepackage{wrapfig}

\newcounter{prob}
        {\end{list}}
\newenvironment{proof}[1][{}]{%\novbskip%
  \begin{trivlist}\item[]\textit{Proof #1}\quad}%
  {\hfill\hspace*{\fill}~$\square$\end{trivlist}}

\theorembodyfont{\normalfont}
\newtheorem{thm}{Theorem}[section]
\newtheorem{prop}[thm]{Proposition}
\newtheorem{claim}[thm]{Claim}

\newtheorem{lem}[thm]{Lemma}
\newtheorem{de}[thm]{Definition}
\newtheorem{cor}[thm]{Corollary}

\definecolor{turquoise}{cmyk}{0.65,0,0.1,0.1}

\newcommand{\rawdef}[1]{\emph{#1}} % no index entry
\newcommand{\defn}[1]{\rawdef{#1}\index{#1}}

\newcommand{\Appref}[1]{Appendix~\ref{#1}}
\newcommand{\Algref}[1]{Algorithm~\ref{#1}}

\newcommand{\Corref}[1]{Corollary~\ref{#1}}

\newcommand{\Defref}[1]{Definition~\ref{#1}}
\newcommand{\Eqnref}[1]{Equation~\eqref{#1}}
\newcommand{\Figref}[1]{Figure~\ref{#1}}
\newcommand{\Lemref}[1]{Lemma~\ref{#1}}

\newcommand{\Secref}[1]{Section~\ref{#1}}

\newcommand{\Thmref}[1]{Theorem~\ref{#1}}
\newcommand{\Propref}[1]{Proposition~\ref{#1}}

\newcommand{\pickratio}{\alpha}
\DeclareMathOperator{\convh}{conv}
\newcommand{\convhull}[1]{\convh(#1)}
\newcommand{\pwrset}[1]{2^{#1}} % the set of subsets

\newcommand{\reel}{\mathbb{R}}

\newcommand{\rthree}{\reel^3}

\newcommand{\rem}{\reel^m}

\newcommand{\norm}[1]{\left\|#1\right\|}
\newcommand{\abs}[1]{\left|#1\right|}
\newcommand{\size}[1]{\left|#1\right|}

\newcommand{\dotprod}[2]{#1\cdot#2}

\newcommand{\tanspace}[2]{T_{#1}{#2}} % e.g. T_xS
\newcommand{\normspace}[2]{N_{#1}{#2}} % e.g. N_xS

\newcommand{\bdry}[1]{\partial{#1}}

\DeclareMathOperator{\starr}{star}
\newcommand{\str}[1]{\starr(#1)}
\DeclareMathOperator{\cosph}{cosph}
\newcommand{\sstr}[2]{\cosph^{#2}(#1)}
\newcommand{\asimplex}[1]{\{#1\}} % abstract simplex, i.e. $\asimplex{i,j,k}$
\newcommand{\simplex}[1]{[#1]} % Euclidean simplex
\newcommand{\carrier}[1]{\abs{#1}} % carrier of simplical complex

\newcommand{\seg}[2]{\simplex{#1,#2}} % segment

\newcommand{\pa}{{\rm p}}
\newcommand{\ed}{{\rm e}}

\newcommand{\ambdim}{N}

\newcommand{\amb}{\reel^{\ambdim}} % ambient space

\newcommand{\gdist}{d} % generic metric (distance function)
\newcommand{\dist}[2]{\gdist(#1,#2)}
\newcommand{\gdistG}[1]{\gdist_{#1}} % metric on space to be specified
\newcommand{\distG}[3]{\gdist_{#1}(#2,#3)}
\newcommand{\gdistEm}{\gdistG{\rem}} % Euclidean metric on R^m
\newcommand{\distEm}[2]{\distG{\rem}{#1}{#2}}

\newcommand{\gdistamb}{\gdistG{\amb}} % Euclidean metric in ambient space
\newcommand{\distamb}[2]{\distG{\amb}{#1}{#2}}
\newcommand{\gdistM}{\gdistG{\man}}
\newcommand{\distM}[2]{\distG{\man}{#1}{#2}}
\newcommand{\gdistRM}{\gdistG{{\amb|_{\man}}}}
\newcommand{\distRM}[2]{\distG{{\amb|_{\man}}}{#1}{#2}}

\newcommand{\close}[1]{\overline{#1}} % topological closure
 
\newcommand{\sphere}[1]{\mathbb{S}^{#1}}

\newcommand{\ballvol}[1]{V_{#1}} % vol of (#1)-dim'l. Eucl. unit ball
\newcommand{\ballvolm}{\ballvol{m}} % vol of (#1)-dim'l. Eucl. unit ball

\newcommand{\ball}[2]{B(#1,#2)} % generic ball
\newcommand{\cball}[2]{\close{B}(#1,#2)} % generic closed ball
\newcommand{\spaceball}[3]{B_{#1}(#2,#3)} % open ball on specified space
\newcommand{\cspaceball}[3]{\close{B}_{#1}(#2,#3)} % closed ball " " "

\newcommand{\ballE}[2]{\spaceball{\rthree}{#1}{#2}} % open ball in R^3 
\newcommand{\ballEd}[2]{\spaceball{\reel^d}{#1}{#2}} % open ball in R^d
\newcommand{\ballEm}[2]{\spaceball{\rem}{#1}{#2}} % open ball in R^m
\newcommand{\cballEm}[2]{\cspaceball{\rem}{#1}{#2}} % closed ball in R^m

\newcommand{\ballM}[2]{\spaceball{\man}{#1}{#2}}

\newcommand{\ballRM}[2]{\spaceball{\amb|_{\man}}{#1}{#2}}

\newcommand{\ballamb}[2]{\spaceball{\amb}{#1}{#2}}

\DeclareMathOperator{\aff}{aff} % affine hull
\newcommand{\affhull}[1]{\aff(#1)}

\newcommand{\angleop}[2]{\angle(#1,#2)}

\newcommand{\pts}{\mathsf{P}}
\newcommand{\tpts}{\tilde{\pts}}
\newcommand{\mpts}{\mathcal{P}} %{\mathfrak{P}} %{\pts}}

\DeclareMathOperator{\vol}{vol}
\newcommand{\man}{\mathcal{M}}
\newcommand{\tman}{\tilde{\man}}

\DeclareMathOperator{\Vor}{Vor}

\newcommand{\vorPd}{\Vor_{\gdist}(\pts)}
\newcommand{\vorcell}[2]{\mathcal{V}_{#1}(#2)}
\newcommand{\vorcelld}[1]{\vorcell{\gdist}{#1}}
\newcommand{\vorcellamb}[1]{\vorcell{\amb}{#1}}
\newcommand{\vorcellman}[1]{\vorcell{\man}{#1}}
\newcommand{\vorcelltman}[1]{\vorcell{\tman}{#1}} % for counterex.tex

\DeclareMathOperator{\Del}{Del}
\newcommand{\del}[2]{\Del_{#1}(#2)} % #1 metric identifier; #2 points
\newcommand{\delPd}{\del{\gdist}{\pts}} % Del wrt metric \gdist
\newcommand{\delof}[1]{\Del(#1)}
\newcommand{\delP}{\delof{\pts}}

\newcommand{\delM}[1]{\del{\man}{#1}} % intrinsic Dt
\newcommand{\delMmpts}{\delM{\mpts}}
\newcommand{\delRM}[1]{\del{\amb|_{\man}}{#1}} % rDt
\newcommand{\delRMmpts}{\delRM{\mpts}}
\newcommand{\tancplx}[1]{\del{T\man}{#1}}
\newcommand{\tancplxmpts}{\tancplx{\mpts}}

\newcommand{\pertconst}{\rho}

\newcommand{\samconst}{\epsilon}
\newcommand{\tsamconst}{\tilde{\epsilon}}
\newcommand{\tsparseconst}{\tilde{\mu}_0} % for ambient metric
\newcommand{\sparseconst}{\mu_0} %min d(p,q) > \sparseconst\samconst
\newcommand{\sparsity}{\lambda} % \sparsity = \sparseconst\samconst
 
\newcommand{\protconst}{\nu_0} % \delta = \protconst\samconst
\newcommand{\pdelta}{\check{\delta}} % for power protection in $\amb$
\newcommand{\pprotconst}{\delta_0} % \pdelta = \pprotconst\tsparseconst
\newcommand{\localconst}{\eta} % no global scope: no single meaning

\newcommand{\tthickbnd}{\tilde{\Upsilon}_0}

\newcommand{\pert}{\zeta} % perturbation function: P \to \tilde{P}
\newcommand{\pertiso}{\overset{\pert}{\cong}}

\newcommand{\incl}{\iota} % inclusion map P to \rem
\DeclareMathOperator{\interior}{int}
\newcommand{\intr}[1]{\interior(#1)}

\newcommand{\sing}[2]{s_{#1}(#2)}

\newcommand{\splxs}{\sigma}
\newcommand{\tsplxs}{\tilde{\sigma}}
\newcommand{\splxt}{\tau}

\newcommand{\splx}[1]{\sigma^{#1}} %simplex with dim
\newcommand{\tsplx}[1]{\tilde{\sigma}^{#1}} % tilde simplex 
\newcommand{\splxjoin}[2]{{#1}*{#2}}
\newcommand{\normhull}[1]{N(#1)}
\newcommand{\wnormhull}[2]{\normhull{#1,#2}} % #1 = splx; #2 = wght fct

\newcommand{\opface}[2]{#2_{#1}} % e.g., \sigma_p : face opposite p
\newcommand{\splxsp}{\opface{p}{\splxs}}
\newcommand{\splxsq}{\opface{q}{\splxs}}
\newcommand{\thickbnd}{\Upsilon_0}
\newcommand{\flakebnd}{\Gamma_0}
\newcommand{\thickness}[1]{\Upsilon(#1)}

\newcommand{\splxalt}[2]{D(#1,#2)} % altitude of #1 above face #2
\newcommand{\longedge}[1]{\Delta(#1)}
\newcommand{\shortedge}[1]{L(#1)}
\newcommand{\circrad}[1]{R(#1)}
\newcommand{\circcentre}[1]{C(#1)}
\newcommand{\wcircrad}[2]{\circrad{#1,#2}} % #1 = splx; #2 = wght fct
\newcommand{\wcirccentre}[2]{\circcentre{#1,#2}} % #1 = splx; #2 = wght fct

\newcommand{\gewf}[1]{\omega_{#1}} % elementary wt fct (global)
\newcommand{\ewf}[2]{\gewf{#1}(#2)}

\newcommand{\X}{X} % a generic set X
\newcommand{\reach}{\text{rch}(\man)}

\newcommand{\scurvbnd}{\mathcal{K}_0}
\DeclareMathOperator{\injr}{inj}
\newcommand{\injrad}[1]{\injr(#1)}
\newcommand{\injradM}{\injrad{\man}}

\RRNo{8273}
\RRdate{March 2013}
\RRauthor{% les auteurs
Jean-Daniel Boissonnat
  \and
Ramsay Dyer
\and
Arijit Ghosh\thanks{Indian Statistical Institute, ACM Unit,
 Kolkata, India}
}
\authorhead{Boissonnat, Dyer, \& Ghosh}
\titlehead{Constructing intrinsic DTs}
\RRtitle{Construire des triangulations de Delaunay intrins\`eque pour
  les sous-vari\'et\'es}
\RRresume{Nous montrons que, pour toute sous variété M
  suffisamment régulière de l'espace euclidien et pour tout
  échantillon P de points de M qui satisfait une condition locale
  de delta-généricité et de epsilon-densité, P admet une
  triangulation de Delaunay intrinsèque qui est égale à la
  triangulation de Delaunay restreinte à M et aussi au complexe de
  Delaunay tangent. Nous montrons également comment produire de tels
  ensembles de points.  
}
\RRmotcle{triangulation de Delaunay, sous-vari\'et\'es, g\'eom\'etrie
  riemannienne} 
\RRetitle{Constructing Intrinsic Delaunay Triangulations of Submanifolds}
\RRabstract{
  We describe an algorithm to construct an intrinsic Delaunay
  triangulation of a smooth closed submanifold of Euclidean space.
  Using results established in a companion paper on the stability of
  Delaunay triangulations on $\delta$-generic point sets, we establish
  sampling criteria which ensure that the intrinsic Delaunay complex
  coincides with the restricted Delaunay complex and also with the
  recently introduced tangential Delaunay complex.  The algorithm
  generates a point set that meets the required criteria while the
  tangential complex is being constructed.  In this way the
  computation of geodesic distances is avoided, the runtime is only
  linearly dependent on the ambient dimension, and the Delaunay
  complexes are guaranteed to be triangulations of the manifold.
}
\RRkeyword{Delaunay triangulation, submanifolds, Riemannian geometry}
\RRprojet{Geometrica}  % cas d'un seul projet
\RCSophia % Sophia Antipolis M\'editerran\'ee
\begin{document}
\makeRR   % cas d'un rapport de recherche
%% \makeRT % cas d'un rapport technique.
%% a partir d'ici, chacun fait comme il le souhaite

\clearpage
\tableofcontents

% \clearpage
% \maketitle
% \input{abstract_con}

\clearpage
% -*- LaTeX -*-
% intro_con.tex
% Introduction to construction version of the stabilty paper
% 20120504
%

\section{Introduction}

This paper addresses the problem of constructing an intrinsic Delaunay
triangulation of a smooth closed submanifold $\man \subset \amb$. We
present an algorithm which generates a point set $\mpts \subset \man$
and a simplicial complex on $\mpts$ that is homeomorphic to $\man$ and
has a connectivity determined by the Delaunay triangulation of $\mpts$
with respect to the intrinsic metric of $\man$.

% In a companion paper~\cite{boissonnat2012stab1} we demonstrated that
% sampling density alone is insufficient to guarantee that the intrinsic
% Delaunay complex is a triangulation of a manifold, contrary to a
% previous claim~\cite{leibon2000}. We introduced sampling conditions
% that would guarantee that the Delaunay complex is a manifold, and that
% it coincides with the restricted Delaunay complex. In this paper we
% present an algorithm which generates such a point set, and provide a
% guantee that the resulting complex is homeomorphic to $\man$. 

%%%%%%%%%%

For a submanifold of Euclidean space, the restricted Delaunay complex
\cite{edelsbrunner1997rdt}, which is defined by the ambient metric
restricted to the submanifold, was employed by Cheng et
al.~\cite{cheng2005} as the basis for a triangulation. However, it was
found that sampling density alone was insufficient to ensure a
triangulation, and manipulations of the complex were employed.

In an earlier work, Leibon and Letscher~\cite{leibon2000} announced
sampling density conditions which would ensure that the Delaunay
complex defined by the intrinsic metric of the manifold was a
triangulation. In fact, as shown in \Secref{sec:qualitative.counterex}
and \Appref{sec:counter.ex}, the stated result is incorrect: sampling
density alone is insufficient to guarantee an intrinsic Delaunay
triangulation (see \Thmref{thm:leibon.wrong}). Topological defects can
arise when the vertices lie too close to a degenerate or
``quasi-cospherical'' configuration.

% motivation
% We address this problem with the introduction of a parameterized
% notion of genericity for Delaunay complexes.  
Our interest in the intrinsic Delaunay complex stems from its close
relationship with other Delaunay-like structures that have been
proposed in the context of non-homogeneous metrics. For example,
anisotropic Voronoi diagrams \cite{labelle2003} and anisotropic
Delaunay triangulations emerge as natural structures when we want to
mesh a domain of $\rem$ while respecting a given metric tensor field.

This paper builds over preliminary results on anisotropic Delaunay
meshes \cite{boissonnat2011aniso.tr} and manifold reconstruction using
the tangential Delaunay complex \cite{boissonnat2011tancplx}. The
central idea in both cases is to define Euclidean Delaunay
triangulations locally and to glue these local triangulations together
by removing inconsistencies between them. We view the inconsistencies
as arising from instability in the Delaunay triangulations, and
% provide explicit bounds on the stability with respect to the
% genericity parameter.
exploit the results of a companion paper~\cite{boissonnat2012stab1}
to define sampling conditions under which these inconsistencies cannot
arise.

%%%%%%%%
The algorithm is based on the tangential Delaunay complex
\cite{boissonnat2011tancplx}, and is an adaptation of a Delaunay
refinement algorithm designed to avoid poorly shaped ``sliver'' simplices
\cite{li2003,boissonnat2010meshing}.  The tangential Delaunay complex
is defined with respect to local Delaunay triangulations restricted to
the tangent spaces at sample points. We demonstrate that the algorithm
produces sampling conditions such that the tangential Delaunay complex
coincides with the restricted Delaunay complex and the intrinsic
Delaunay complex. The refinement algorithm avoids the problem of
slivers without the need to resort to a point weighting strategy 
\cite{cheng2000,cheng2005,boissonnat2011tancplx}, which alters the
definition of the restricted Delaunay complex.

% Our interest in the intrinsic Delaunay complex stems from its close
% relationship with other Delaunay-like structures that have been
% proposed in the context of non-homogeneous metrics. For example,
% anisotropic Voronoi diagrams \cite{labelle2003} and anisotropic
% Delaunay triangulations \cite{boissonnat2011aniso.tr} emerge as
% natural structures when we want to mesh a domain of $\rem$ while
% respecting a given metric tensor field.

We present background and foundational material in
\Secref{sec:background}. Then, in \Secref{sec:equating}, we exploit
results established in \cite{boissonnat2012stab1} to demonstrate
sampling conditions under which the intrinsic Delaunay complex, the
restricted Delaunay complex, and the tangential Delaunay complex
coincide and are manifold.  The algorithm itself is presented in
\Secref{sec:algorithms}, and the analysis of the algorithm is
presented in \Secref{sec:alg.analysis}.

% -*- LaTeX -*-
% background.tex
% 20120426 from Stab I background (orig 20111026)
% resurrected from background_con.tex in svn 20120809 (it had been
% deleted in previous versions) 
%
% Background for Stability Construction paper

\section{Background}
\label{sec:background}

Within the context of the standard $m$-dimensional Euclidean space
$\rem$, when distances are determined by the standard norm,
$\norm{\cdot}$, we use the following conventions.  The distance
between a point $p$ and a set $\X \subset \rem$, is the infimum of the
distances between $p$ and the points of $\X$, and is denoted
$\distEm{p}{\X}$.  We refer to the distance between two points $a$ and
$b$ as $\norm{b-a}$ or $\distEm{a}{b}$ as convenient. A ball
$\ballEm{c}{r} = \{ x \, | \, \norm{x-c}< r \}$ is open, and
$\cballEm{c}{r}$ is its topological closure. Generally, we denote the
topological closure of a set $\X$ by $\close{\X}$, the interior by
$\intr{\X}$, and the boundary by $\bdry{\X}$. The convex hull is
denoted $\convhull{\X}$, and the affine hull is $\affhull{\X}$.

We will make use of other metrics besides the Euclidean one. A generic
metric is denoted $\gdist$, and the associated open and closed balls
are $\ball{c}{r}$, and $\cball{c}{r}$. If a specific metric is
intended, it will be indicated by a subscript, for example in
\Secref{sec:equating} we introduce $\gdistM$, the intrinsic metric on
a manifold $\man$, which has associated balls $\ballM{c}{r}$.

If $A$ is a $k \times j$ matrix, we denote its $i^{th}$ singular value
by $\sing{i}{A}$. We use the operator norm $\norm{A} = \sing{1}{A} =
\sup_{\norm{x}=1} \norm{Ax}$.
% , and employ the following standard
% observation:
% \begin{lem}
%   \label{lem:col.mat.norm}
%   If $\localconst > 0$ is an upper bound on the norms of the columns of $A$,
%   then $\norm{A} \leq \sqrt{j}\localconst$.
% \end{lem}
% We will also be interested in obtaining a lower bound on the smallest
% singular value, for which the following observation is useful:
% \begin{lem}
%   \label{lem:svd.pseudo.inv}
%   If $A$ is a $k \times j$ matrix of rank $j \leq k$, then the
%   \defn{pseudo inverse} $\pseudoinv{A} = (\transp{A}A)^{-1}\transp{A}$ is
%   the unique left inverse of $A$ whose kernel is the orthogonal
%   complement of the column space of $A$.  Furthermore,
%   \begin{equation*}
%     s_i(\pseudoinv{A}) = s_{j-i+1}(A)^{-1}.
%  \end{equation*}
% \end{lem}

If $U$ and $V$ are vector subspaces of $\rem$, with $\dim U \leq \dim
V$, the \defn{angle} between them is defined by 
% \begin{equation*}
%   \cos \angleop{U}{V} = \inf_{u \in U} \sup_{v \in V}
%   \frac{\transp{u}v}{\norm{u}\norm{v}}.  
% \end{equation*}
\begin{equation*}
  \sin \angleop{U}{V} = \sup_{u \in U} \norm{ u - \pi_V u},
\end{equation*}
where $\pi_V$ is the orthogonal projection onto $V$.
This is the largest principal angle between $U$ and $V$. The angle
between affine subspaces $K$ and $H$ is defined as the angle between
the corresponding parallel vector subspaces. 

\subsection{Sampling parameters and perturbations}

The structures of interest will be built from a finite set $\pts
\subset \rem$, which we consider to be a set of \defn{sample points}.
If $D \subset \rem$ is a bounded set, then $\pts$ is an
\defn{$\samconst$-sample set} for $D$ if $\distEm{x}{\pts} <
\samconst$ for all $x \in \close{D}$. We say that $\samconst$ is a
\defn{sampling radius} for $D$ satisfied by $\pts$.  If no domain $D$
is specified, we say $\pts$ is an \defn{$\samconst$-sample set} if
$\dist{x}{\pts \cup \bdry{\convhull{\pts}}} < \samconst$ for all $x
\in \convhull{\pts}$.  Equivalently, $\pts$ is an $\samconst$-sample
set if it satisfies a sampling radius $\samconst$ for
\begin{equation*}
  D_\samconst(\pts) = \{ x \in \convhull{\pts} \, | \,
  \distEm{x}{\bdry{\convhull{\pts}}} \geq \samconst \}.
\end{equation*}
In particular, if $\mpts$ is an $\samconst$-sample set for $U$, and
$\pts = U \cap \mpts$, and $\convhull{\pts} \subset U$, then $\pts$ is
an $\samconst$-sample set.

A set $\pts$ is \defn{$\sparsity$-sparse} if $\distEm{p}{q} >
\sparsity$ for all $p,q \in \pts$. We usually assume that the sparsity
of a $\samconst$-sample set is proportional to $\samconst$, thus:
$\sparsity = \sparseconst \samconst$.

We consider a perturbation of the points $\pts \subset \rem$ given by
a function $\pert: \pts \to \rem$. If $\pert$ is such that
$\distEm{p}{\pert(p)} \leq \pertconst$, we say that $\pert$ is a
\defn{$\pertconst$-perturbation}. As a notational convenience, we
frequently define $\tpts = \pert(\pts)$, and let $\tilde{p}$ represent
$\pert(p) \in \tpts$. We will only be considering
$\pertconst$-perturbations where $\pertconst$ is less than half the
sparsity of $\pts$, so $\pert: \pts \to \tpts$ is a bijection.

Points in $\pts$ which are not on the boundary of $\convhull{\pts}$
are \defn{interior points} of $\pts$.
% \begin{lem}
%   \label{lem:pert.int.pt}
%   Suppose $\pts$ is an $\samconst$-sample set, and $\pert: \pts \to
%   \tpts$ is a $\pertconst$-perturbation with $2\pertconst \leq
%   \samconst$. If point $p \in \pts$ satisfies
%   $\distEm{p}{\bdry{\convhull{\pts}}} \geq 3\samconst$, then
%   $\tilde{p} = \pert(p)$ is an interior point of~$\tpts$.
% \end{lem}
% \begin{proof}
%   Let $S = \bdry{B}$ be the bounding sphere for $B =
%   \ballEm{\tilde{p}}{3\samconst/2}$. Then $\distEm{p}{S} \geq
%   \samconst$ and for any $x \in S$,
%   $\distEm{x}{\bdry{\convhull{\pts}}} \geq \samconst$. Thus the
%   sampling assumption ensures that for every point $x \in S$, there is
%   a point $q \in \pts$ with $p \neq q$ and $\distEm{x}{q} <
%   \samconst$. It follows that $\distEm{x}{\pert(q)} < 3\samconst/2$,
%   and thus that $\tilde{p}$ is not the closest point in $\tpts$
%   for any point on $S$.

%   Thus $\tilde{p}$ cannot belong to $\bdry{\convhull{\tpts}}$. Indeed
%   if $\tilde{p} \in \bdry{\convhull{\tpts}}$, then take a unit vector
%   $v$ in an outward direction orthogonal to a closed half-space
%   supporting $\convhull{\tpts}$ at $\tilde{p}$. The ray from
%   $\tilde{p}$ defined by $v$ must intersect $S$ at some point $y$, and
%   $\tilde{p}$ would be the closest point in $\tpts$ to $y$, a
%   contradiction.
% \end{proof}

\subsection{Simplices}

% FIXME -- we should move to the abstract simplex viewpoint that was
% adopted in flat_pert. 
%
% TMP: \ednote{This sounds like a good idea Arijit. We should make it
%   global (in the background section if it isn't already.):} Let
% $\sigma$ be a subset of $\mpts$, we will use $\sigma$ to denote both
% the subset of points as well as simplex $\convh(\sigma)$ formed by the
% subset $\sigma$.  \ednote{Introduce abstract simplices as well, and
%   notation abuse ...}

Given a set of $j+1$ points $\asimplex{p_0, \ldots, p_j} \subset \pts
\subset \rem$, a (geometric) \defn{$j$-simplex} $\splxs =
\simplex{p_0, \ldots, p_j}$ is defined by the convex hull: $\splxs =
\convhull{\asimplex{p_0, \ldots, p_j}}$. The points $p_i$ are the
\defn{vertices} of $\splxs$. Any subset $\asimplex{p_{i_0}, \ldots,
  p_{i_k}}$ of $\asimplex{p_0, \ldots, p_j}$ defines a $k$-simplex
$\splxt$ which we call a \defn{face} of $\splxs$. We write $\splxt
\leq \splxs$ if $\splxt$ is a face of $\splxs$, and $\splxt < \splxs$
if $\splxt$ is a \defn{proper face} of $\splxs$, i.e., if the vertices
of $\splxt$ are a proper subset of the vertices of $\splxs$.

The \defn{boundary} of $\splxs$, is the union of its proper faces:
$\bdry{\splxs} = \bigcup_{\splxt < \splxs}\splxt$. In general this is
distinct from the topological boundary defined above, but we denote it
with the same symbol. The \defn{interior} of $\splxs$ is
$\intr{\splxs} = \splxs \setminus \bdry{\splxs}$. Again this is
generally different from the topological interior. 
% In particular, a $0$-simplex $p$ is equal to its interior: it has no
% boundary.
Other geometric properties of $\splxs$ include its diameter (the
length of its longest edge), $\longedge{\splxs}$, and the length of
its shortest edge, $\shortedge{\splxs}$. If $\splxs$ is a $0$-simplex,
we define $\shortedge{\splxs} = \longedge{\splxs} = 0$.

For any vertex $p \in \splxs$, the \defn{face oppposite} $p$ is the
face determined by the other vertices of $\splxs$, and is denoted
$\splxsp$. If $\splxt$ is a $j$-simplex, and $p$ is not a vertex of
$\splxt$, we may construct a $(j+1)$-simplex $\splxs =
\splxjoin{p}{\splxt}$, called the \defn{join} of $p$ and $\splxt$. It
is the simplex defined by $p$ and the vertices of $\splxt$, i.e.,
$\splxt = \splxsp$.

Our definition of a simplex has made an important departure from
standard convention: we do not demand that the vertices of a simplex
be affinely independent. A $j$-simplex $\splxs$ is a \defn{degenerate
  simplex} if $\dim \affhull{\splxs} < j$. If we wish to emphasise
that a simplex is a $j$-simplex, we write $j$ as a superscript:
$\splx{j}$; but this always refers to the \defn{combinatorial}
dimension of the simplex.
% , and is not generally assumed to reflect the
% dimension of the affine hull.
% , unless it has been established that $\splx{j}$ is
% not degenerate.

% Degenerate simplices are a complication that we choose to bring into
% the framework. A central theme underlying the observations of
% \Secref{sec:param.gen} is that emphasising a binary distinction
% between degeneracy and non-degeneracy may obscure the nature of
% problems that occur near degenracy.

If $\splxs$ is non-degenerate, then it has a \defn{circumcentre},
$\circcentre{\splxs}$, which is the centre of the smallest
circumscribing ball for $\splxs$. The radius of this ball is the
\defn{circumradius} of $\splxs$, denoted $\circrad{\splxs}$. 
A degenerate simplex may or may not have a circumcentre and
circumradius. We write $\circrad{\splxs}$ to indicate that a simplex
has a circumcentre. 
% In the
% context of the Euclidean Delaunay complexes we will work with, the
% degenerate simplices we may encounter also have these properties.
% The ratio of the circumradius to the shortest edge is denoted
% $\radedge{\splxs} = \circrad{\splxs}/\shortedge{\splxs}$.
We will make use of the affine space $\normhull{\splxs}$ composed of
the centres of the balls that circumscribe $\splxs$. We sometimes
refer to a point $c \in \normhull{\splxs}$ as a \defn{centre for
  $\splxs$}. The space $\normhull{\splxs}$ is orthogonal to
$\affhull{\splxs}$ and intersects it at the circumcentre of
$\splxs$. Its dimension is $m - \dim \affhull{\splxs}$.

% Degenerate simplices are a problem, but so are simplices that are
% arbitrarily close to being degenerate.  
The \defn{altitude} of $p$ in $\splxs$ is $\splxalt{p}{\splxs} =
\distEm{p}{\affhull{\splxsp}}$. A poorly-shaped simplex can be
characterized by the existence of a relatively small altitude. The
\defn{thickness} of a $j$-simplex $\splxs$ is the dimensionless
quantity
\begin{equation*}
  \thickness{\splxs} =
  \begin{cases}
    1& \text{if $j=0$} \\
    \min_{p \in \splxs} \frac{\splxalt{p}{\splxs}}{j
      \longedge{\splxs}}& \text{otherwise.}
  \end{cases}
\end{equation*}
We say that $\splxs$ is $\thickbnd$-thick, if $\thickness{\splxs} \geq
\thickbnd$. If $\splxs$ is $\thickbnd$-thick, then so are
all of its faces. Indeed if $\splxt \leq \splxs$, then the smallest
altitude in $\splxt$ cannot be smaller than that of $\splxs$, and also
$\longedge{\splxt} \leq \longedge{\splxs}$.

Although he worked with volumes rather than altitudes,
Whitney~\cite[p. 127]{whitney1957} proved that the affine hull of a
thick simplex makes a small angle with any hyperplane which lies near
all the vertices of the simplex. We can state this~\cite[Lemma
2.5]{boissonnat2012stab1} as:
\begin{lem}[Whitney angle bound]
  \label{lem:whitney.approx}
  Suppose $\splxs$ is a $j$-simplex whose vertices all lie within a
  distance $\localconst$ from a $k$-dimensional affine space, $H
  \subset \rem$, with $k \geq j$. Then
  \begin{equation*}
    \sin \angleop{\affhull{\splxs}}{H} \leq
    \frac{2\localconst}{\thickness{\splxs}\longedge{\splxs}}.
  \end{equation*}
\end{lem}

%%%%%%%%%%%%%%%%%%%%%%%%%%%%%%%%%%%%%%%%%%%%%%%%%%%%%%%%%%%%
% from simplex_pert.tex

\subsubsection{Simplex perturbation}
\label{sec:splx.pert}

We will make use of two results displaying the robustness of
thick simplices with respect to small perturbations of their
vertices.  The first
observation bounds the change in thickness itself under small
perturbations:
\begin{lem}[Thickness under perturbation]
  \label{lem:pert.thick.bnd}
  Let $\splxs = \simplex{p_0,\ldots,p_j}$ and $\tilde{\splxs} =
  \simplex{\tilde{p}_0, \ldots, \tilde{p}_j}$ be  $j$-simplices such
  that $\norm{\tilde{p}_i - p_i} \leq \rho$ for all $i \in \{0,\ldots,
  j \}$. For any positive $\localconst \leq 1$, if
  \begin{equation}
    \label{eq:thick.pert}
    \rho \leq \frac{(1-\localconst)\thickness{\splxs}^2 \shortedge{\splxs}}{14},
  \end{equation}
  then
  \begin{equation*}
    \splxalt{\tilde{p}_i}{\tilde{\splxs}} \geq \localconst
    \splxalt{p_i}{\splxs}, 
  \end{equation*}
  for all $i \in \{0, \ldots, j \}$. It follows that
  \begin{equation*}
    \thickness{\tilde{\splxs}}\longedge{\tilde{\splxs}} \geq \localconst
    \thickness{\splxs} \longedge{\splxs}, 
  \end{equation*}
  and
  \begin{equation*}
    \thickness{\tilde{\splxs}} \geq \left( 1 -
      \frac{2\rho}{\longedge{\splxs}} \right) \localconst \thickness{\splxs}
    \geq \frac{6}{7}\localconst\thickness{\splxs}.
  \end{equation*}
\end{lem}
\begin{proof}
  Let $p,q \in \splxs$ with $\tilde{p},\tilde{q}$ the corresponding
  vertices of $\tilde{\splxs}$. Let $v = p-q$ and $\tilde{v} =
  \tilde{p} - \tilde{q}$. Define $\theta =
  \angleop{v}{\affhull{\splxsp}}$ and $\tilde{\theta} =
  \angleop{\tilde{v}}
  {\affhull{\opface{\tilde{p}}{\tilde{\splxs}}}}$. Since
  $\thickness{\splxs} \leq \thickness{\splxsp}$, Whitney's
  \Lemref{lem:whitney.approx} lets us bound
  $\angleop{\affhull{\splxsp}}{\affhull{\opface{\tilde{p}}{\tilde{\splxs}}}}$
  by the angle $\alpha$ defined by
  \begin{equation*}
    \sin \alpha = \frac{2\rho}{\thickness{\splxs}{\longedge{\splxs}}}.
  \end{equation*}
  Also, by an elementary geometric argument,
  \begin{equation*}
    \sin \gamma = \frac{2\rho}{\norm{v}}
  \end{equation*}
  defines $\gamma$ as an upper bound on the angle between the lines
  generated by $v$ and $\tilde{v}$.

  Thus we have
  \begin{equation*}
    \splxalt{\tilde{p}}{\tilde{\splxs}} =
    \norm{\tilde{v}}\sin \tilde{\theta} \geq
    (\norm{v} - 2\rho)\sin (\theta - \alpha - \gamma).
  \end{equation*}
  Using the addition formula for sine together with the facts that for
  $x,y \in [0,\frac{\pi}{2}]$, $(1-x) \leq \cos x$; $2\sin x \geq x$;
  and $\sin x + \sin y \geq \sin(x+y)$, we get
  \begin{equation*}
    \splxalt{\tilde{p}}{\tilde{\splxs}} \geq (\norm{v} - 2\rho)
    \left[ \left( 1 - 2 \left( \frac{2\rho}{\thickness{\splxs}
            \longedge{\splxs}} + \frac{2\rho}{\norm{v}} \right)
      \right) \frac{\splxalt{p}{\splxs}}{\norm{v}}
      - \left( \frac{2\rho}{\thickness{\splxs}
          \longedge{\splxs}} + \frac{2\rho}{\norm{v}} \right)
    \right].
  \end{equation*}
  For convenience, define $\mu = \frac{2\rho}{\shortedge{\splxs}}
  \geq \frac{2\rho}{\norm{v}} \geq
  \frac{2\rho}{\longedge{\splxs}}$. Then
  \begin{equation*}
    \begin{split}
      \splxalt{\tilde{p}}{\tilde{\splxs}}
      &\geq \norm{v}(1 - \mu) \left[ \left( 1 - 2\left( 1 +
            \frac{1}{\thickness{\splxs}} \right)\mu \right)
        \frac{\splxalt{p}{\splxs}}{\norm{v}} -
        \left( 1 + \frac{1}{\thickness{\splxs}} \right) \mu
      \right]\\
      &\geq (1 - \mu) \left[ \left( 1 - 
          \frac{4\mu}{\thickness{\splxs}} \right)
        \splxalt{p}{\splxs}
        - \frac{2\mu \norm{v}}{\thickness{\splxs}} \right] \\
      &\geq (1 - \mu) \left[ \left( 1 - 
          \frac{4\mu}{\thickness{\splxs}} \right)
        \splxalt{p}{\splxs}
        - \frac{2\mu \norm{v}}{\thickness{\splxs}^2
          \longedge{\splxs}} \splxalt{p}{\splxs} \right] \\
      &\geq (1 - \mu) \left( 1 - 
        \frac{4\mu}{\thickness{\splxs}}
        - \frac{2\mu}{\thickness{\splxs}^2} \right)
      \splxalt{p}{\splxs} \\
      &\geq \left( 1 - \frac{7\mu}{\thickness{\splxs}^2} \right)
      \splxalt{p}{\splxs} \\
      &\geq K         \splxalt{p}{\splxs} \qquad \qquad \text{ when
      }
      \mu \leq \frac{(1 - K)\thickness{\splxs}^2}{7}.
    \end{split}
  \end{equation*}
  The condition on $\mu$ is satisfied when $\rho$ satisfies
  Inequality~\eqref{eq:thick.pert}.

  The bound on $\thickness{\tilde{\splxs}} \longedge{\tilde{\splxs}}$
  follows immediately from the bounds on the
  $\splxalt{\tilde{p}}{\tilde{\splxs}}$, and the bound on
  $\thickness{\tilde{\splxs}}$ itself follows from the observation that
  \begin{equation*}
    \frac{\longedge{\splxs}}{\longedge{\tilde{\splxs}}} \geq
    \frac{\longedge{\splxs}}{\longedge{\splxs} + 2\rho} \geq       
    \left( 1 - \frac{2\rho}{\longedge{\splxs}} \right) \geq
    \left( 1 - \frac{\thickness{\splxs}^2}{7} \right) \geq
    \frac{6}{7},
  \end{equation*}
  when $\rho$ satisfies Inequality~\eqref{eq:thick.pert}. 
\end{proof}

%%%%%%%%%%%%%%%%%%%%%%%%%%%%%%%%%%%%%%%%%%%%%%%%%%%%%%%
% New circentre pert:

We will also make use of a bound relating circumscribing balls of a
simplex that undergoes a perturbation:
\begin{lem}[Circumscribing balls under perturbation]
  \label{lem:pert.circ.ball}
  Let $\splxs = \simplex{p_0,\ldots,p_j}$ and $\tilde{\splxs} =
  \simplex{\tilde{p}_0, \ldots, \tilde{p}_j}$ be  $j$-simplices such
  that $\norm{\tilde{p}_i - p_i} \leq \rho$ for all $i \in \{0,\ldots,
  j \}$.
 Suppose $B= \ballEm{c}{r}$, with $r < \samconst$, is a
 circumscribing ball for $\splxs$.
%  a $j$-simplex $\splxs = \simplex{p_0,\ldots,p_j}$. Suppose that
%  another simplex $\tilde{\splxs} = \simplex{\tilde{p}_0,\ldots,
%    \tilde{p}_j}$ is such that
% %  $\tilde{p}_0 = p_0$ and
%   $\norm{\tilde{p}_i - p_i} \leq \pertconst$ for all $i \in
%   [1,\ldots,j]$. 
  If
  \begin{equation*}
    \pertconst \leq   \frac{\thickness{\splxs}^2 \shortedge{\splxs}}{28},
  \end{equation*}
  then there is a circumscribing ball $\tilde{B} =
  \ballEd{\tilde{c}}{\tilde{r}}$ for $\tilde{\splxs}$ with
  \begin{equation}
    \label{eq:pert.bound.cc}
%    \abs{\tilde{r} - r} \leq
    \norm{\tilde{c} - c} <
    \left(\frac{8\samconst}{\thickness{\splxs}\longedge{\splxs}}
    \right) \pertconst
  \end{equation}
  % If instead of equality we have $\norm{\tilde{p}_0 - p_0} \leq
  % \pertconst$, then the bound on $\norm{\tilde{c}-c}$ still applies,
  and 
  \begin{equation*}
    \abs{\tilde{r} - r} <
    \left(\frac{9\samconst}{\thickness{\splxs}\longedge{\splxs}}
    \right) \pertconst.    
  \end{equation*}
  If, in addition, we have that $\tilde{p}_0 = p_0$, then
  $\abs{\tilde{r}-r} \leq \norm{\tilde{c}-c}$, and
  \eqref{eq:pert.bound.cc} serves also as a bound on $\abs{\tilde{r}-r}$.
%  so the radius satisfies the same bound~\eqref{eq:pert.bound.cc}.
\end{lem}
\begin{proof}
  By the perturbation bounds, the distances between $c$ and the
  vertices of $\tsplxs$ differ by no more than $2\pertconst$.
  % $c$ is a $\relconst$-centre for
  % $\tsplxs$ with $\relconst = 2\pertconst$.
  Also, $\norm{c-p_i} < \tsamconst = \samconst + \pertconst$, and so
  by \cite[Lemma 4.3]{boissonnat2012stab1} we have
%  \Lemref{lem:close.to.centres} we have
  \begin{equation*}
    \distEm{c}{\normhull{\tsplxs}} <
    \frac{2\tsamconst\pertconst}{\thickness{\tsplxs}\longedge{\tsplxs}}.
  \end{equation*}
  The bound on $\pertconst$ allows us to apply
  \Lemref{lem:pert.thick.bnd} with $K= \frac{1}{2}$, so
  $\thickness{\tsplxs}\longedge{\tsplxs} \geq
  \frac{1}{2}\thickness{\splxs}\longedge{\splxs}$, and we obtain the
  bound on $\norm{\tilde{c}-c}$ with the observation that $\tsamconst
  \leq 2\samconst$. Indeed, $\pertconst \leq \samconst$ because
  $\shortedge{\splxs} \leq 2\samconst$.

  By the triangle inequality $\abs{\tilde{r}-r} \leq
  \norm{\tilde{p}_0-p_0} + \norm{\tilde{c}-c}$, and the stated bound
  on $\abs{\tilde{r}-r}$ follows from the observation that
  $\frac{\samconst}{\thickness{\splxs}\longedge{\splxs}} \geq 1$ if $j
  > 1$.  Under the assumption that $\tilde{p}_0=p_0$, the bound on
  $\norm{\tilde{c}-c}$ also serves as a bound on $\abs{\tilde{r}-r}$.
 % The assumption can be replaced by
  % $\norm{\tilde{p}_0-p_0} \leq \pertconst$ without affecting the bound
  % on $\norm{\tilde{c}-c}$, and the bound on $\abs{\tilde{r}-r}$ becomes
  % \begin{equation*}
  %   \abs{\tilde{r} - r} <
  %   \left(\frac{9\samconst}{\thickness{\splxs}\longedge{\splxs}}
  %   \right) \pertconst,
  % \end{equation*}
  % by the triangle inequality and the observation that
  % $\frac{\samconst}{\thickness{\splxs}\longedge{\splxs}} \geq 1$.
\end{proof}

\subsubsection{Flakes}
\label{sec:thin.flakes}

For algorithmic reasons, it is convenient to have a more structured
constraint on simplex geometry than that provided by a simple
thickness bound. A simplex that is not thick has a relatively small
altitude, but we wish to exploit a family of bad simplices for which
\emph{all} the altitudes are relatively small.  As shown by
\Lemref{lem:thin.flake.alt.bnd} below, the $\flakebnd$-flakes form
such a family. The flake parameter $\flakebnd$ is a positive real
number smaller than one.
\begin{de}[$\flakebnd$-good simplices and $\flakebnd$-flakes]
  \label{def:thin.flake}
  \label{def:good.simplex}
  A simplex $\splxs$ is \defn{$\flakebnd$-good} if
  $\thickness{\splxs^j} \geq \flakebnd^j$ for all $j$-simplices
  $\splxs^j \leq \splxs$. A simplex is \defn{$\flakebnd$-bad} if it
  is not $\flakebnd$-good. A \defn{$\flakebnd$-flake} is
  a $\flakebnd$-bad simplex in which all the proper faces are
  $\flakebnd$-good.
\end{de}
Observe that a flake must have dimension at least $2$, since
$\thickness{\splxs^j} = 1$ for $j < 2$.  A flake that has
an upper bound on the ratio of its circumradius to its shortest edge
is called a \defn{sliver}. The flakes we will be considering have no
upper bound on their circumradius, and in fact they may be degenerate
and not even have a circumradius.

Ensuring that all simplices are $\flakebnd$-good is the same as
ensuring that there are no flakes. 
Indeed, if $\splxs$ is $\flakebnd$-bad, then it has a $j$-face $\splxs^j
\leq \splxs$ that is not $\flakebnd^j$-thick. By considering such a
face with minimal dimension we arrive at the following important
observation: 
\begin{lem}
  \label{lem:bad.has.flake}
  A simplex is $\flakebnd$-bad if and only if it has a face that is a
  $\flakebnd$-flake.
\end{lem}

We obtain an upper bound on the altitudes of a $\flakebnd$-flake
through a consideration of dihedral angles. In particular, we observe
the following general relationship between simplex altitudes:
\begin{lem}
  \label{lem:alt.ratios}
  If $\splxs$ is a $j$-simplex with $j \geq 2$, then for any two
  vertices $p,q \in \splxs$, the dihedral angle between $\splxsp$ and
  $\splxsq$ defines an equality between ratios of altitudes:
  \begin{equation*}
    \sin \angleop{\affhull{\splxsp}}{\affhull{\splxsq}} =
    \frac{\splxalt{p}{\splxs}}{\splxalt{p}{\splxsq}}
    =
    \frac{\splxalt{q}{\splxs}}{\splxalt{q}{\splxsp}}.
  \end{equation*}
\end{lem}
\newcommand{\spq}{\splxs_{pq}}
\begin{proof}
  Let $\spq = \splxsp \cap \splxsq$, and let $p_*$ be the projection
  of $p$ into $\affhull{\spq}$. Taking $p_*$ as the origin, we see
  that $\frac{p-p_*}{\splxalt{p}{\splxsq}}$ has the maximal distance
  to $\affhull{\splxsp}$ out of all the unit vectors in
  $\affhull{\splxsq}$, and this distance is
  $\frac{\splxalt{p}{\splxs}}{\splxalt{p}{\splxsq}}$. By definition
  this is the sine of the angle between $\affhull{\splxsp}$ and
  $\affhull{\splxsq}$. A symmetric argument is carried out with $q$ to
  obtain the result.
\end{proof}

% \begin{lem}
%   \label{lem:alt.ratios}
%   If $\splxs$ is a $j$-simplex with $j \geq 2$, then for any two
%   vertices $p,q \in \splxs$, the following equality between ratios of
%   altitudes holds:
%   \begin{equation*}
%     \frac{\splxalt{p}{\splxs}}{\splxalt{q}{\splxs}}
%     =
%     \frac{\splxalt{p}{\opface{q}{\splxs}}}{\splxalt{q}{\splxsp}}.
%   \end{equation*}
% \end{lem}
% \begin{proof}
%   The volume of a $j$-simplex $\splxs$ can be defined recursively by
%   $\volop{\splxs} = 1$ if $j=0$, and
%   \begin{equation*}
%     \volop{\splxs} = \frac{\splxalt{p}{\splxs} \volop{\splxsp}} {j},
%   \end{equation*}
%   when $j>0$ and $p \in \splxs$ is an arbitrary vertex. It follows
%   then that
%   \begin{equation*}
%     \splxalt{p}{\splxs} \volop{\splxsp} = \splxalt{q}{\splxs} \volop{\splxsq}.
%   \end{equation*}
%   If we define $\splxt$ by $\splxsp = \splxjoin{q}{\splxt}$, then it
%   follows that $\splxsq = \splxjoin{p}{\splxt}$. Thus
%   \begin{equation*}
%     \volop{\splxsp} = \frac{\splxalt{q}{\splxsp} \volop{\splxt}} {j-1},
%   \end{equation*}
%   and 
%   \begin{equation*}
%     \volop{\splxsq} = \frac{\splxalt{p}{\splxsq} \volop{\splxt}} {j-1},
%   \end{equation*}
%   and the result follows.
% \end{proof}

We arrive at the following important observation about flake simplices:
\begin{lem}[Flakes have small altitudes]
  \label{lem:thin.flake.alt.bnd}
  If a $k$-simplex $\splxs$ is a $\flakebnd$-flake, then for
  every vertex $p \in \splxs$, the altitude satisfies the bound
  \begin{equation*}
    \splxalt{p}{\splxs} < \frac{k \longedge{\splxs}^2
      \flakebnd} {(k-1) \shortedge{\splxs}}.
  \end{equation*}
\end{lem}
\begin{proof}
  Recalling \Lemref{lem:alt.ratios} we have
  \begin{equation*}
    \splxalt{p}{\splxs} = \frac{\splxalt{q}{\splxs}
      \splxalt{p}{\splxsq} } {\splxalt{q}{\splxsp}},
  \end{equation*}
  and taking $q$ to be a vertex with minimal altitude, we have
  \begin{equation*}
    \splxalt{q}{\splxs} = k \thickness{\splxs} \longedge{\splxs}
    < k \flakebnd^k \longedge{\splxs},
  \end{equation*}
  and
  \begin{equation*}
    \splxalt{q}{\splxsp} \geq (k-1) \thickness{\splxsp}
    \longedge{\splxsp} \geq (k-1) \flakebnd^{k-1} \shortedge{\splxs},
  \end{equation*}
  and
  \begin{equation*}
    \splxalt{p}{\splxsq} \leq \longedge{\splxsq} \leq \longedge{\splxs},
  \end{equation*}
  and the bound is obtained.
\end{proof}

%%%%%%%%%%%%%%%%%%%%%%%%%%%%%%%%%%%%%%%%%%%%%%%%%%%%%%%%%%%%

\subsection{Complexes}

Given a finite set $\pts$, an \defn{abstract simplicial complex} is a
set of subsets $K \subset \pwrset{\pts}$ such that if $\splxs \in K$,
then every subset of $\splxs$ is also in $K$. 
The Delaunay complexes we study are abstract simplicial complexes, but
their simplices carry a canonical geometry induced from the inclusion
map $\incl: \pts \hookrightarrow \rem$. (We assume $\incl$ is
injective on $\pts$, and so do not distinguish between $\pts$ and
$\incl(\pts)$.) To each abstract simplex $\splxs \in K$, we have an
associated geometric simplex $\convhull{\incl(\splxs)}$, and normally
when we write $\splxs \in K$, we are referring to this geometric
object.  Occasionally, when it is convenient to emphasise a
distinction, we will write $\incl(\splxs)$ instead of $\splxs$.

Thus we view such a $K$ as a set of simplices in $\rem$, and we refer
to it as a \defn{complex}, but it is not generally a (geometric)
simplicial complex.  A geometric \defn{simplicial complex} is a
finite collection $G$ of non-degenerate simplices in $\amb$ such that
if $\splxs \in G$, then all of the faces of $\splxs$ also belong to
$G$, and if $\splxs, \tsplxs \in G$ and $\splxt = \splxs \cap \tsplx
\neq \emptyset$, then $\splxt \leq \splxs$ and $\splxt \leq
\tsplxs$. 
An abstract simplicial complex is defined from a geometric simplicial
complex in an obvious way.  A \defn{geometric realization} of an
abstract simplicial complex $K$ is a geometric simplicial complex
whose associated abstract simplicial complex may be identified with
$K$.  A geometric realization always exists for any complex. Details
can be found in algebraic topology textbooks; the book by
Munkres~\cite{munkres1984} for example.

The \defn{carrier} of an abstract complex $K$ is the underlying
topological space $\carrier{K}$, associated with a geometric
realization of $K$.  Thus if $G$ is a geometric realization of $K$,
then $\carrier{K} = \bigcup_{\splxs \in G} \splxs$.  For our
complexes, the inclusion map $\incl$ induces a continous map $\incl:
\carrier{K} \to \rem$, defined by barycentric interpolation on each
simplex. If this map is injective, we say that $K$ is
\defn{embedded}. In this case $\incl$ also defines a geometric
realization of $K$, and we may identify the carrier of $K$ with the
image of $\incl$.
% Otherwise, if $\incl$ is not
% injective, we distinguish between $\carrier{K}$ and
% $\incl(\carrier{K})$.
%
% , and refer to the latter as the
% \defn{shadow}~\cite{ghrist2011} of $K$.

A subset $K' \subset K$ is a \defn{subcomplex} of $K$ if it is also a
complex.  The \defn{star} of a subcomplex $K' \subseteq K$ is the
subcomplex generated by the simplices incident to $K'$. I.e., it is
all the simplices that share a face with a simplex of $K'$, plus all
the faces of such simplices. This is a departure from a common usage
of this same term in the topology literature. The star of $K'$ is
denoted $\str{K'}$ when there is no risk of ambiguity, otherwise we
also specify the parent complex, as in $\str{K';K}$.

A \defn{triangulation} of $\pts \subset \rem$ is an embedded
complex $K$ with vertices $\pts$ such that $\carrier{K} =
\convhull{\pts}$.
\begin{de}[Triangulation at a point]
A complex $K$ is a \defn{triangulation at $p \in
  \rem$} if:
\begin{itemize}[noitemsep,topsep=0pt,parsep=0pt,partopsep=0pt]
\item $p$ is a vertex of $K$.
\item $\str{p}$ is embedded.
\item $p$  lies in $\intr{\carrier{\str{p}}}$.
\item For all $\splxt \in K$, and $\splxs \in \str{p}$, if
  $\intr{\splxt} \cap \splxs \neq \emptyset$, then $\splxt \in
  \str{p}$.
\end{itemize}
\end{de}
A complex $K$ is a \defn{$j$-manifold complex} if the star of every
vertex is isomorphic to the star of a triangulation of $\reel^j$.

If $\splxs$ is a simplex with vertices in $\pts$, then any map $\pert:
\pts \to \tpts \subset \rem$ defines a simplex $\pert(\splxs)$ whose
vertices in $\tpts$ are the images of vertices of $\splxs$. If $K$ is
a complex on $\pts$, and $\tilde{K}$ is a complex on $\tpts$, then
$\pert$ induces a \defn{simplicial map} $K \to \tilde{K}$ if
$\pert(\splxs) \in \tilde{K}$ for every $\splxs \in K$. We denote this
map by the same symbol, $\pert$. We are interested in the %particular
case when $\pert$ is an \defn{isomorphism}, which means it establishes
a bijection between $K$ and $\tilde{K}$. We then say that $K$ and
$\tilde{K}$ are \defn{isomorphic}, and write $K \cong \tilde{K}$, or
$K \pertiso \tilde{K}$ if we wish to emphasise that the correspondence
is given by $\pert$.

% A simplicial map $\pert: K \to \tilde{K}$ defines a continuous map
% $\pert: \carrier{K} \to \carrier{\tilde{K}}$, by barycentric
% interpolation on each simplex $\splxs \in K$.  
We use the
following local version of a standard result \cite[Lemma
2.7]{boissonnat2012stab1}: 
\begin{lem}
  \label{lem:inject.triang}
  Suppose $K$ is a complex with vertices $\pts \subset \rem$, and
  $\tilde{K}$ a complex with vertices $\tpts \subset \rem$. Suppose
  also that $K$ is a triangulation at $p \in \pts$, and that $\pert:
  \pts \to \tpts$ induces an injective simplicial map $\str{p} \to
  \str{\pert(p)}$. If $\tilde{K}$ is a triangulation at $\pert(p)$,
  then
  \begin{equation*}
   \pert(\str{p}) = \str{\pert(p)}.
 \end{equation*}
\end{lem}

%%%%%%%%%%%%%%%%%%%%%%%%%%%%%%%%%%%%%%%%%%%
% param_gen.tex

\subsection{The Delaunay complex}

An \defn{empty ball} is one that contains no point from $\pts$. 
\begin{de}[Delaunay complex]
  \label{def:Delaunay.complex}
  A \defn{Delaunay ball} is a maximal empty ball. Specifically, $B =
  \ballEm{x}{r}$ is a Delaunay ball if any empty ball centred at $x$
  is contained in $B$. A simplex $\splxs$ is a \defn{Delaunay
    simplex}, if there exists some Delaunay ball $B$ such that the
  vertices of $\splxs$ belong to $\bdry{B}\cap \pts$.  The
  \defn{Delaunay complex} is the set of Delaunay simplices, and is
  denoted $\delP$.
\end{de}
The Delaunay complex has the combinatorial structure of an abstract
simplicial complex, but $\delP$ is embedded only when
$\pts$ satisfies appropriate genericity requirements
\cite{boissonnat2012stab1}. 

%%%%%%%%%%%%%%%%%%%%%%%%%%%%%%%%%%%%%%%%%%%%
% protection.tex

\subsubsection{Protection}
\label{sec:protection}

A Delaunay simplex $\splxs$ is \defn{$\delta$-protected} if it has a
Delaunay ball $B$ such that $\distEm{q}{\bdry{B}} > \delta$ for all
$q \in \pts \setminus \splxs$.  We say that $B$ is a
$\delta$-protected Delaunay ball for $\splxs$. If $\splxt < \splxs$,
then $B$ is also a Delaunay ball for $\splxt$, but it cannot be a
$\delta$-protected Delaunay ball for $\splxt$.
% Thus we can talk
% about a $\delta$-protected Delaunay ball, without explicit reference
% to the simplex being protected: it must be the simplex defined by all
% the points in $\bdry{B} \cap \pts$. 
We say that $\splxs$ is \defn{protected} to mean that it is
$\delta$-protected for some unspecified $\delta > 0$.
\begin{de}[$\delta$-generic]
  \label{def:delta.generic}
  A finite set of points $\pts \subset \rem$ is
  \defn{$\delta$-generic} if all the Delaunay $m$-simplices are
  $\delta$-pro\-tec\-ted. The set $\pts$ is simply \defn{generic} if it is
  $\delta$-generic for some unspecified $\delta > 0$.
\end{de}
% In other words, $\pts$ is $\delta$-generic if for all $\splx{m} \in
% \delP$, we have $\distEm{\circcentre{\splx{m}}}{q} >
% \circrad{\splx{m}} + \delta$ for all $q \in \pts \setminus \splx{m}$.

We have previously demonstrated~\cite{boissonnat2012stab1} that
$\delta$-generic point sets impart a quantifiable stability on the
Delaunay complex.  In \Secref{sec:equating} we review the main
stability result and develop it to define the sampling conditions that
will be met by the algorithm that we introduce in
\Secref{sec:algorithms}.

%%%%%%%%%%%%%%%%%%%%%%

\subsubsection{The Delaunay complex in other metrics}
\label{sec:Delaunay.alt.metric}

We will also consider the Delaunay complex defined with respect to a
metric $\gdist$ on $\rem$ which differs from the Euclidean
one. Specifically, if $\pts \subset U \subset \rem$ and $\gdist: U
\times U \to \reel$ is a metric, then we define the Delaunay complex
$\delPd$ with respect to the metric $\gdist$.

The definitions are exactly analogous to the Euclidean case: A
Delaunay ball is a maximal empty ball $\ball{x}{r}$ in the metric
$\gdist$. The resulting Delaunay complex $\delPd$ consists of all the
simplices which are circumscribed by a Delaunay ball with respect to
the metric $\gdist$. The simplices of $\delPd$ are, possibly
degenerate, geometric simplices in $\rem$. As for $\delP$, $\delPd$
has the combinatorial structure of an abstract simplicial complex, but
unlike $\delP$, $\delPd$ may fail to be embedded even
when there are no degenerate simplices.

%%%%%%%%%%%%%%%%%%%

\subsubsection{Obtaining Delaunay triangulations in other metrics}
\label{sec:qualitative.counterex}

Delaunay~\cite{delaunay1934} showed that if $\pts \subset \rem$ is
generic, then $\delP$ is a triangulation. Point sets that are not
generic are often dismissed in theoretical work, because an
arbitrarily small perturbation of the points can be made which will
yield a generic point set. Thus in the sense of the standard measure
in the configuration space $\reel^{m \times \size{\pts}}$, almost all point
sets will yield a Delaunay triangulation. However, when the metric is
no longer Euclidean, this is no longer true. 

In contrast to the purely Euclidean case, topological problems arise
in point sets that are ``near degenerate'' , i.e., point sets that are
not $\delta$-generic for a sufficiently large $\delta$. How large
$\delta$ needs to be depends on how much the metric differs from the
Euclidean one. Indeed, this was the initial motivation for the
introduction of $\delta$-generic point sets
\cite{boissonnat2012stab1}, which are central to the results
presented in this paper.

As we show here with a qualitative argument, the problem can be viewed
as arising from the fact that when $m$ is greater than two, the
intersection of two metric spheres is not uniquely specified by $m$
points. We demonstrate the issue in the context of Delaunay balls. The
problem is developed quantitatively in terms of the Voronoi diagram in
\Appref{sec:counter.ex}.  

We work exclusively on a three dimensional domain, and we are not
concerned with ``boundary conditions''; we are looking at a coordinate
patch on a densely sampled compact $3$-manifold.

One core ingredient in Delaunay's triangulation
result~\cite{delaunay1934} is that any triangle $\splxt$ is the face
of exactly two tetrahedra. This follows from the observation that a
triangle has a unique circumcircle, and that any circumscribing sphere
for $\splxt$ must include this circle. The affine hull of $\splxt$
cuts space into two components, and if $\splxt \in \delP$, then it
will have an empty circumsphere centred at a point $c$ on the line
through the circumcentre and orthogonal to $\affhull{\splxt}$.  The
point $c$ is contained on an interval on this line which contains all
the empty spheres for $\splxt$. The endpoints of the interval are the
circumcentres of the two tetrahedra that share $\splxt$ as a face.

The argument hinges on the assumption that the points are in general
position, and the uniqueness of the circumcircle for $\splxt$. If
there were a fourth vertex lying on that circumcircle, then there
would be three tetrahedra that have $\splxt$ as a face, but this
configuration would violate the assumption of general position.

Now if we allow the metric to deviate from the Euclidean one, no
matter how slightly, the guarantee of a well defined unique
circumcircle for $\splxt$ is lost. In particular, If three spheres $S_1$,
 $S_2$ and $S_3$ all circumscribe $\splxt$, their pairwise intersections
 will be different in general. I.e.,
 \begin{equation*}
   S_1 \cap S_3 \neq S_2 \cap S_3.
 \end{equation*}
Although these intersections may be topological circles that are
``arbitrarily close'' assuming the deviation of the metric from the
Euclidean one is small enough, ``arbitrarily close'' is not good
enough when the only genericity assumption allows configurations that
are arbitrarily bad. 

\begin{figure}
  \begin{center}
    \includegraphics[width=.8\columnwidth]{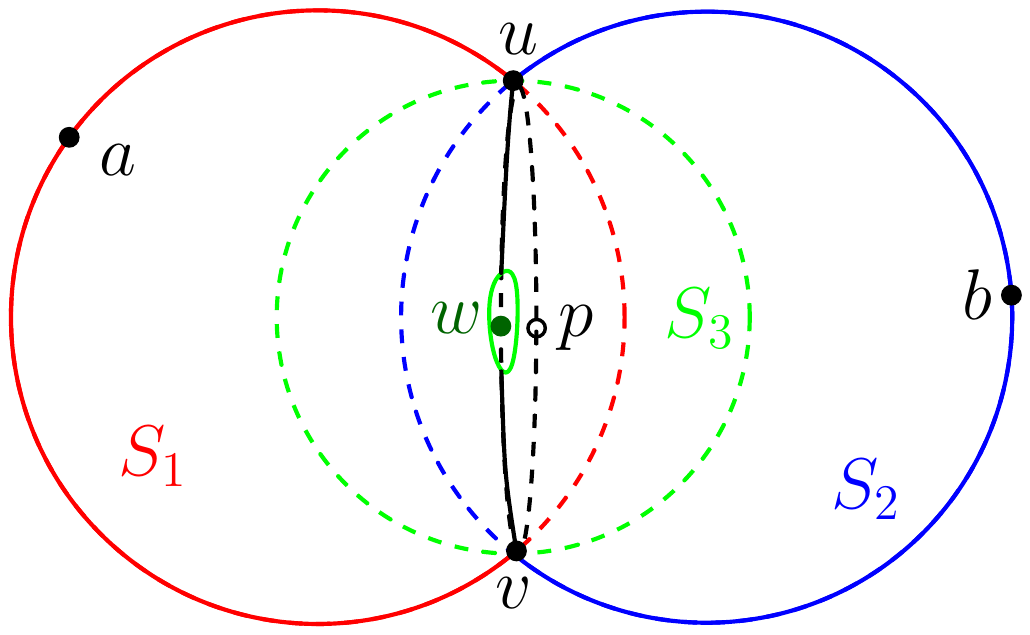} 
  \end{center}
  \caption{In three dimensions, three closed geodesic balls can all
    touch three points, $u,v,p$, on their boundary and yet no one of
    them is contained in the union of the other two.}
  \label{fig:three.tets}
\end{figure}
An attempt to illustrate the problem is given in
\Figref{fig:three.tets}, where $\splxt = \asimplex{u,v,p}$. Here, the
sphere $S_3$ would be contained inside the spheres $S_1$ and $S_2$ if
the metric were Euclidean, but any aberration in the metric may leave
a part of $S_3$ exposed to the outside. This means that in principle
another sample point $w$ could lie on %or just inside 
$S_3$, while
$S_1$ and $S_2$ remain empty. Thus there are three tetrahedra that
share $\splxt$ as a face. 

The essential difference between dimension 2 and the higher dimensions
can be observed by examining the topological intersection properties
of spheres. Specifically, two $(m-1)$-spheres intersect transversely
in an $(m-2)$-sphere. For a non-Euclidean metric, even if this
property holds for sufficently small geodesic spheres, only in
dimension two is the sphere of intersection of the Delaunay spheres of
two adjacent $m$-simplices uniquely determined by the vertices of the
shared $(m-1)$-simplex. See \Figref{fig:three.tets}.

%%%%%%%%%%%%%%%%%%%

\subsubsection{The Voronoi diagram}

We will occasionally make reference to the Voronoi diagram, which is a
structure dual to the Delaunay complex. It offers an alternative way
to interpret observations made with respect to the Delaunay complex.

The \defn{Voronoi cell} associated with $p \in \pts$ with respect to
the metric $\gdist: U \times U \to \reel$ is given by
\begin{equation*}
  \vorcelld{p} = \{x \in U \,|\, \dist{x}{p} \leq \dist{x}{q}
  \text{ for all } q \in \pts \}. 
\end{equation*}
More generally, a \defn{Voronoi face} is the intersection of a set of
Voronoi cells: given $\{p_0,\ldots, p_k\} \subset \pts$, let
$\splxs$ denote the corresponding abstract simplex. We define the associated
Voronoi face as
\begin{equation*}
  \vorcelld{\splxs} = \bigcap_{i=0}^{k} \vorcelld{p_i}.
\end{equation*}
It follows that $\splxs$ is a Delaunay simplex if and only if
$\vorcelld{\splxs} \neq \emptyset$. In this case, every point in
$\vorcelld{\splxs}$ is the centre of a Delaunay ball for $\splxs$.
Thus every Voronoi face corresponds to a Delaunay simplex. The Voronoi
cells give a decomposition of $U$, denoted $\vorPd$, called the
\defn{Voronoi diagram}.  Our definition of the Delaunay complex of
$\pts$ corresponds to the nerve of the Voronoi diagram.

% -*- LaTeX -*-
% equating_intro.tex
% equating Del structs section intro for Stab 2
% 20121206
%

\section{Equating Delaunay structures}
\label{sec:equating}

We now turn to the task of triangulating $\man$, a smooth, compact
$m$-manifold, without boundaries embedded in $\amb$. In this section
we demonstrate our main structural result, \Thmref{thm:tan.cplx},
which is stated at the end of \Secref{sec:del.structs}.  It says that
the complex constructed by the algorithm we describe in
\Secref{sec:algorithms} is in fact an intrinsic Delaunay
triangulation of the manifold, which we introduce next.

\subsection{Delaunay structures on manifolds}
\label{sec:del.structs}

The \defn{restricted Delaunay complex} is the Delaunay complex
$\delRMmpts$ obtained when distances on the manifold are measured with
the metric $\gdistRM$. This is the Euclidean metric of the ambient
space, restricted to the submanifold $\man$. In other words,
$\distRM{x}{y} = \distamb{x}{y}$. We use this notation to avoid
ambiguities in conjunction with the local Euclidean metrics discussed
below.  The Delaunay complex $\delRMmpts$ is a substructure of
$\del{\amb}{\mpts}$.

Alternatively, distances on the manifold may be measured with
$\gdistM$, the \defn{intrinsic metric} of the manifold. This metric
defines the distance between $x$ and $y$ as the infimum of the lengths
of the paths on $\man$ which connect $x$ and $y$. Since the length of
a path on $\man$ is defined as its length as a curve in $\amb$, this
metric is also induced from $\gdistamb$. The \defn{intrinsic Delaunay
  complex} is the Delaunay structure $\delMmpts$ associated with this
metric.

Although neither of these metrics are Euclidean, the idea is that
locally, in a small neighbourhood of any point, these metrics may be
well approximated by $\gdistEm$. Then, if the sampling satisfies
appropriate $\delta$-generic and $\samconst$-dense criteria in these
local Euclidean metrics, the global Delaunay complex in the metric of
the manifold will coincide locally with a Euclidean Delaunay
triangulation, and we can thus guarantee a manifold complex.

\subsubsection{Local Euclidean metrics}
\label{sec:local.Euclidean}

A \defn{local coordinate chart} at a point $p \in \man$, is a pair
$(W,\phi_p)$, where $W \subset \man$ is an open neighbourhood of $p$,
and $\phi_p: W \to U = \phi_p(W) \subset \rem$ is a homeomorphism onto
its image, with $\phi_p(p) = 0$. A local coordinate chart allows us to
\defn{pull back} the Euclidean metric to $W$. For all $x,y \in W$, the
metric $\distG{{\phi_p}}{x}{y} = \distEm{\phi_p(x)}{\phi_p(y)}$ is a
a \defn{local Euclidean metric} for $p$
on $W$. This metric depends upon the choice of $\phi_p$; there
are different ways to impose a Euclidean metric on $W$.

It is convenient to take the reciprocal point of view, and
work with a \defn{local parameterization} at a point $p \in
\man$. This is a pair $(U,\psi_p)$, such that $U \subset \rem$, and
$(W, \psi_p^{-1})$ is a local coordinate chart for $p$, where $W =
\psi_p(U)$.  We can then use $\psi_p$ to pull back the metric of the
manifold to $U$, and to simplify the notation we write $\distM{x}{y}$
for $x,y \in U$, where it is to be understood that this means
$\distM{\psi_p(x)}{\psi_p(y)}$, and likewise for $\distRM{x}{y}$.
Indeed, once $W$ and $U$ have been coupled together by a
homeomorphism, we can transfer the metrics between them and the
distinction becomes only one of perspective; the standard
metric $\gdistEm$ on $U$ is a local Euclidean metric for $p$.

% In this work the only local parameterization we consider is defined by
% the projection map $\pi_p$ described in \Secref{sec:proj.map}.

We wish to generate a sample set $\mpts \subset \man$ that will allow
us to exploit the stability results for Delaunay
triangulations~\cite{boissonnat2012stab1}. We consider the stability
of a Delaunay triangulation in a local Euclidean metric. The following
definition is convenient when stating the stability results:
%\Secref{sec:stability}. 
%We recall the following convenient definition:
\begin{de}[Secure simplex]
  \label{def:secure}
  A simplex $\splxs \in \delP$ is \defn{secure} if it is a
  $\delta$-protected $m$-simplex that is $\thickbnd$-thick and
  satisfies $\circrad{\splxs} < \samconst$ and $\shortedge{\splxs} \geq
  \sparseconst \samconst$.
\end{de}
We will make reference to the following result~\cite[Theorem
4.17]{boissonnat2012stab1}:
\begin{thm}[Metric stability assuming thickness]
  \label{thm:thick.metric.stability}
  Suppose $\convhull{\pts} \subseteq U \subset \rem$ and the metric
  $\gdist: U \times U \to \reel$ is such that $\abs{\dist{x}{y} -
    \distEm{x}{y}} \leq \pertconst$ for all $x,y \in U$. Suppose also
  that $\pts_J \subseteq \pts$ is such that every $m$-simplex $\splxs
  \in \str{\pts_J; \delP}$ is secure and satisfies
  $\distEm{p}{\bdry{U}} \geq 2\samconst$ for every vertex $p \in
  \splxs$. If
  \begin{equation*}
    \pertconst \leq \frac{\thickbnd\sparseconst}{36}\delta,
  \end{equation*}
  then
  \begin{equation*}
    \str{\pts_J;\delPd} = \str{\pts_J;\delP}.
  \end{equation*}
\end{thm}
In our context the point set $\pts$ used in
\Thmref{thm:thick.metric.stability} will come from a larger point set
$\mpts$, such that $\pts = W \cap \mpts$. We will write $\pts_W$ in
order to emphasise this dependence on $W$. We want to ensure that
\begin{equation}
  \label{eq:star.in.subset}
  \str{\pts_J;\del{\gdist}{\pts_W}} =
  \str{\pts_J;\del{\gdist}{\mpts}}. 
\end{equation}
This requirement is attained by demanding that $\mpts$ satisfy a
sampling radius of $\samconst$ with respect to the metric
$\gdistM$. Since $\distEm{x}{y} \leq \distM{x}{y}$ for all $x,y \in U
\cong W$, by our particular choice of $\psi_p$, we will have that
$\pts_W$ is an $\samconst$-sample set with respect to the metric
$\gdistEm$.  We ensure that $U$ is large enough so that
$\distEm{p}{\bdry{U}} \geq 4\samconst$ for all $p \in \pts_J$. It then
follows that $\circrad{\splxs} < \samconst$ for any simplex $\splxs
\in \str{\pts_J;\delof{\pts_W}}$, because $\pts_W$ is an
$\samconst$-sample set \cite[Lemma 3.6]{boissonnat2012stab1}, and
thus $\distEm{q}{\bdry{U}} \geq 2\samconst$ for any $q \in \splxs$. It
follows that $\distM{q}{\bdry{U}} \geq 2\samconst$ as well, and thus
the sampling radius on $\mpts$ ensures that \Eqnref{eq:star.in.subset}
is satisfied. For our purposes $\pts_J$ will consist of a single point
$p$, and the sampling radius $\samconst$ is constrained by the
requirement that $U$ be small enough that the metric distortion
introduced by $\psi_p$ meets the requirements of
\Thmref{thm:thick.metric.stability}.

\subsubsection{The tangential Delaunay complex}

The algorithm we describe in \Secref{sec:algorithms} is a variation of
the algorithm described by Boissonnat and
Ghosh~\cite{boissonnat2010meshing}. This algorithm builds the
\defn{tangential Delaunay complex}, which we denote by
$\tancplxmpts$. This is not a Delaunay complex as we have defined
them, since it cannot be defined by the Delaunay empty ball criteria
with respect to any single metric. However, it is a Delaunay-type
structure, and as with $\delRMmpts$, the tangential Delaunay complex
is a substructure of $\del{\amb}{\mpts}$.  We will demonstrate
sampling conditions which ensure that $\tancplxmpts = \delMmpts =
\delRMmpts$.
\begin{de}[Tangential Delaunay complex]
  The \defn{tangential Delaunay complex} for $\mpts \subset \man
  \subset \amb$ is defined by the criterion that $\splxs \in
  \tancplxmpts$ if it has an empty circumscribing ball
  $\ballamb{c}{r}$ such that $c \in \tanspace{p}{\man}$ for some
  vertex $p \in \splxs$.
\end{de}

We define some local complexes to facilitate discussions of the
tangential Delaunay complex.  For all $p \in \mpts$, let
\begin{equation*}
  K(p) = \{ \sigma \; | \;  \vorcellamb{\sigma} \cap
  \tanspace{p}{\man} \neq \emptyset \},
\end{equation*}
and define
\begin{equation}
  \label{eq:define.star}
  \str{p} = \str{p; K(p)}.
\end{equation}
Then the tangential Delaunay complex is the union of the complexes
$\str{p}$ for all $p \in \mpts$.

Boissonnat et al.~\cite[Lemma~2.3]{boissonnat2011tancplx} showed that
$\vorcellamb{\mpts} \cap T_{p}\man$ is equal to the $m$-dimensional
weighted Voronoi diagram of $\mpts' \subset T_{p}\man$, where $\mpts'$
is the orthogonal projection of $\mpts$ onto $T_{p}\man$ and the
squared weight of a point $p_{i}' \in \mpts'$ is $-\|p_{i}
-p_{i}'\|^{2}$. Therefore, $K(p)$ is isomorphic to a dual complex
(the nerve) of the $k$-dimensional weighted Voronoi diagram of
$\mpts'$.

\subsubsection{Power protection}

The algorithm introduced in \Secref{ssec-algorithm-manifold-case} will
ensure that for every simplex $\splxs$ in the tangential Delaunay
complex, and every vertex $p \in \splxs$, there is a Delaunay ball for
$\splxs$ that is centred on $\tanspace{p}{\man}$ and is protected in
the following sense:
\begin{de}[Power protection]
  A simplex $\splxs$ with Delaunay ball $\ballamb{C}{R}$ is
  \defn{$\pdelta^2$-power-protected} if $\distamb{C}{q}^2 - R^2 >
  \pdelta^2$ for all $q \in \mpts \setminus \splxs$.
\end{de}
Observe that, if $C \not \in \man$, the ball $\ballamb{C}{R}$ is not
an object that can be described by the metric $\gdistRM$.  In the
context of the tangential Delaunay complex we use power-protection
rather than the protection described in \Secref{sec:protection}
because working with squared distances is convenient when we consider
the Delaunay complex restricted to an affine subspace.

\subsubsection{Main structural result}

The rest of \Secref{sec:equating} is devoted to the proof of
\Thmref{thm:tan.cplx} below. It says that for the point set generated
by our algorithm, the tangential Delaunay complex is isomorphic with
the intrinsic Delaunay complex of $\man$. It then follows, from a
previously published result \cite[Theorem 5.1]{boissonnat2010meshing},
that the intrinsic Delaunay complex is in fact homeormorphic to
$\man$; it is an intrinsic Delaunay triangulation.

Thus we obtain a partial recovery of the kind of results attempted by
Leibon and Letscher~\cite{leibon2000}. Our sampling conditions, and
our algorithm (existence proof) rely on the embedding of $\man$ in
$\amb$; we leave purely intrinsic sampling conditions for future work.

\begin{thm}[Intrinsic Delaunay triangulation]
  \label{thm:tan.cplx}
  Suppose $\mpts \subset \man$ is $(\tsparseconst \samconst)$-sparse with
  respect to $\gdistamb$, and every $m$-simplex $\tsplxs \in
  \tancplxmpts$ is $\tthickbnd$-thick, and has, for every vertex $p
  \in \tsplxs$, a $\pdelta^2$-power-protected empty ball of radius
  less than $\samconst$ centred on $\tanspace{p}{\man}$, with $\pdelta
  \geq \pprotconst \tsparseconst \samconst$. If
  $\pprotconst^2\tsparseconst^2 \leq \frac{1}{7}$, and
  \begin{equation*}
%    \label{eq:sam.thm.tan.cplx}
    \samconst \leq \frac{\tthickbnd^2 \tsparseconst^3 \pprotconst^2
      \reach}{1.5 \times 10^6},
  \end{equation*}
  then
  \begin{equation*}
    \tancplxmpts = \delRMmpts = \delMmpts,
  \end{equation*}
  and for $\samconst$ sufficiently small, these will be homeomorphic
  to $\man$:
  \begin{equation*}
    \carrier{\delMmpts} \cong \man.
  \end{equation*}
\end{thm}

    % -*- LaTeX -*-
% equating_local.tex
% equating Del structs section: choice of local Euclidean metric
% subsection, for Stab 2
%
% 20121206
%

\subsection{Choice of local Euclidean metric}
\label{sec:proj.map}

A local parameterization at $p \in \man$ will be constructed with the
aid of the orthogonal projection
\begin{equation}
  \label{eq:defn.pip}
  \pi_p: \amb \rightarrow \tanspace{p}{\man},
\end{equation}
restricted to $\man$. As shown in \Lemref{lem:tanspace.proj.diffeo},
Niyogi et al.~\cite[Lemma 5.4]{niyogi2008} demonstrated that if $r <
\frac{\reach}{2}$, then $\pi_p$ is a diffeomorphism from $W
=\ballRM{p}{r}$ onto its image $U \subset \tanspace{p}{\man}$. We will
identify $\tanspace{p}{\man}$ with $\rem$, and define the
homeomorphism
\begin{equation}
  \label{eq:defn.psip}
  \psi_p =\pi_p |_{W}^{-1}: U \longrightarrow W.  
\end{equation}
Using $\psi_p$ to pull back the metrics $\gdistM$ and $\gdistRM$ to
$\rem$, we can view them as perturbations of $\gdistEm$. The magnitude
of the perturbation is governed by the radius of the ball used to
define $W$. 
\begin{de}
  \label{def:admissible}
  We call a neighbourhood $W$ of $p \in \man$ \defn{admissible} if $W
  \subseteq \ballRM{p}{r}$, with $r \leq \frac{\reach}{100}$.
\end{de}
In all that follows, any mention of a local Euclidean metric refers to
the one defined by $\pi_p$ restricted to an admissible neighbourhood.
The requirement $r\leq \frac{\reach}{100}$ is simply a convenient
bound that yields a small integer constant in the perturbation bound
of the following lemma, and does not constrain subsequent results. The
bound could be relaxed to $r \leq \frac{\reach}{4}$ at the expense of
a weaker bound on the perturbation.
\begin{lem}[Metric distortion]
  \label{lem:raw.pert.bound}
  Suppose $(U, \psi_p)$ is a local parameterisation at $p \in W
  \subset \man$ with $W = \psi_p(U)$.  If $W \subseteq \ballRM{p}{r}$,
  with $r \leq \frac{\reach}{100}$, then for all $x,y \in U$,
  \begin{equation*}
    \abs{ \distRM{x}{y} - \distEm{x}{y} } \leq    \abs{ \distM{x}{y} -
      \distEm{x}{y} } \leq \frac{23 r^2}{\reach}. 
  \end{equation*}
\end{lem}
\begin{proof}
  Let $u,v \in W \subset \ballRM{p}{r}$, and
  let $\theta$ be the angle between the line segments $\seg{u}{v}$ and
  $\seg{\pi_p(u)}{\pi_p(v)}$, $\theta_{1}$ the angle between
  $\seg{u}{v}$ and $\tanspace{u}{\man}$, and $\theta_{2}$ the angle
  between $\tanspace{p}{\man}$ and $\tanspace{u}{\man}$.  Thus
  $\theta \leq \theta_{1} + \theta_{2}$,
 and 
 % \begin{equation*}
  $   \distEm{\pi_{p}(u)}{\pi_{p}(v)} = \distamb{u}{v} {\cos \theta}$.
%  \end{equation*}
  Defining $\localconst = \frac{r}{\reach}$,
  \Lemref{lem:geodesic.distance} yields
  \begin{equation}
    \label{eq:bnd.geod}
    \distM{u}{v} \leq  \distamb{u}{v} \left( 1+ 4 \localconst \right),
  \end{equation}
  and so
  \begin{equation*}
    \distEm{\pi_{p}(u)}{\pi_{p}(v)} \geq  \frac{\distM{u}{v} \, \cos
      \theta}{1+ 4\localconst}.
  \end{equation*}

  Using \Lemref{lem:dist.to.tanspace}, we find $\sin \theta_{1} \leq
  \localconst$, and \Lemref{lem:tangent.variation}, yields $\sin
  \theta_{2} \leq 6\localconst$.  Therefore, since
 $\sin \theta \leq \sin\theta_1 + \sin\theta_2$, we have $\cos \theta =
  (1 - \sin^2\theta)^{1/2} \geq 1 - \sin\theta \geq 1 - 7\localconst$
  and we get
 \begin{equation*}
    \begin{split}
      \distEm{\pi_{p}(u)}{\pi_{p}(v)}
     &\geq  {\distM{u}{v}} \left( \frac{1- 7\localconst}{1+
          4\localconst} \right) \\  
      &\geq   {\distM{u}{v}} \left( 1- 7\localconst \right) (1-4\localconst)\\
      &\geq  {\distM{u}{v}} (1-11\localconst).
   \end{split}
  \end{equation*}
  Using \Eqnref{eq:bnd.geod} we find $\distM{u}{v} \leq
  \frac{208r}{100}$, so $\distEm{\pi_p(u)}{\pi_p(v)} \geq \distM{u}{v}
  - 23 \frac{r^2}{\reach}$, and the result follows since $\distM{u}{v}
  \geq \distRM{u}{v} \geq \distEm{\pi_{p}(u)}{\pi_{p}(v)}$.
\end{proof}

Our sampling radius is constrained by the size of a Euclidean ball that
can be contained in an admissible neighbourhood. The following lemma
gives a convenient expression for this:
\begin{lem}
  \label{lem:sample.bound}
  If $1 < a \leq 10^4$ and $a\samconst \leq \frac{\reach}{100}$,
  and $U = \ballEm{p}{(a-1)\samconst}$, then $\psi_p(U) = W \subseteq
  \ballRM{p}{a\samconst}$.
% , and
%   \begin{equation*}
%     \abs{ \distRM{x}{y} - \distEm{x}{y} } \leq
%     \abs{ \distM{x}{y} - \distEm{x}{y} } \leq \frac{23a^2\samconst^2}{\reach}.
%   \end{equation*}
\end{lem}
\begin{proof}
  % Since $a\samconst \leq \frac{\reach}{100}$, the perturbation bound
  % follows from \Lemref{lem:raw.pert.bound}, and the fact that
  % \begin{equation*}
  %   \distEm{x}{y} \leq \distRM{x}{y} \leq \distM{x}{y},
  % \end{equation*}
  % for any $x,y \in U = \pi_p(W)$. It remains to show that
  % $\ballEm{p}{(a-1)\samconst} \subset
  % \pi_p(\ballRM{p}{a\samconst})$. 
  Using \Lemref{lem:dist.to.tanspace}, we have that $\ballEm{p}{r}
  \subseteq \pi_p(\ballRM{p}{a\samconst})$ if
  \begin{equation*}
    \begin{split}
      r^2 \leq a^2\samconst^2 - \left( \frac{a^2\samconst^2}{2\reach}
      \right)^2
      &= a^2\samconst^2 \left( 1 - \left( \frac{a\samconst}{2\reach}
        \right)^2 \right)\\
      &\leq a^2\samconst^2 \left( 1 - \left( \frac{1}{200}
        \right)^2 \right).
    \end{split}
  \end{equation*}
  Thus we require
%  \begin{equation*}
$
    r \leq \sqrt{\frac{200^2-1}{200^2}}a\samconst,
$
%  \end{equation*}
  which is satisfied by $r = (a-1)\samconst$ if $a \leq 79999$.
\end{proof}

Lemmas \ref{lem:raw.pert.bound} and \ref{lem:sample.bound} lead to a
sampling radius which allows us to employ
\Thmref{thm:thick.metric.stability}, and so obtain an equivalence
between Delaunay structures:
\begin{prop}[Equating Delaunay complexes]
  \label{prop:equating.del.cplxs}
  Suppose $\mpts \subset \man$ is an $\samconst$-sample set with
  respect to $\gdistRM$, and that for every $p \in \mpts$, in the
  local Euclidean metric on $W = \ballRM{p}{5\samconst}$, every
  $m$-simplex in $\str{p;\delof{\pts_W}}$ is secure,
  where $\pts_W = \mpts \cap W$, and $\delta = \protconst \samconst$. If
  \begin{equation*}
    \samconst \leq \frac{\thickbnd \sparseconst \protconst
      \reach}{20700}
  \end{equation*}
  then
  \begin{equation}
    \label{eq:equate.local.cplxs}
    \str{p;\delof{\pts_W}} = \str{p;\delRM{\pts_W}} = \str{p;\delM{\pts_W}}.
  \end{equation}
  Thus
  \begin{equation*}
    \delRMmpts = \delMmpts,
  \end{equation*}
  and they are manifold complexes.
\end{prop}
\begin{proof}
  As usual, let $U = \pi_p(W)$. Then by \Lemref{lem:sample.bound}
  $\ballEm{p}{4\samconst} \subseteq U$, and thus $\distEm{q}{\bdry{U}}
  \geq 2\samconst$ for any vertex $q$ of a simplex in $\str{p;\delof{\pts_W}}$.
  Thus \Lemref{lem:raw.pert.bound} allows us to apply
  \Thmref{thm:thick.metric.stability} provided
  \begin{equation*}
    \frac{23a^2\samconst^2}{\reach} \leq  \frac{\thickbnd \sparseconst
      \protconst \samconst}{36},
  \end{equation*}
  when $a=5$,
  and we obtain the required bound on $\samconst$. Thus the star
  of every vertex in $\delMmpts$ is equal to the star of that point in
  the local Euclidean metric, and likewise for $\delRMmpts$.
  The claim follows since $\splxs \in \delMmpts$ if and only if it is
  in the local Euclidean Delaunay triangulation of every one of its
  vertices, and likewise for the simplices in $\delRMmpts$.
\end{proof}

    % -*- LaTeX -*-
% equating_protan.tex
% equating Del structs section: Protected Tan complex subsection
% for the stability II paper
% 20121206 <-- 20120426 <-- 20111101
%

\subsection{The protected tangential complex}

We obtain \Thmref{thm:tan.cplx} by means of
\Thmref{thm:thick.metric.stability} via the observation that power
protection of the ambient Delaunay balls translates into protection in
the local Euclidean metrics. We must distinguish between the geometry
of a simplex defined with respect to the Euclidean metric $\gdistamb$
of the ambient space, as opposed to a local Euclidean metric
$\gdistEm$. In general, we use a tilda to indicate simplices in the
ambient space, and their properties.
\begin{lem}[Protection under projection]
  \label{lem:pprotection.protects}
  Suppose $\mpts \subset \man$ and that $\tsplxs \in
  \del{\amb}{\mpts}$ is an $\tthickbnd$-thick $m$-simplex, with
  $\shortedge{\tsplxs} \geq \tsparseconst \samconst$ and
  $\ballamb{C}{R}$ is a $\pdelta^2$-power-protected empty ball for
  $\tsplxs$, with respect to the metric $\gdistamb$, where $\pdelta^2
  \geq \pprotconst^2 \tsparseconst^2 \samconst^2$. Suppose also that
  $C \in \tanspace{p}{\man}$, for some vertex $p \in \tsplxs$.

  If $R < \samconst$, with
  \begin{equation}
    \label{eq:raw.tancplx.sample.bound}
    \samconst \leq \frac{\tthickbnd^2  \tsparseconst^3 \pprotconst^2
      \reach}{512},
  \end{equation}
  then $\splxs = \pi_p(\tsplxs)$ has a $\delta$-protected Delaunay
  ball $\ballEm{c}{r}$ with respect to the local Euclidean metric
  $\gdistEm$ for $p$ on any admissible neighbourhood $W$ that contains
  $\ballRM{p}{3\samconst}$, and $\delta = \protconst \samconst$, with
  \begin{equation}
    \label{eq:protconst.from.tancplx}
    \protconst = \frac{\pprotconst^2\tsparseconst^2}{8}.
  \end{equation}
  % and
  % \begin{equation*}
  %   r \leq R + \frac{24\samconst^2}{\tthickbnd^2 \tsparseconst \reach}.
  % \end{equation*}
\end{lem}
\begin{proof}
  We first find a bound for $\distEm{C}{c}$ and $r$.  Let $\tsplxs =
  \simplex{\tilde{p}_0, \ldots, \tilde{p}_m}$, and $\splxs =
  \simplex{p_0, \ldots, p_m}$ so that $\pi_p(\tilde{p}_i) = p_i$, and
  $p = p_0 = \tilde{p}_0$.  We will first show that, near $C$, there
  is a circumcentre $c$ for $\splxs$ in the metric $\gdistEm$. For any
  $p_i \in \splxs$, $\distamb{p}{p_i} < 2R$, and so by
  \Lemref{lem:dist.to.tanspace} we have
  % the bound on the distance to the tangent space \cite[Lemma
  % 5.1]{boissonnat2012stab1} we have
  \begin{equation*}
    \distamb{\tilde{p}_i}{p_i} \leq \frac{2R^2}{\reach} <
    \frac{2\samconst^2}{\reach}. 
  \end{equation*}
  In order to apply \Lemref{lem:pert.circ.ball} we require
  $\frac{2\samconst^2}{\reach} \leq
  \frac{\tthickbnd^2\tsparseconst\samconst}{28}$, or
  \begin{equation*}
%    \label{eq:required.bnd.for.cc}
    \samconst \leq \frac{\tthickbnd^2\tsparseconst\reach}{56},
  \end{equation*}
  which is satisfied by \Eqnref{eq:raw.tancplx.sample.bound}.
  Since $\affhull{\splxs} = \tanspace{p}{\man}$, the circumcentre $c
  \in \tanspace{p}{\man}$ is the closest point in $\normhull{\splxs}$
  to $C$, \Lemref{lem:pert.circ.ball} yields
  \begin{equation*}
    \abs{R-r} \leq  \distEm{C}{c} = \distamb{C}{c} <
    \frac{16\samconst^2}{\tthickbnd\tsparseconst \reach}.
  \end{equation*}

  Now we seek a lower bound on the protection of $\ballEm{c}{r}$.
  Suppose $\tilde{q} \in \mpts \setminus \tsplxs$. We wish to
  establish a lower bound on $\distEm{c}{q} - r$, where $q =
  \pi_p(\tilde{q})$. We may assume that $\distamb{C}{\tilde{q}} <
  3\samconst$, since otherwise $q$ will lie outside of our region of
  interest.

  Let $z = \frac{(3\samconst)^2}{2\reach}$ be the upper bound on
  $\distamb{\tilde{q}}{q}$ given by
  \Lemref{lem:dist.to.tanspace}.
  % the bound on the distance to the tangent space \cite[Lemma
  % 5.1]{boissonnat2012stab1}.
  Then $\distEm{C}{q}^2 \geq
  \distamb{C}{\tilde{q}}^2 - z^2 > R^2 + \pdelta^2 - z^2$. Thus
  \begin{equation*}
    \distEm{C}{q} - R > \frac{\pdelta^2 - z^2}{\distEm{C}{q} + R}
    > \frac{\pdelta^2 - z^2}{4\samconst},
  \end{equation*}
  since $R < \samconst$.  Then $\distEm{c}{q} - r \geq (\distEm{C}{q}
  - \distEm{C}{c}) - (R + \abs{R-r}) > \frac{\pdelta^2 -
    z^2}{4\samconst} - 2\distEm{C}{c}$. Putting this together, using
  $\pdelta^2 \geq \pprotconst^2 \tsparseconst^2 \samconst^2$, we get
  \begin{equation*}
    \distEm{c}{q} - r >
    \left( \frac{1}{4}  \pprotconst^2 \tsparseconst^2 - \frac{81
        \samconst^2}{16 \reach^2} - \frac{32\samconst}{\tthickbnd
        \tsparseconst \reach} \right) \samconst.
  \end{equation*}
  In order to simplify away the final term, we demand
  \begin{equation*}
    \frac{32 \samconst}{\tthickbnd \tsparseconst \reach} \leq
    \frac{1}{16} \pprotconst^2 \tsparseconst^2,
  \end{equation*}
  which is satisfied by \Eqnref{eq:raw.tancplx.sample.bound}.  Under
  this constraint, the central term is also seen to be less than
  $\frac{1}{16} \pprotconst^2 \tsparseconst^2$, and we obtain
  \begin{equation*}
    \delta \geq \frac{1}{8} \pprotconst^2 \tsparseconst^2 \samconst.
  \end{equation*}
\end{proof}

\Propref{prop:equating.del.cplxs} requires a thickness $\thickbnd$ and
shortest edge bound $\sparseconst \samconst$ for the simplex $\splxs
\subset \rem$, but \Lemref{lem:pprotection.protects} is expressed in terms
of the corresponding quantities $\tthickbnd$ and $\tsparseconst
\samconst$ for the corresponding simplex $\tsplxs \subset \amb$.
\begin{lem}[Simplex distortion under projection]
  \label{lem:translate.thick.sparse}
  Let $\tsplxs \in \del{\amb}{\mpts}$ be an $m$-simplex as described
  in \Lemref{lem:pprotection.protects}, and let $\splxs =
  \pi_p(\tsplxs)$ be its projection in the local Euclidean metric for
  $p$ on any admissible neighbourhood that contains
  $\ballRM{p}{2\samconst}$, where $p$ is a vertex of $\tsplxs$. If
  $\samconst$ satisfies \Eqnref{eq:raw.tancplx.sample.bound}, and
  $\pprotconst^2\tsparseconst^2 \leq \frac{1}{7}$, then
  $\shortedge{\splxs} > \sparseconst \samconst$, where
  \begin{equation*}
    \sparseconst = \frac{20}{21} \tsparseconst,
    % \sparseconst = \left( 1 - \frac{ \tthickbnd^2\pprotconst^2
    %     \tsparseconst^2 }{3} \right) \tsparseconst,
  \end{equation*}
  and $\thickness{\splxs} > \thickbnd$, where
  \begin{equation*}
    \thickbnd = \frac{6}{49} \tthickbnd.    
%    \thickbnd = (1 - 2 \pprotconst^2 \tsparseconst^2 ) \tthickbnd.    
  \end{equation*}
\end{lem}
\begin{proof}
  Since $\pi_p(\ballRM{p}{2\samconst}) \subseteq
  \ballEm{p}{2\samconst}$, it is sufficient to apply
  the Metric distortion lemma~\ref{lem:sample.bound}
%  the metric distortion lemma~\cite[Lemma 5.6]{boissonnat2012stab1}
  with $a = 3$.

  For the shortest edge length, we find
  \begin{equation*}
    \begin{split}
      \shortedge{\splxs}
      &\geq \shortedge{\tsplxs} - \frac{3^2 \times 23 \samconst^2}{\reach}\\
      &= \tsparseconst \left( 1 - \frac{207 \tthickbnd^2\pprotconst^2
      \tsparseconst^2 }{512} \right) \samconst\\
      &> \tsparseconst \left( 1 - \frac{ \tthickbnd^2\pprotconst^2
          \tsparseconst^2 }{3} \right) \samconst\\
      &> \frac{20}{21}\tsparseconst \samconst.
    \end{split}
  \end{equation*}

  For the thickness bound, in order to apply
  \Lemref{lem:pert.thick.bnd}, using $\tilde{\localconst} =
  (1-\localconst)$, we require
  \begin{equation*}
    \frac{207\tthickbnd^2 \pprotconst^2 \tsparseconst^3}{512} \leq
    \frac{\tilde{\localconst} \tthickbnd^2 \tsparseconst}{14},
  \end{equation*}
  which is satisfied if we choose
  \begin{equation*}
    \tilde{\localconst} > 6 \pprotconst^2 \tsparseconst^2.
  \end{equation*}
  Then \Lemref{lem:pert.thick.bnd} yields
  \begin{equation*}
    \thickness{\splxs} \geq \frac{6}{7} \left( 1 - 6
      \pprotconst^2 \tsparseconst^2 \right) \thickness{\tsplxs}
    > \frac{6}{49}\thickness{\tsplxs}.
  \end{equation*}
\end{proof}

We can now express the sampling conditions in terms of the output
parameters of the tangential complex algorithm, and this allows us to
apply \Propref{prop:equating.del.cplxs} and obtain our main structural
result:
\begin{proof}[of \Thmref{thm:tan.cplx}]
  We first translate the sampling requirements of
  \Propref{prop:equating.del.cplxs} in terms of properties of
  simplices in the ambient metric $\gdistamb$.
  Using \Lemref{lem:translate.thick.sparse}, together with
%  \Lemref{lem:pprotection.protects},
  \Eqnref{eq:protconst.from.tancplx},
  the upper bound on the sampling radius demanded by
  \Propref{prop:equating.del.cplxs} becomes
  \begin{equation*}
    \samconst \leq \frac{ 20 \times 6
      \tthickbnd \pprotconst^2 \tsparseconst^3 \reach }{
      21 \times 49 \times 8 \times 20700}.
  \end{equation*}
  We obtain the stated sampling radius bound after multiplying by
  $\tthickbnd$ in order to ensure that the demand of
  \Eqnref{eq:raw.tancplx.sample.bound} is also met. Thus the stated
  sampling radius satisfies the requirements of both
  \Lemref{lem:pprotection.protects} and
  \Propref{prop:equating.del.cplxs}. 

  The fact that the structures are isomorphic follows from the fact
  that they are all locally isomorphic to the Delaunay triangulation
  in the local Euclidean metric.  To see that $\str{p; \tancplxmpts}
  \cong \str{p; \delof{\pts_W}}$, observe that
  \Lemref{lem:pprotection.protects} implies that there is an injective
  simplicial map $\str{p; \tancplxmpts} \to \str{p;
    \delof{\pts_W}}$. The isomorphism is established by
  \Lemref{lem:inject.triang}, once it is established that
  $\str{p;\tancplxmpts}$ is a triangulation at $p$. In fact $\str{p;
    \tancplxmpts}$ is isomorphic to the star of $p$ in a regular
  triangulation of the projected points $\pts_W$; it is a
  \defn{weighted Delaunay triangulation}~\cite{boissonnat2010meshing},
  and with our choice of $W$, the point $p$ is an interior point in
  this triangulation~\cite[Lemma 2.7(1)]{boissonnat2010meshing}.  Thus
  $\str{p; \tancplxmpts}$ is a triangulation at $p$, and it follows
  that
%  by \Lemref{lem:inject.triang} that
%  \cite[Lemma 2.7]{boissonnat2012stab1} that
  \begin{equation*}
    \str{p; \tancplxmpts} \cong
    \str{p; \delof{\pts_W}}.
  \end{equation*}
  The equality of the Delaunay complexes now follows from
  \Propref{prop:equating.del.cplxs}, \Eqnref{eq:equate.local.cplxs}.

  The homeomorphism assertion follows from previous
  work~\cite[Theorem 5.1]{boissonnat2010meshing}. 
\end{proof}

% -*- LaTeX -*-
% algo_prelim.tex
% 20120808
% First part of the Algorithm Section
%
% was the first part of "New section" in algorithms.tex

\section{Algorithm}
\label{sec:algorithms}

In this section we introduce a Delaunay refinement algorithm which,
while constructing a tangential Delaunay complex, will transform the
input sample set into one which meets the requirements of
\Thmref{thm:tan.cplx}.  In particular we wish to construct a
tangential Delaunay complex in which every $m$-simplex $\splxs$ is
$\tthickbnd$-thick and for every $p \in \splxs$, there is a
$\pdelta^2$-power-protected Delaunay ball for $\splxs$ centred on
$\tanspace{p}{\man}$. We demand $\pdelta \geq \pprotconst
\tsparseconst \samconst$, where $\samconst$ provides a strict upper
bound on the radius of these Delaunay balls, and $\tsparseconst
\samconst$ provides a lower bound on the shortest edge length of any
simplex in $\tancplxmpts$. The constants $\pprotconst$ and
$\tsparseconst$ are both positive and smaller than one.

The algorithm is in the same vein as that of Boissonnat and
Ghosh~\cite{boissonnat2010meshing}, which is in turn an adaptation of
the algorithm introduced by Li~\cite{li2003}. It is described in
\Secref{ssec-algorithm-manifold-case}, after we introduce terminology
and constructs which are used in the algorithm in
\Secref{sec:algo.components}.

\subsection{Components of the algorithm}
\label{sec:algo.components}

We now introduce the primary concepts that are used as building blocks
of the algorithm.

\subsubsection{Elementary weight functions}

Elementary weight functions are a convenient device to facilitate the
identification of simplices $\splxs$ that are not
$\pdelta^2$-power-protected for $\pdelta = \pprotconst
\shortedge{\splxs}$.

In order to emphasise that we are considering a function defined only
on the set of vertices of a simplex, we denote by $\mathring{\sigma}$
the set $\{p_{0}, \, \dots,\, p_{k}\}$ of vertices of $\sigma =
\simplex{p_{0} \, \dots \, p_{k} }$.  We will call $\omega_{\sigma}:
\mathring{\sigma} \rightarrow [0, \infty)$ an \defn{elementary weight
  function} if it satisfies the following conditions:
\begin{enumerate}
\item There exists $p_{i} \in \mathring{\sigma}$ such that
  $\omega_{\sigma}(p_{i}) \in [0, \, \delta_{0}L(\sigma) ]$, and
\item for all $p_{j} \in \mathring{\sigma} \setminus p_{i}$,
  $\omega_{\sigma}(p_{j}) = 0$.
\end{enumerate}

For a given $\splxs = [p_{0}, \, \dots, \, p_{k}]$ and elementary
weight function $\gewf{\splxs}$, we define
$\wnormhull{\splxs}{\gewf{\splxs}}$ as the set of solutions to the
following system of $k$ equations:
\begin{equation*}
  \| x - p_{i} \|^{2} - \| x - p_{0} \|^{2} = \ewf{\splxs}{p_{i}}^{2}
  - \ewf{\splxs}{p_{0}}^{2}.
\end{equation*}
In direct analogy with the space $\normhull{\splxs}$ of centres of
$\splxs$, the set $\wnormhull{\splxs}{\gewf{\splxs}}$ is an affine
space of dimension $m - \dim \affhull{\splxs}$ that is orthogonal to
$\affhull{\splxs}$.  We denote by $C(\splxs,\gewf{\splxs})$ the unique
point in $\wnormhull{\splxs}{\gewf{\splxs}} \cap \aff(\sigma)$, and we
define
\begin{equation*}
  R(\sigma, \omega_{\sigma})^{2} = \|p_{0} - C(\sigma,
  \omega_{\sigma})\|^{2} -\omega_{\sigma}(p_{0})^{2},
\end{equation*}
% and
% \begin{equation*}
%   \Phi(\sigma, \omega_{\sigma}) = \frac{R(\sigma, \omega_{\sigma})}{L(\sigma)}
% \end{equation*}
where the notation is chosen to emphasise the close relationship with
the circumcentre $\circcentre{\splxs}$ and circumradius
$\circrad{\splxs}$. The following lemma exposes some properties of
$\wcircrad{\splxs}{\gewf{\splxs}}$ in this spirit:
\begin{lem}
  \label{claim-rel-edge-length-ortho-radius}
  For a given $\sigma = [p_{0}, \dots , p_{k} ]$, with $k \geq 1$,
  and elementary weight function $\omega_{\sigma}$, we have:
  \begin{enumerate}
  \item If $\sigma_{1} \leq \sigma$ then $\omega_{\sigma_{1}} =
    \omega_{\sigma} \mid_{\mathring{\sigma}_{1}}$ is an elementary
    weight function, and
    \begin{equation*}
      R(\sigma_{1}, \omega_{\sigma_{1}}) \leq
      R(\sigma, \omega_{\sigma}).
    \end{equation*}
  \item $\Delta (\sigma) \leq \frac{2}{1-\delta^{2}_{0}}\, R(\sigma,
    \omega_{\sigma})$.
  \item If $\Upsilon(\sigma) > 0$, then
    $$
    1 -  \localconst \; \leq \;
    \frac{R(\sigma, \omega_{\sigma})}{R(\sigma)} \;  \leq \; 1 + \localconst,
    $$
    with $ \localconst = \frac{\delta^{2}_{0}}{\Upsilon(\sigma)} $.
  \end{enumerate}
\end{lem}
\begin{proof}
  1. That $\gewf{\sigma_1}$ is an elementary weight function follows
  from the observation that $\shortedge{\sigma_1} \geq
  \shortedge{\sigma}$. Since $\wnormhull{\sigma}{\gewf{\sigma}}
  \subseteq \wnormhull{\sigma_1}{\gewf{\sigma_1}}$, the projection of
  $\wcirccentre{\sigma}{\gewf{\sigma}}$ into $\affhull{\sigma_1}$ is
  $\wcirccentre{\sigma_1}{\gewf{\sigma_1}}$.  The result then follows
  from the Pythagorean theorem.

  2. Let $e=[p_{0} , p_{1}]$ be the longest edge of $\sigma$, and let
  $c$ denote the projection of $C(\sigma, \omega_{\sigma})$ onto
  $\aff(e)$. Without loss of generality we assume that $\omega(p_{0})
  = 0$.
	
  We have
  \begin{equation*}
    \begin{split}
      \norm{p_0 - c}^2 &= \norm{p_1 - c}^2 - \ewf{\sigma}{p_1}^2 \\
      &= \norm{(p_1 - p_0) - (c-p_0)}^2 - \ewf{\sigma}{p_1}^2 \\
      &= \longedge{\sigma}^2 - 2\dotprod{(p_1 - p_0)}{(c-p_0)} +
      \norm{p_0 -c}^2 - \ewf{\sigma}{p_1}^2. \\
    \end{split}
  \end{equation*}
  Since $p_0$, $p_1$, and $c$ are colinear, we have $2\dotprod{(p_1 -
    p_0)}{(c-p_0)} = 2\longedge{\sigma}\norm{p_0 - c}$, and using the
  fact that $\omega_{\sigma}(p_{1}) \leq \delta_{0} L(\sigma)$, we get
  \begin{equation*}
    \begin{split}
      \| p_{0} - c \|
      &= \frac{\Delta(\sigma)}{2} \left( 1
      -\frac{\omega_{\sigma}(p_{1})^{2} }{\Delta(\sigma)^{2}} \right) \\
    &\geq \frac{(1-\delta^{2}_{0}) \Delta(\sigma)}{2}.
    \end{split}
  \end{equation*}
 The result follows from the fact that $R(\sigma, \omega_{\sigma})
  \geq \|p_{0} - c\|$.

  3. Using the fact that $\omega_{\sigma}(p) = 0$ for all vertices $p
  \in \mathring{\sigma}$, except at most one, we get $p_{i}
  \in \partial B_{\reel^{k}}(C(\sigma, \omega_{\sigma}), R(\sigma,
  \omega_{\sigma}))$ for all $p_{i} \in \mathring{\sigma}$ except at
  most one.
	
  Let $\localconst= \| C(\sigma, \omega_{\sigma}) - C(\sigma) \|$, and
  assume, without loss of generality, that the vertex $p_{0}
  \in \partial B(C(\sigma, \omega_{\sigma}), R(\sigma,
  \omega_{\sigma}))$.  Therefore,
  \begin{eqnarray}\label{eqn-1-claim-rel-edge-length-ortho-radius}
    \|C(\sigma) - p_{0} \| - \localconst \; \; \leq  &\|C(\sigma,
    \omega_{\sigma})-p_{0}\|&   \leq  \; \;  \|C(\sigma) - p_{0}\| +
    \localconst \nonumber \\
    R(\sigma)+ \localconst \; \; \leq &R(\sigma, \omega_{\sigma})& \leq\; \;
    R(\sigma) + \localconst.
  \end{eqnarray}

  Since the point in $\normhull{\sigma}$ that is closest to
  $\wcirccentre{\sigma}{\omega_{\sigma}}$ is
  $\circcentre{\sigma}$, and
 $\thickness{\sigma} > 0$,   we obtain the following bound using Lemma 4.1
  from~\cite{boissonnat2012stab1}:
  \begin{eqnarray}\label{eqn-2-claim-rel-edge-length-ortho-radius}
    \localconst &\leq& \frac{\delta_{0}^{2}L(\sigma)^{2} }{2 \Upsilon(\sigma)
      \Delta(\sigma)} \nonumber \\
    &\leq& \frac{\delta^{2}_{0}}{\Upsilon(\sigma)}R(\sigma),
    ~~~~\mbox{since $L(\sigma) \leq \Delta(\sigma) \leq 2R(\sigma)$.}
  \end{eqnarray}
  The result now follows from
  Eq.~\eqref{eqn-1-claim-rel-edge-length-ortho-radius} and
  \eqref{eqn-2-claim-rel-edge-length-ortho-radius}.
\end{proof}

If $\splxs = \splxjoin{p}{\splxsp}$, and $\gewf{\splxs}$ is an
elementary weight function that vanishes on $\mathring{\splxsp}$, then
$\wnormhull{\splxs}{\gewf{\splxs}} \subseteq \normhull{\splxsp}$, but
no point in $\wnormhull{\splxs}{\gewf{\splxs}}$ can be the centre
of a $\pdelta^2$-power-protected Delaunay ball for $\splxsp$ for any
$\pdelta \geq \pprotconst \shortedge{\splxs}$. In other words,
$\splxs$ and $\gewf{\splxs}$ define a quasi-cospherical configuration
that is an obstruction to the power protection of $\splxsp$ at all
points in $\wnormhull{\splxs}{\gewf{\splxs}}$.

\subsubsection{Quasicospherical configurations}
\label{ssec-def-prop-stars-delta-stars}

% Let $\mpts$ be an $\epsilon$-sample of $\man$ with $\epsilon \leq
% \frac{\reach}{16}$, and we define
% \begin{equation*}\label{eqn-def-normalized-parameter}
%   \tilde{\epsilon}  \; \stackrel{{\rm def}}{=}  \; \frac{\epsilon}{\reach}
% \end{equation*}

% Before we can move forward, we need to give some additional
% definitions.

We now define the family of simplices that our algorithm must
eliminate in order to ensure that the final point set has the desired
protection properties. 

Recalling the definition~\eqref{eq:define.star} of $\str{p}$, we have
the following \cite[Lemma~2.7~(1)]{boissonnat2010meshing}:
\begin{lem}
  \label{lem:Rp.small} % redundant label from redundant lemma
  \label{lem-bound-tangent-vorronoi-intrinsic}
  Let $\mpts \subset \man$ satisfy a sampling radius of $\epsilon$
  with respect to $\gdistamb$ such that $\epsilon \leq {\reach} /
  {16}$.  Then for all $x \in \vorcellamb{p} \cap T_{p}\man$,
  we have $\|p-x\| \leq 4 \epsilon$. In particular, for all $ p \in
  \mpts $, and every $m$-simplex $\sigma \in \str{p}$, we have
  $R_{p}(\sigma) \leq 4\epsilon$.
  % Let $\mpts$ be an $\epsilon$-sample of $\man$ with $\epsilon \leq
   %  \frac{\reach}{16}$.  Then for all $x \in \vorcellamb{p} \cap
   %  T_{p}\man$, we have $\|p-x\| \leq 4 \epsilon$.
\end{lem}
%Using~\cite[Lemma~2.3]{boissonnat2011tancplx} and
Since by Lemma~\ref{lem-bound-tangent-vorronoi-intrinsic}, the Voronoi
cell of $p$ restricted to $\tanspace{p}{\man}$ is bounded, we get:
\begin{lem}
  If $\epsilon \leq \frac{\reach}{16}$, then the combinatorial
  dimension
  % \footnote{{\em Combinatorial dimension} of a simplex
  %   $\sigma^{k}$ is $k$ rather than the dimension of the affine space
  %   $\aff(\sigma^{k})$.}
  of the maximal simplices in $\str{p}$ is
  at least $m$.
\end{lem}
We will always assume that $\mpts$ satisfies a sampling radius of
$\samconst \leq \frac{\reach}{16}$.  If $\sigma$ is a maximal simplex
in $\str{p}$, then $\vorcellamb{\sigma}$ intersects
$\tanspace{p}{\man}$ at a single point.  Indeed, since
$\vorcellamb{\sigma} \subset \vorcellamb{p}$, by
Lemma~\ref{lem-bound-tangent-vorronoi-intrinsic} the convex set
$\vorcellamb{\sigma} \cap \tanspace{p}{\man}$ is bounded, and if it
had a nonempty interior, then $\sigma$ would not be maximal.
Let $\sigma$ be a maximal simplex in $\str{p}$. Then, for all
$\sigma^{m} \leq \sigma$, the unique point in $\vorcellamb{\sigma} \,
\cap \, T_{p}\man \, .  $ will be denoted by $c_{p}(\sigma^{m})$.
We denote the radius of the circumscribing ball centred at
$c_p(\sigma^m)$ by $R_p(\sigma^m)$, i.e., $R_p(\sigma^m) = \norm{p -
  c_p(\sigma^m)}$.

In our algorithm we will use the following complex, whose definition
employs a particular elementary weight function:
\begin{align}
  \label{eq:cosph}
  \sstr{p}{\delta_{0}}
  &= \Bigg\{~\sigma^{m+1} = p_{m+1}*\sigma^{m} \; \big| \; \sigma^{m}
  \in \str{p}, \,
  R_{p}(\sigma^{m}) < \epsilon,\nonumber\\
  &\quad\quad~\sigma^{m}~\mbox{is $\flakebnd$-good, and }
  \exists \, \gewf{\sigma^{m+1}} \mbox{ with }
  \gewf{\sigma^{m+1}}|_{\mathring{\sigma}^m} = 0\nonumber\\
%	&&\quad\quad\quad\quad
  &\quad\quad~\mbox{and}~
  c_{p}(\sigma^{m}) \in \wnormhull{\sigma^{m+1}}{\gewf{\sigma^{m+1}}} \Bigg\}.
\end{align}
The $(m+1)$-dimensional simplices in $\sstr{p}{\delta_{0}}$ are
analogous to inconsistent configurations defined in
\cite{boissonnat2010meshing, boissonnat2011tancplx}.

Unless otherwise stated, whenever $\sigma^{m+1} = p_{m+1} *
\sigma^{m} \in \sstr{p}{\delta_{0}}$, with $\sigma^{m} \in \str{p}$,
the mention of $\omega_{\sigma^{m+1}}$
will refer to the elementary weight function identified in
\Eqnref{eq:cosph}. In particular,
$$
\omega_{\sigma^{m+1}}(p_{i}) = 0
\text{ for all } p_{i}\in \mathring{\sigma}^{m+1}\setminus
p_{m+1},
$$
and
%\omega_{\sigma^{m+1}}(p_{m+1}) = \omega(p_{m+1})
$$
\omega_{\sigma^{m+1}}(p_{m+1}) \in [0, \delta_{0} L(\sigma^{m+1})]
$$
satisfies
$$
\|c_{p}(\sigma^{m}) - p\|^{2} = \|c_{p}(\sigma^{m}) - p_{m+1}\|^{2} -
\ewf{\sigma^{m+1}}{p_{m+1}}^{2}.
$$

% The following lemma is a consequence of
% Lemma~\ref{claim-rel-edge-length-ortho-radius}:
We will exploit the following observations:
\begin{lem}\label{corollary-rel-edge-length-ortho-radius}
  If $\sigma^{m+1} = \splxjoin{p_{m+1}}{\sigma^{m}} \in
  \sstr{p}{\delta_{0}}$ with $\sigma^{m} \in \str{p}$, then
  \begin{equation*}
  \wcircrad{\sigma^{m+1}}{\gewf{\sigma^{m+1}}} \leq
  R_p(\sigma^m)
  \end{equation*}
  and
  \begin{equation*}
    \Delta(\sigma^{m+1}) \leq \frac{2}{1-\delta^{2}_{0}}
    R_{p}(\sigma^{m})
  \end{equation*}
\end{lem}
\begin{proof}
  Since $c_p(\sigma^m) \in
  \wnormhull{\sigma^{m+1}}{\gewf{\sigma^{m+1}}}$, it follows that
  $C(\sigma^{m+1}, \omega_{\sigma^{m+1}})$ is the projection of
  $c_{p}(\sigma^{m})$ into $\aff(\sigma^{m+1})$, and therefore
  $R_{p}(\sigma^{m}) \geq R(\sigma^{m+1}, \omega_{\sigma^{m+1}})$. The
  bound on $\longedge{\sigma^{m+1}}$ now follows directly from
  Lemma~\ref{claim-rel-edge-length-ortho-radius}.
\end{proof}

Boissonnat et al.~\cite{boissonnat2011tancplx}, using
Lemma~\ref{lem-bound-tangent-vorronoi-intrinsic}, showed that we can
compute $\str{p}$ by computing a weighted Delaunay triangulation on
$T_{p}{\man}$ of the points obtained by projecting $\mpts$ onto
$T_{p}\man$.  Once $\str{p}$ has been computed, we can compute
$\sstr{p}{\delta_{0}}$ by a simple distance computation.

The importance of $\sstr{p}{\delta_0}$ lies in the observation that if
an $m$-simplex $\splxs^m \in \str{p}$ is not sufficiently
power-protected, then there will be a simplex in $\sstr{p}{\delta_0}$
that is a witness to this. It is a direct consequence of the
definitions, but we state it explicitly for reference:
\begin{lem}
  \label{lem:empty.cosph.pwr.protects}
  If $\mpts$ is $\tsparseconst \samconst$-sparse, and $\sstr{p}{\delta_{0}} =
  \emptyset$, then every $\sigma^{m} \in \str{p}$ is $\delta_{0}^{2}
  \tsparseconst^{2}\samconst^{2}$-power protected on $T_{p}\man$.
\end{lem}

%%%%%%%%%%%%%%%%%%%%%%%%%%%%%%%%%%%%%%%%%%%%%%%%%%%%%%%%%%%%

\subsubsection{Unfit configurations and the picking region}

The refinement algorithm, at each step, kills an \defn{unfit
  configuration} by inserting a new point $x = \psi_{p}(x')$ where
$x'$ belongs to the so-called {\em picking region} of the unfit
configuration, and $\psi_p$ is the inverse projection defined in
\Eqnref{eq:defn.psip}. We use the term unfit configuration to
distinguish the elements under consideration from other simplices. An
unfit configuration $\phi$ may be one of two types:
\begin{description}
%Type-1 bad
\item[{\em Big configuration:}] An $m$-simplex $\phi =
  \sigma^{m}$ in $\str{p}$ is a \defn{big configuration} if
  $R_{p}(\sigma^{m}) \geq \epsilon $.
%Type-2 bad
\item[{\em Bad configuration:}] A simplex $\phi$ is a \defn{bad
    configuration} if it is $\flakebnd$-bad and it is either an
  $m$-simplex $\phi = \sigma^m \in \str{p}$ that is not a big
  configuration, or it is an $(m+1)$-simplex $\phi = \sigma^{m+1} \in
  \sstr{p}{\delta_0}$.
%   An $m$-simplex $\phi =
%   \sigma^{m}$ (or $m+1$-simplex $\phi = \sigma^{m+1}$) in $\str{p}$
%   (or in $\sstr{p}{\delta_{0}}$) is a {\em bad configuration}
%   if
%   \begin{enumerate}
%   \item[i.]  $R_{p}(\sigma^{m}) < \epsilon$, and
%   \item[ii.]  $\phi$ is $\flakebnd$-bad.
% %$\exists~\tau \leq \phi$ which is a $\Gamma_{0}$-flake.
%   \end{enumerate}
\end{description}
We will show in \Secref{ssec:output-quality},
\Lemref{cor-sigma-m+1-flake}, that in fact \emph{every}
$(m+1)$-simplex in $\sstr{p}{\delta_{0}}$ is a bad configuration.

%\subsubsection{The picking region}

The size of the picking region is governed by a positive parameter
$\pickratio < 1$ called the \defn{picking ratio}.
\begin{de}[Picking region]
  \label{def:picking.region}
  The {\em picking region} of a bad configuration, $\sigma^{m} \in
  \str{p}$ or $p_{m+1} * \sigma^{m} \in \sstr{p}{\delta_{0}}$ with
  $\sigma^{m} \in \str{p}$, denoted by $P(\sigma^{m}, p)$ and
  $P(\sigma^{m+1}, p)$ respectively, is defined to be the
  $m$-dimensional ball
  \begin{equation*}
    \ballamb{c_{p}(\sigma^{m})}{\pickratio R_{p}(\sigma^{m})}
    \cap~\tanspace{p}{\man}.
  \end{equation*}
\end{de}

We choose a point in the picking region so as to minimize the
introduction of new unfit configurations. We are able to avoid
creating new bad configurations provided that the radius of the
potential configuration is not too large. To this end, we introduce
the parameter $\beta > 1$.
\begin{de}[Hitting sets and good points]
  \label{def-hitting-set-map-good-point}
  Let $\phi = \sigma^{m} \in \str{p}$ or $\phi = q * \sigma^{m} \in
  \sstr{p}{\delta_{0}}$ with $\sigma^{m} \in \str{p}$, and $x =
  \psi_{p}(y)$ where $y \in P(\phi, p)$. A set $\sigma \subset \mpts$
  of size $k$, with $k \leq m+1$, is called a {\em hitting set} of $x$
  if
	
  \quad a. $\splxt = x * \sigma$ is a $k$-dimensional
  $\flakebnd$-flake
	
  and there exists an elementary weight function $\omega_{\splxt}$
  satisfying the following condition:
	
  \quad b. $R(\splxt, \omega_{\splxt}) < \beta
  R_{p}(\sigma^{m})$

  The elementary weight function $\omega_{\splxt}$ is called a {\em
    hitting map}, and we sometimes say $\sigma$ \defn{hits} $x$.
	
  A point $x = \psi_{p}(y)$, where $y \in P(\phi, p)$, is said to be a
  {\em good point} if it is not hit by any set $\sigma \subset \mpts$
  with $|\sigma| \leq m+1$.
\end{de}
A simplex $\splxs$ which defines a hitting set of $x$, is necessarily
$\flakebnd$-good. This follows from the requirement that
$\splxjoin{x}{\splxs}$ be a $\flakebnd$-flake.

%\rd{Arijit: \Large{Comment: How to compute hitting set?} Complexity
%algebraic/time?}

    % -*- LaTeX -*-
% algo_main.tex
% 20120808
% main part of the algorithm section

\subsection{The refinement algorithm}
%\subsection{Refinement terminology}
\label{ssec-algorithm-manifold-case}

In this section, we show that we can refine an $\epsilon$-net of
$\man$ so that the simplices of the Delaunay tangential complex of the
refined sample $\tancplx{\mpts}$ are power-protected.  An
\defn{$\epsilon$-net} is a point sample ${\mpts} \subset \man$ that is
an $\epsilon$-sparse $\epsilon$-sample set of $\man$ for the metric
$\gdist_{\reel^{N}}$. One can obtain an $\epsilon$-net by using a
farthest point strategy to select a subset of a sufficiently dense
sample set.
% using a simple
% greedy construction analyzed by Gonzalez~\cite{gonzales1985}.  
We will
assume that we know the dimension $m$ of the submanifold $\man$ and
the tangent space $T_{p}\man$ at any point $p$ in $\man$.

The algorithm takes as input $\mpts_{0}$, an $\epsilon$-net of $\man$,
and the positive input parameters $\epsilon$, $\flakebnd$, $\alpha
<\frac{1}{2}$, $\beta > 1$ and $\delta_{0} < \frac{1}{4}$.
% (The $\delta_0 < 1/4$ is from the forbidden zone vol lemma)
The algorithm refines the input point sample such that:
\begin{enumerate}
\item[(1)] The output sample $\mpts \supseteq \mpts_{0}$ is an
  $\tsparseconst \samconst$-sparse $\samconst$-sample set of $\man$
  with respect to $\gdistamb$, where $\mu_{0} = \frac{1}{9}$.
\item[(2)] For all $p \in \mpts$, every $m$-simplex $\sigma^{m} \in
  \str{p;\tancplx{\mpts}}$, $\sigma^{m}$ is $\Gamma_{0}$-good and
  $\delta_{0}^2\tsparseconst^2 \samconst^{2}$-power protected on
  $\tanspace{p}{\man}$.
\end{enumerate}

\begin{algorithm}
\caption{Refinement algorithm}
\label{algo-refinement-algorithm-manifold}
\begin{algorithmic}
	\STATE{Input}~~~~~$\epsilon$-net ${\mpts_{0}}$ of $\man$, 
	and input parameters $\flakebnd$, $\alpha$ and $\delta_{0}$;

	\STATE{Initalize}~~$\mpts \gets \mpts_{0}$, and calculate $\tancplx{\mpts}$;

	\STATE{Rule~(1)} {{\em Big configuration ($\samconst$-big radius):}}\\
	~~~~~~~~~~~~if $\exists \; p \in \mpts$ such that $\exists \;
        \sigma^{m} \in \str{p}$ with
	$R_{p}(\sigma^{m}) \geq \epsilon$,\\
	~~~~~~~~~~~~then Insert$(\psi_{p}(c_{p}(\sigma^{m})))$;

	\STATE{Rule~(2)} {{\em Bad configuration ($\flakebnd$-bad):}}\\
	~~~~~~~~~~~~{if} $\exists \; p \in \mpts$ and $\exists \; \sigma^{m} \in \str{p}$ 
	s.t. $\sigma^m$ is $\flakebnd$-bad, \\
        % $\exists$ $\Gamma_{0}$-flake $\sigma^{k}$ and $\sigma^{k}
        % \leq \sigma^{m}$, \\ 
	~~~~~~~~~~~~{then} Insert$({\rm Pick \_ valid}(\sigma^{m}, p))$;\\
	~~~~~~~~~~~~{if} $\exists \; p \in \mpts$ and $\exists \;
        \sigma^{m+1} \in \sstr{p}{\delta_{0}}$ 
	s.t. $\sigma^{m+1}$ is $\flakebnd$-bad, \\
        % $\exists$ $\Gamma_{0}$-flake $\sigma^{k}$
	% and $\sigma^{k} \leq \sigma^{m+1}$, \\
	~~~~~~~~~~~~{then} Insert$({\rm Pick \_ valid}(\sigma^{m+1}, p))$;\\
	
	\STATE{Output} ~$\tancplx{\mpts} = \cup_{p \in \mpts} \; \str{p}$;
\end{algorithmic}
\end{algorithm}
 The algorithm, described in
Algorithm~\ref{algo-refinement-algorithm-manifold}, applies two rules
with a priority order: Rule (2) is applied only if Rule (1) cannot be
applied. The algorithm ends when no rule applies any more.  Each rule
inserts a new point to kill an unfit configuration: either a big
configuration or a bad configuration.  

A crucial procedure, that selects the location of the point to be
inserted, is {\rm Pick\_valid}, given in Algorithm~\ref{pick-valid-x}.
{\rm Pick\_valid}$(\phi, p)$ returns a good point $\psi_{p}(y)$ where
$y \in P(\phi, p)$.
\begin{algorithm}
\caption{Pick\_valid$(\sigma, p)$}
\label{pick-valid-x}
\begin{algorithmic}

	\STATE // Assume that $\sigma$ is either equal to $\sigma^{m} \in \str{p}$ 
	
	\STATE // or 
	$\sigma^{m+1} = p_{m+1} * \sigma^{m} \in \sstr{p}{\delta_{0}}$
        with $\sigma^{m} \in \str{p}$ 
	
	\STATE{Step 1.} {Pick randomly} $y \in P(\sigma^{m}, p)$ (or $P(\sigma^{m+1}, p)$);

	\STATE // Recall that $\psi_{p}$ projects points from $T_{p}\man$ onto $\man$ along $N_{p}\man$
	\STATE{Step 2.} $x \gets \psi_{p}(y)$;
	
	\STATE{Step 3.} {\em Avoid hitting sets:} \\
	~~~~~~~~~~~// $|\tilde{\sigma}|$ denotes the cardinality of $\tilde{\sigma}$\\
	~~~~~~~~~~~if $\exists \; \tilde{\sigma} \subset \mpts$, with $|\tilde{\sigma}| \leq m+1$,  which is a {\em hitting set} of $x$,\\
	~~~~~~~~~~~then discard $x$, and go back to {Step 1};
	
	\STATE{Step 4.} Return $x$;
\end{algorithmic}
\end{algorithm}

The refinement algorithm will also use the procedure {Insert}$(p)$,
given in Algorithm~\ref{function-insert-p}.
\begin{algorithm}
\caption{{Insert}$(p)$}
\label{function-insert-p}
\begin{algorithmic}
	\STATE{Step 1.} Add $p$ to $\mpts$;
		
	\STATE{Step 2.} Compute $\str{p}$ and $\sstr{p}{\delta_{0}}$;
		
	\STATE{Step 3.} For all $x \in \mpts \setminus \{ p \} $, update $\str{x}$ and $\sstr{x}{\delta_{0}}$;
\end{algorithmic}
\end{algorithm}

% -*- LaTeX -*-
% analysis_sparsity.tex
% 20120823
%
% First part of the Analysis Section and the first part of the
% Termination subsection (including all of the Sparsity subsubsection)
%

\section{Analysis of the algorithm}
\label{sec:alg.analysis}

We now turn to the demonstration of the correctness of
Algorithm~\ref{algo-refinement-algorithm-manifold}. In
\Secref{sec:termination} we show that the algorithm must terminate,
and in \Secref{sec:quality} we show that the output of the algorithm
meets the requirements of \Thmref{thm:tan.cplx}. In order to complete
the demonstrations we impose a number of requirements on the input
parameters, listed as Hypotheses $\mathcal{H}0$ to $\mathcal{H}5$
below.

Recall that our input parameters are the following positive numbers:
$\samconst$, which is the sampling radius and sparsity bound satisfied
by $\mpts_0$, the input $\samconst$-net sample set; $\delta_0$, which
is used to describe the amount of power-protection enjoyed by the
$m$-simplices in the final complex; $\flakebnd$, which is used to
quantify the quality of the output simplices; $\beta$, which is used
to describe an upper bound on the radius of the bad configurations
that we will avoid; and $\alpha$, which governs the relative size of
the picking region.

It is often convenient to represent the sampling radius by a
dimension-free parameter that has the reach of the
manifold factored out. We define
\begin{equation*}
  \tilde{\epsilon} %\stackrel{{\rm def}}{=}
  = \frac{\epsilon}{\reach}.
\end{equation*}
The volume of the $m$-dimensional Euclidean unit-ball is denoted
$\ballvolm$.  In order to state the hypotheses on the input parameters,
we use some additional symbols:
\begin{align*}
  \tilde{\epsilon}_{0} &= %\stackrel{{\rm def}}{=}
  \frac{1}{2^{4}(2^{4}+1)^{2}},\\
  B &= 4 + 2 (1+2^{7}3^{2}\beta^{2})^{2},\\
%  \intertext{and}
  \beta' &= %\stackrel{{\rm def}}{=}
  \frac{\beta}{1-2^{4}\tilde{\epsilon}_{0}},
\end{align*}
as well as $\xi$, $E$, and $D$. The term $\xi$ is introduced in
Lemma~\ref{lem-1-appendix-7}, and depends on $m$ and $\reach$, and the
term $E$, defined in \Eqnref{eqn-bound-b+-pp-E}, depends on $\xi$ and
$\beta$. The symbol $D$ is introduced in
Lemma~\ref{lem-bound-flake-forbidden-volume}, where it is said to
depend on $m$ and $\beta$. %what abt $\reach$?

In order to guarantee termination, we demand the following hypotheses
on the input parameters:
\begin{description}
\item[${\mathcal H0}$.] $\alpha < 1/2$
\item[${\mathcal H1}$.] $\beta \geq \frac{2}{(1-\delta^{2}_{0}) (1-
    \alpha - 4.5\, \tilde{\epsilon}_{0})}$
\item[${\mathcal H2}$.] $\Gamma_{0} < \min \left\{ \frac{\ballvolm
      \alpha^{m}}{E^{m+1} \beta^{m}D}, \frac{1}{B+1}\right\}$
\item[${\mathcal H3}$.] $\delta^{2}_{0} \leq \Gamma^{m+1}_{0}$
\item[${\mathcal H4}$.]  $\tilde{\epsilon} \leq \min \left\{
    \frac{\xi}{2(\beta+ \beta')\reach}, \, \frac{\Gamma^{m+1}_{0}}{8
      \beta} \right\}$
\end{description}
To meet the quality requirements of \Thmref{thm:tan.cplx} we
demand an additional constraint on the sampling radius:
\begin{description}
\item[${\mathcal H5.}$] $\tilde{\epsilon} \leq \frac{\delta^{2}_{0}
    \Gamma^{2m}_{0}}{1.1 \times 10^{9}}$
\end{description}

The make use of the following observation:
\begin{lem}
  \label{obs:bound-on-epsilon-delta-H0-H5}
  From hypotheses~${\mathcal H0}$ to ${\mathcal H4}$ we have
  $\tilde{\epsilon} < \tilde{\epsilon}_{0}$ and $\delta^{2}_{0} <
  2^{4}\, \tilde{\epsilon}_{0}$, and
  \begin{equation}
    \label{eq:sparsity.bound}
    \frac{(1 - \delta^{2}_{0}) ( 1 - \alpha - 4.5
      \tilde{\epsilon}_{0})\epsilon}{4} > \frac{\epsilon}{9}
    \stackrel{{\rm def}}{=} \tsparseconst \epsilon.
  \end{equation}
\end{lem}
\begin{proof}
  From ${\mathcal H}1$ we have $\beta > 2$ and using the fact that $B
  > \beta^{4}$ and ${\mathcal H}2$ we have $\Gamma_{0} <
  \frac{1}{2^{4}+1}$. And using the fact, from ${\mathcal H4}$, that
  \begin{equation*}
    \tilde{\epsilon} \leq \frac{\Gamma_{0}^{m+1}}{8\beta} \leq
    \frac{\Gamma_{0}^{2}}{8\beta} < \frac{1}{2^{4}(2^{4}+1)^{2}} =
    \tilde{\epsilon}_{0}.
  \end{equation*}
  Similarly the bound on $\delta^{2}_{0}$ follows from ${\mathcal
    H}3$.
    
  Inequality~\eqref{eq:sparsity.bound} follows from $\mathcal{H}0$ and
  the definition of $\tilde{\epsilon}_0$.
\end{proof}
From \Eqnref{eq:sparsity.bound} we can see that we require $\beta
\geq 4.5$. Given $\alpha$ satisfying $\mathcal{H}0$, and a valid
choice for $\beta$, the hypotheses $\mathcal{H}2$ to $\mathcal{H}4$
sequentially yield upper bounds on the parameters $\flakebnd$,
$\delta_0$, and $\tilde{\epsilon}$; we are able to choose parameters
that satisfy all of the hypotheses.

The main result of this section can now be summarised:
\begin{thm}[Algorithm guarantee]
  \label{thm:algo.guarantee}
  If the input parameters satisfy hypotheses $\mathcal{H}0$ to
  $\mathcal{H}5$, then \Algref{algo-refinement-algorithm-manifold}
  terminates after producing an intrinsic Delaunay complex $\delMmpts$
  that triangulates $\man$.
\end{thm}

%%%%%%%%%%%%%%%%%%%%%%%%%%%%%%%%%%%%%%%%%%%%%%%%%%%%%%%
\subsection{Termination of the algorithm}
\label{ssec:alg.termination}
\label{sec:termination} %redundant label

This subsection is devoted to the proof of the following theorem:
\begin{thm}[Algorithm termination]
  \label{lem-2-chapter-7-algorithm-analysis}
  \label{thm:termination} %my redundant label
  Under hypotheses $\mathcal{H}0$ to $\mathcal{H}4$, the application
  of Rule~$(1)$ or Rule~$(2)$ on a big or a bad configuration $\phi$
  always leaves the interpoint distance greater than
  %does not decrease the interpoint distance to less than
  % $$
  % \frac{(1-\delta^{2}_{0})
  %   (\,1-\alpha-4.5\,\tilde{\epsilon}_{0}\,)\, \epsilon}{4} >
  % \frac{\epsilon}{9}
  % \stackrel{{\rm def}}{=} \mu_{0} \epsilon,
  %   $$
  \begin{equation*}
    \tsparseconst \epsilon = \frac{\epsilon}{9},
  \end{equation*}
  and if $\phi$ is a bad configuration then there exists $x \in
  P(\phi, p)$ such that $\psi_{p}(x)$ is a good point.  Since $\man$
  is a compact manifold this implies that the refinement algorithm
  terminates and returns a point sample $\mpts$ which is an
  $\tsparseconst \epsilon$-sparse $\epsilon$-sample of the manifold
  $\man$.
\end{thm}
We will prove that at every step the algorithm maintains the following
two {\em invariants}:
\begin{description}
\item[Sparsity:] Whenever a refinement rule inserts a new point $x =
  \psi_p(y)$, the distance between $x$ and the existing point set
  $\mpts$ is greater than $\tsparseconst \samconst$.
\item[Good points:] For a bad configuration $\phi$ refined by
  Rule~(2), there exists a set of positive volume $G \subseteq
  P(\phi,p)$ such that if $x\in G$, then $\psi_{p}(x)$ is a good
  point.
\end{description}
The Termination \Thmref{thm:termination} is a direct consequence of
these two algorithmic invariants. We first prove the sparsity
invariant in \Secref{sec:sparsity}, using an induction argument that
relies on the fact that the algorithm only inserts good points.  The
existence of good points is then established in \Secref{sec:goodpts},
using the sparsity invariant and a volumetric argument. Termination
must follow since $\man$ is compact and therefore can only support a
finite number of sample points satisfying a minimum interpoint
distance.

%%%%%%%%%%%%%%%%%%%%%%%%%%%%%%%%%%%%%%%%%%%%%%%%%%%%%
\subsubsection{The sparsity invariant}
\label{sec:sparsity}

The proof of the sparsity invariant employs the following observation,
which serves to bound the distance between a point inserted by
Rule~(2) and the existing point set:
\begin{lem}
  \label{lem-dist-psi-p-x-pp}
  Assume Hypotheses~${\mathcal H0}$ to ${\mathcal H4}$.  Let $\phi =
  \sigma^{m} \in \str{p}$ or $\phi = p_{m+1} * \sigma^{m} \in
  \sstr{p}{\delta_{0}}$ be a bad configuration being refined by
  Rule~(2). Then for all $x \in P(\phi, p)$ we have
  \begin{equation*}
    d_{\reel^{N}}(c_{p}(\sigma^{m}) , \psi_{p}(x)) < (\alpha +
    4.5 \tilde{\epsilon}_{0} ) R_{p}(\sigma^{m})
  \end{equation*}
  and
  \begin{equation*}
    d_{\reel^{N}}(\psi_{p}(x), \mpts) > (1-\alpha - 4.5
    \tilde{\epsilon}_{0})R_{p}(\sigma^{m})
    > \frac{R_{p}(\sigma^{m})}{3}.
  \end{equation*}
\end{lem}
\begin{proof}%[of \Lemref{lem-dist-psi-p-x-pp}]
% We will first consider the case when $\phi = \sigma^{m} \in \str{p}$ or
% $\phi = p_{m+1} * \sigma^{m} \in \sstr{p}{\delta_{0}}$ with
% $\sigma^{m} \in \str{p}$, to be refined by Rule~(2).
%
  Using the facts that $\alpha < \frac{1}{2}$, and $\tilde{\epsilon} <
  \tilde{\epsilon}_{0}$, and $R_{p}(\sigma^{m}) < \epsilon$, we have
  that for all $x \in P(\phi, p)$
  $$
  \| p - x \| < (1+\alpha) R_{p}(\sigma^{m}) < \frac{3\epsilon}{2} <
  \frac{3\,\tilde{\epsilon}_{0}}{2}< \frac{1}{4},
  $$
  and so we may apply \Lemref{lem:dist.tan.to.psip}
  % Lemma~\ref{lem-new-addition-1}~(2)
  to get
  \begin{equation*}
    % \label{eqn-distance-x-psi-p-x}
    \| x -\psi_{p}(x)\| \leq \frac{2\| p - x\|^{2}}{\reach} \leq
    \frac{2(1+\alpha)^{2}R_{p}(\sigma^{m})^{2}}{\reach}
    \leq 4.5\, \tilde{\epsilon}_{0} R_{p}(\sigma^{m}),
  \end{equation*}
  and
  \begin{eqnarray}\label{eqn-distance-cp-x}
    \| c_{p}(\sigma^{m}) - \psi_{p}(x) \| &\leq& \|
    c_{p}(\sigma^{m}) - x \| +\| x - \psi_{p}(x) \| \nonumber\\
    &\leq& \left( \alpha+ 4.5\,\tilde{\epsilon}_{0} \right)\,  R_{p}(\sigma^{m}).
  \end{eqnarray}

  Let $S_{p} = \partial B_{\reel^{N}}(c_{p}(\sigma^{m});
  R_{p}(\sigma^{m}))$. From Eq.~\eqref{eqn-distance-cp-x} we have for
  $x \in P(\phi, p)$
  \begin{eqnarray*}%\label{eqn-distance-psi-p-x-mpts}
    d_{\reel^{N}}(\psi_{p}(x); \mpts) &\geq&
    d_{\reel^{N}}(\psi_{p}(x), S_{p}) \nonumber\\ 
    &>& \left( 1 - \alpha - 4.5\,\tilde{\epsilon}_{0} \right)\,
    R_{p}(\sigma^{m})\\
    &>& \frac{R_p(\sigma^m)}{3},
  \end{eqnarray*}
  where the final inequality follows from $\mathcal{H}0$ and the
  definition of $\tilde{\epsilon}_0$.
\end{proof}
%
% We will exploit an upper bound on $R_p(\sigma)$ that was established
% in previous work~\cite[Lemma~2.7~(1)]{boissonnat2010meshing}:
% \begin{lem}
% %  \label{lem-new-addition-1}
%   \label{lem:Rp.small}
%   Let $\mpts \subset \man$ satisfy a sampling radius of $\epsilon$
%   with respect to $\gdistamb$ such that $\tilde{\epsilon} \leq {1} /
%   {16}$.  Then, for all $ p \in \mpts $, and every $m$-simplex $\sigma
%   \in \str{p}$, we have $R_{p}(\sigma) \leq 4\epsilon$.
% \end{lem}

We introduce some additional terminology to facilitate the
demonstration of the sparsity invariant.  An abstract simplex in the
initial sample set $\sigma \subset \mpts_{0}$ is called an
\defn{original simplex}, otherwise $\sigma \subset \mpts$ is called a
\defn{created simplex}.

Let $\phi$ be an unfit configuration that was refined by inserting a point $x$.
We say that $x$ {\em created $\sigma$} if $x \in \sigma$ and $x$ is the
last inserted vertex of the simplex $\sigma$, i.e., $\sigma\setminus \{x\}$
already existed just before the refinement of the unfit configuration $\phi$.
The unfit configuration $\phi$ is called the {\em parent of $\sigma$} and will
be denoted $\pa(\sigma)$.

Let $\sigma$ denote the simplex being refined by the refinement algorithm. We
will denote by $\ed(\sigma)$ the distance between the  point newly
inserted to refine $\sigma$ and the current sample set.

The sparsity invariant is demonstrated by induction.  We use a case
analysis according to the type of unfit configuration being refined;
it is necessary to consider sub-cases. The induction hypothesis is
employed only in the sub-case \textbf{Case~2(b)(ii)} and the implicit
similar \textbf{Case~3(b)(ii)}. The base for the induction hypothesis,
i.e., the insertion of the first point, cannot involve
\textbf{Case~2(b)} or \textbf{Case~3(b)}.

\begin{description}
\item[Case 1.]
  Let $\phi = \sigma^{m} \in \str{p}$ be a big
  configuration being refined by Rule~(1).

  Since $\mpts_{0}$ ($\subseteq \mpts$) is an $\epsilon$-net, we have
  from the fact that $\tilde{\epsilon} \leq \tilde{\epsilon}_{0} <
  \frac{1}{16}$ and
  \Lemref{lem:Rp.small}, %Lemma~\ref{lem-new-addition-1}~(1),
  $R_{p}(\sigma^{m}) \leq 4\epsilon$. Rule (1) will refine $\sigma$ by
  inserting $\psi_{p}(c_{p}(\sigma^{m}))$.  Using the fact that
  $\tilde{\epsilon} < \tilde{\epsilon}_{0} < \frac{1}{16}$,
  $R_{p}(\sigma^{m}) \leq 4\epsilon$ and $R_{p}(\sigma^{m}) \geq
  \epsilon$ (since $\sigma^{m}$ is being refined by Rule~(1)), and
%Lemma~\ref{lem-new-addition-1}~(2),
  \Lemref{lem:dist.tan.to.psip}, the distance between
  $\psi_{p}(c_{p}(\sigma^{m}))$ and any vertex inserted before
  $\psi_{p}(c_{p}(\sigma^{m}))$ is not less than
  \begin{align*}%\label{eqn:bound-for-rule-1}
    R_{p}(\sigma) - \|c_{p}(\sigma) - \psi_{p}(c_{p}(\sigma))\| &\geq
    R_{p}(\sigma)-\frac{2R_{p}(\sigma)^{2}}{\reach}\\
    &> (1 - 8\, \tilde{\epsilon}_{0}) \, \epsilon\\
    &> \frac{\epsilon}{2},
  \end{align*}
  which establishes the sparsity invariant for this case.

\item[Case 2.]
  Consider now the case where $\phi = p_{m+1} *
  \sigma^{m} \in \sstr{p}{\delta_{0}}$, with $\sigma^{m} \in \str{p}$,
  is being refined by Rule~(2).
% Let $\omega(p_{m+1}) \in [0, \, \delta_{0}L(\phi)]$
% be such that
% \begin{equation}\label{eqn-defining-omega-p-m+1}
%   \| c_{p}(\sigma^{m}) - p \|^{2} = \| c_{p}(\sigma^{m}) -
%   p_{m+1} \|^{2} - \omega(p_{m+1})^{2}
% \end{equation}
%
% Since $\phi$ is being refined by Rule~(2) we have
  In this case, recalling \Lemref{lem:bad.has.flake}, we have
  \begin{itemize}
  \item $R_{p}(\sigma^{m}) < \epsilon$, and
  \item there exists a face of $\phi$ that is a $\flakebnd$-flake.
  \end{itemize}
  Let $\sigma_{1} \subseteq \phi$ denote a face of $\phi$ that is a
  $\Gamma_{0}$-flake. We have to now consider two cases:
  \begin{enumerate}[label=\textbf{(\alph*)}]
  \item $\sigma_{1}$ is an original simplex
  \item $\sigma_{1}$ is a created simplex
  \end{enumerate}

\item[Case 2(a).]
  If $\sigma_{1}$ is an original simplex then
  $\sigma_{1} \subseteq \mpts_{0}$, and since $\mpts_{0}$ is an
  $\samconst$-net, $\shortedge{\sigma_1} \geq \samconst$. Since a
  flake must have at least three vertices, $\sigma_1$ and $\sigma^m$
  must share at least two vertices, and therefore $R(\sigma^{m}) \geq
  \epsilon / 2$.
% $R(\sigma_{1})
% \geq \epsilon/2$, and therefore $R(\phi) \geq \epsilon / 2$.  Also note
% that since $|\sigma_{1} \cap \sigma^{m} | \geq 2$, considering
% $\sigma_{1}$ and $\sigma^{m}$ as subset of the point sample. Since
% $\sigma_{1} \subset \mpts_{0}$, $R(\sigma^{m}) \geq \epsilon / 2$.

  Let $x=\psi_{p}(x')$ be point inserted to refine $\phi$ where $x'
  \in P(\phi, p)$. Using Lemma~\ref{lem-dist-psi-p-x-pp} and the fact
  that $R(\sigma^{m}) \geq \epsilon / 2$, we therefore have
  \begin{align*}%\label{eqn-distance-psi-x-rest-points}
    \distamb{x}{\mpts} &>
    (1-\alpha-4.5\,\tilde{\epsilon}_{0}) R_{p}(\sigma^{m}) \\
    &\geq
    (1-\alpha-4.5\,\tilde{\epsilon}_{0}) R(\sigma^{m}) \\
    &\geq \frac{(1-\alpha-4.5\,\tilde{\epsilon}_{0})\epsilon}{2}\\
    &> \tsparseconst \samconst.
  \end{align*}
  where the final inequality follows from 
  Inequality~\eqref{eq:sparsity.bound}.  Hence the sparsity invariant is
  maintained on the refinement of $\phi$ if $\sigma_{1}$ is an
  original simplex.

\item[Case 2(b)]
  We will now consider the case when $\sigma_{1}$ is a
  created simplex.  We denote by $\pa(\sigma_{1})$ the parent simplex
  whose refinement gave birth to $\sigma_{1}$.

% We will now prove that distance is at least
% $\frac{(1-\delta^{2}_{0})(1-\alpha-4.5\,\tilde{\epsilon}_{0})\epsilon}{{4}}$
% between $x= \psi_{p}(x')$, where $x' \in P(\phi, p)$, and the point
% set $\mpts$.  From Lemma~\ref{lem-dist-psi-p-x-pp}
% %as in the case of
% %Eq.~\eqref{eqn-distance-psi-x-rest-points},
% %{together with \Lemref{lem-2-appendix-7}},
% we get the
% distance between $x$ and the rest of the points in $\mpts$ is no less than
% $$
% (1 - \alpha - 4.5\,\tilde{\epsilon}_{0})R_{p}(\sigma^{m})
% $$
  We will bound the distance between $x= \psi_{p}(x')$, where $x' \in
  P(\phi, p)$, and the point set $\mpts$.  Let $x^*$ denote the
  point whose insertion killed
  $\pa(\sigma_{1})$.  %See Figure~\ref{fig:termination-chap-7}.
  By definition $x^*$ is a vertex of $\sigma_{1}$, and hence also of
  $\phi$ since $\sigma_{1} \leq \phi$.  We distinguish the
  following two cases:
%  \begin{figure}
%    \begin{center}
%      \includegraphics[width=5.0cm]{imgs/termination-chap-7}
%    \end{center}
%    \caption{A schematic illustration of the main objects in
%      \textbf{Case~2(b)} of the demonstration of the sparsity invariant.}
%   \label{fig:termination-chap-7}
%  \end{figure}

\item[Case 2(b)(i)] %\ednote{Check this calculation:}
  Suppose $p(\sigma_{1})$ was a big configuration refined by the
  application of Rule~(1).  According to {\bf Case~1}, the lengths of
  the edges incident to $x^*$ in $\sigma_{1}$ are greater than
  $\epsilon /2$. %, thus $\longedge{\sigma_1} > \frac{\samconst}{2}$.
  Therefore
%  the distance between $x$ and the other points is no less than
  \begin{align*}
    \distamb{x}{\mpts} &\geq
    (1-\alpha-4.5\,\tilde{\epsilon}_{0})R_{p}(\sigma^{m}) \quad
    &\text{by \Lemref{lem-dist-psi-p-x-pp}}\\
    &\geq \frac{(1-\delta^{2}_{0})(1-\alpha-4.5\,\tilde{\epsilon}_{0})
      \Delta(\phi)}{2}
    &\mbox{by Lemma~\ref{corollary-rel-edge-length-ortho-radius}}\\
    &\geq \frac{(1-\delta^{2}_{0})(1-\alpha-4.5\,\tilde{\epsilon}_{0})
      \Delta(\sigma_{1})}{2}
    &\mbox{since $\sigma_{1} \leq \phi$}\\
    &> \frac{(1-\delta^{2}_{0})(1-\alpha-4.5\,\tilde{\epsilon}_{0})
      \epsilon}{4} &\text{by \textbf{Case~1}}\\
    &> \tsparseconst \samconst
    &\text{Inequality~\ref{eq:sparsity.bound}},
  \end{align*}
  and the sparsity invariant is maintained.

\item[Case 2(b)(ii)]
  Suppose $\pa(\sigma_{1})$ was a bad configuration
  refined by Rule~(2).  Thus $\pa(\sigma_{1})$ was either an
  $m$-simplex $\sigma^{m}_{2} \in \str{q}$ or an $(m+1)$-simplex
  $q_{m+1} * \sigma_{2}^{m} \in \sstr{q}{\delta_{0}}$ with
  $\sigma_{2}^{m} \in \str{q}$.
	
  Consider the elementary weight function $\gewf{\sigma_1} =
  \gewf{\phi}|_{\mathring{\sigma_1}}$, where $\gewf{\phi}$ is the
  weight function~\eqref{eq:cosph} identifying $\phi$ as a member of
  $\sstr{p}{\delta_0}$.  From
  \Lemref{claim-rel-edge-length-ortho-radius}(1), and
  \Lemref{corollary-rel-edge-length-ortho-radius} we have that
  $R_p(\sigma^m) \geq \wcircrad{\sigma_1}{\gewf{\sigma_1}}$.
  We also have that $R(\sigma_{1}, \omega_{\sigma_{1}}) \geq \beta\,
  R_{q}(\sigma_{2}^{m})$. Indeed, otherwise $\sigma_1 \setminus \{x^*\}$
  would be a hitting set for $x^*$, contradicting the hypothesis that
  $\pa(\sigma_{1})$ was refined according to Rule~(2) by the insertion
  of a good point $x^{*}$.  Thus we have
% Using the
%   fact that $R(\sigma_{1}, \omega_{\sigma_{1}}) \geq \beta
%   R_{q}(\sigma_{2}^{m})$, we get the distance between $x$ and the
%   other points is thus greater than
  \begin{align*}
    \distamb{x}{\mpts} &>
    (1-\alpha-4.5\tilde{\epsilon}_{0})R_{p}(\sigma^{m})&\\
    &\geq (1-\alpha-4.5\tilde{\epsilon}_{0})R(\sigma_{1},
    \omega_{\sigma_{1}} )&\\
    &\geq (1-\alpha-4.5\tilde{\epsilon}_{0})\,\beta\,  R_{q}(\sigma_{2}^{m})&\\
    &\geq
    \frac{(1-\delta_{0}^{2})(1-\alpha-4.5\tilde{\epsilon}_{0})\beta
      \Delta (\sigma_{2}^{m} )}{2}
    &\text{from \Lemref{corollary-rel-edge-length-ortho-radius}}\\
    &> \Delta(\sigma_{2}^{m})
    &\text{from Hypotheses~${\mathcal H}1$ on $\beta$}\\
    % &\geq \frac{(1-\delta_{0}^{2})(1-\alpha-4.5\tilde{\epsilon}_{0})
    %   \epsilon}{4} & \text{induction hypothesis}\\
    &> \tsparseconst \samconst, &
  \end{align*}
  where the last inequality follows from the induction hypothesis.
  Again the sparsity invariant is maintained after refinement of
  $\phi$.

\item[Case 3]
  The proof for the case of a bad configuration $\phi =
  \sigma^{m} \in \str{p}$ to be refined by Rule~(2) is similar to
  \textbf{Case~2}, and the lower bound on the interpoint distances is
  the same.
\end{description}
This completes the demonstration of the sparsity invariant.

    % -*- LaTeX -*-
% goodpts.tex
% 20120823
%
% Second part of the Termination subsection: all of the Existence of
% good points subsubsection
%

\subsubsection{The good points invariant}
\label{sec:goodpts}

We will now show that the good point invariant is maintained if $\phi$
is a bad configuration being refined by Rule~(2). Without loss of
generality, we will assume that $\phi$ is either equal to $\sigma^{m}
\in \str{p}$ or to $q * \sigma^{m} \in \sstr{p}{\delta_{0}}$, with
$\sigma^{m} \in \str{p}$.

Recall the picking region $P(\phi,p)$ introduced in
\Defref{def:picking.region}.  We will show that there exists $y \in
P(\phi, p)$ such that $x = \psi_{p}(y)$ is a good point.  Let $Y
\subseteq P(\phi,p)$ be the set of points that $\psi_p$ maps to a
point with a hitting set:
\begin{equation*}
  Y = \{y \in P(\phi,p) \, | \, \psi_p(y) \text{ is not a good point} \}.
\end{equation*}
We will show that the volume of $P(\phi,p)$ exceeds the volume of $Y$.
To this end, we will first bound the number of simplices that could
hit some point in $\psi_p(Y)$. Then we will bound the
volume that each potential hitting set can contribute to $Y$.

In order to bound the number of hitting sets, we will use the sparsity
invariant together with the following lemma
\cite[Lemma~4.7]{boissonnat2010meshing} to bound the number of points
that can be a vertex of a hitting set:
\begin{lem}[Bound on sparse points]
  \label{lem-1-appendix-7}
  For a point $p \in \man$ and $R > 0$, let $V$ be a maximal set of
  points in $\ballRM{p}{R}$ %$B_{\reel^{N}}(p, R) \cap \man$
  such that the smallest interpoint distance is not less than
  $2r$. There exists $\xi$ that depends on $m$ and $\reach$, and $A$
  that depends on $m$, such that if $R+r \leq \xi$, then
  $$
  |V| \leq \frac{1+A\xi}{1-A\xi} \left(\frac{R}{r}+1\right)^{m} .
  $$
\end{lem}

We obtain the following bound on the number of hitting sets:
\begin{lem}
  \label{lem:num.hitting.sets}
  Let ${\bf S}(\phi)$ denote the set of simplices contained in ${\bf
    B^{+}} \cap \mpts$ that can hit a point in $\psi_p(Y)$. Then
  \begin{equation}
    \label{eq:num.hitting.sets}
    |\mathbf{S}(\phi)| \leq \frac{E^{m+1}}{2^{m}},
  \end{equation}
  where
  \begin{equation}
    \label{eqn-bound-b+-pp-E}
    E \stackrel{{\rm def}}{=}\;  2\left( \frac{1+A \xi}{1-A \xi} \right)
    \left( 18(\alpha + 2 \beta' + 6.5 \tilde{\epsilon}_{0} ) +1 \right)^{m}.
  \end{equation}
\end{lem}
\begin{proof}
Suppose $\sigma \subseteq \mpts$ is a hitting set of a point $x =
\psi_{p}(y)$, where $y \in P(\phi, p)$, with $|\sigma| = k$ and $k\leq
m+1$.  Let $\tilde{\sigma} = x * \sigma$, and let
$\omega_{\tilde{\sigma}}$ denote the corresponding hitting map (see
Definition~\ref{def-hitting-set-map-good-point}).
% Since $\sigma$ is a hitting set of $x$,
Therefore, we have $R(\tilde{\sigma}, \omega_{\tilde{\sigma}}) < \beta
R_{p}(\sigma^{m})$, and it follows from
\Lemref{claim-rel-edge-length-ortho-radius}(2) that
$\longedge{\tilde{\sigma}} \leq \frac{2\beta}{1 -
  \delta_0^2}R_p(\sigma^m)$.  Thus from Lemma~\ref{lem-dist-psi-p-x-pp} and
the Triangle inequality we have $\sigma \subset \mathbf B^{-}
\stackrel{{\rm def}}{=} \ballamb{c_p(\sigma^m)}{r^-}$, where
\begin{equation*}
  r^- = 4.5\tilde{\epsilon}_{0} R_{p}(\sigma^{m}) + \alpha
  R_{p}(\sigma^{m}) + \frac{2 \beta}{1-\delta^{2}_{0}}
  R_{p}(\sigma^{m}).
\end{equation*}

Let $c = \psi_{p}(c_{p}(\sigma^{m}))$. Then using
\Lemref{lem:dist.tan.to.psip}
% Lemma~\ref{lem-new-addition-1}(2)
and the fact that $R_{p}(\sigma^{m}) < \epsilon < \tilde{\epsilon}_{0}
\,\reach$ we have
\begin{equation*}
  \| c_{p}(\sigma^{m}) - c \| \leq \frac{2
    R_{p}(\sigma^{m})^{2}}{\reach} < 2\tilde{\epsilon}
  R_{p}(\sigma^{m}) \leq 2\tilde{\epsilon}_{0} R_{p}(\sigma^{m}).
\end{equation*}

Using $\delta_{0}^{2} < 2^{4} \tilde{\epsilon}_{0}$ from
Lemma~\ref{obs:bound-on-epsilon-delta-H0-H5}, and $\beta' =
\frac{\beta}{1-2^{4} \tilde{\epsilon}_{0}}$, and $R_{p}(\sigma^{m}) <
\epsilon$, we find
\begin{equation*}
  \begin{split}
    \norm{c_p(\sigma^m) - c} + r^-
    &\leq %2\tilde{\epsilon}_{0} R_{p}(\sigma^{m}) +
     \alpha R_{p}(\sigma^{m})+ 6.5 \tilde{\epsilon}_{0}\,
    R_{p}(\sigma^{m}) + \frac{2 \beta}{1-\delta^{2}_{0}}\,
    R_{p}(\sigma^{m}) \\
   &\leq
    \left(\alpha + 2\beta' + 6.5\, \tilde{\epsilon}_{0} \right) \epsilon\\
    &\stackrel{{\rm def}}{=} R.
\end{split}
\end{equation*}
Thus $\mathbf{B^-} \subseteq \mathbf{B^+} \stackrel{{\rm def}}{=}
\ballamb{c}{R}$, and $y \in Y$ if and only if there exists $\splxs
\subset \mathbf{B}^+ \cap \mpts$ such that $\splxs$ hits $\psi_p(y)$.

Using Lemma~\ref{lem-1-appendix-7} we will bound the number of
sample points in ${\mathbf B ^{+}}\cap \mpts$. Set
% $R=\left(\alpha+2\beta' + 6.5\, \tilde{\epsilon}_{0}
% \right)\,\epsilon$,
$r= \frac{\tsparseconst \samconst}{2} = \frac{\epsilon}{18}$ and
observe that
$$
R+r = \left( \frac{1}{18}+ \alpha + 2 \beta' + 6.5\,
  \tilde{\epsilon}_{0} \right)\, \epsilon \leq (2\beta'+1) \epsilon
\leq \xi,
$$
by Hypothesis~${\mathcal H4}$.  The sparsity invariant and
Lemma~\ref{lem-1-appendix-7} then yields
\begin{equation*}
 | {\bf B}^{+} \cap \mpts |
  \leq   \frac{1+A \xi}{1-A \xi}\times
  \left( \frac{(\alpha + 2 \beta' + 6.5 \tilde{\epsilon}_{0} )}{1/18}
    +1 \right)^{m} = \frac{E}{2}.
\end{equation*}

Since the number of $k$-simplices is less than $\left( \frac{E}{2}
\right)^{k+1}$, and the maximum dimension of a hitting set is $m$, we
have $|\mathbf{S}(\phi)| \leq \frac{E^{m+1}}{2^{m}}$.
\end{proof}

We now turn to the problem of bounding the volume of $Y$.  We will
consider the contribution of each $\splxs \in \mathbf{S}(\phi)$.  The
following definition characterises the set of points in $\man$ that
can be hit by $\splxs$:
\begin{de}[Forbidden region]
  For a $k$-simplex $\sigma$ with vertices in $\man$
%  $L(\sigma) > \frac{\epsilon}{9}$
    with $k \leq m$ and parameter $t < \epsilon$, the \defn{forbidden region}, $F(\sigma, t)$, is the set
  of points $x \in \man$ such that $\sigma_{1} = x * \sigma$ satisfies
  the following conditions: %for some positive $t$:
  \begin{itemize}
  \item $L(\sigma_{1}) > \frac{t }{9}$
  \item $\sigma_{1}$ is a $\Gamma_{0}$-flake
  \item there exists an elementary weight function
    $\omega_{\sigma_{1}}$ s.t. $R(\sigma_{1}, \omega_{\sigma_{1}}) <
    \beta t$
%  \item $t < \epsilon$.
  \end{itemize}
\end{de}
We will use the following lemma, which is proved in
Appendix~\ref{app-proof-lem-bound-flake-forbidden-volume}. It bounds
the volume of the set of points that can be hit by a given simplex:
\begin{lem}[Volume of forbidden region]
  \label{lem-bound-flake-forbidden-volume}
  Let $\sigma$ be a $k$-simplex with vertices on $\man$ and $k \leq
  m$.  If
  \begin{enumerate}
  \item \label{vol:flake.bnd} $\Gamma_{0} \leq \frac{1}{B+1}$,
  \item \label{vol:rad.bnd} $\tilde{\epsilon} \leq \min \{ \frac{\xi}{4\beta\, \reach}, \,
    \frac{\Gamma^{m+1}_{0}}{8\beta} \}$ and
  \item \label{vol:dno.bnd} $\delta_{0}^{2} \leq \min \{ \Gamma_{0}^{m+1}, \frac{1}{4}\}$,
  \end{enumerate}
  then
  \begin{equation*}
    \vol(F(\sigma, t)) \leq D\, \Gamma_{0}\, R(\sigma)^{m},
  \end{equation*}
  where $D$ depends on $m$ and $\beta$. %\ednote{and not $\reach$?}
\end{lem}
\Lemref{lem-bound-flake-forbidden-volume}, together with
\Lemref{lem:num.hitting.sets}, yields a bound on the set of
points $Y$ in the picking region that do not map to a good point:
\begin{lem}
  \label{lem:bound.bad.points}
  The volume of the set $Y \subset P(\phi,p)$ of points that do not map to a
  good point is bounded as follows:
  \begin{equation*}
    \vol(Y) \leq E^{m+1}\beta^{m}\, D\, \Gamma_{0}\,
    R_{p}(\sigma^{m})^{m}. 
 \end{equation*}
\end{lem}
\begin{proof}
Let $t_{0} = R_{p}(\sigma^{m}) < \epsilon$.  For a given $\splxs \in
\mathbf{S}(\phi)$, let $Y_{\splxs} \subseteq Y$ be the set of points
$y$ for which $\splxs$ hits $x = \psi_p(y)$.  Then from
Hypotheses~${\mathcal H0}$ to ${\mathcal H4}$ and
Lemma~\ref{lem-bound-flake-forbidden-volume}, we have
\begin{align}
  \label{eqn-volume-w-sigma}
  \vol(Y_{\sigma})
  & \leq  \vol(\pi_{p}(F(\sigma, t_{0}))) &\nonumber\\
  & \leq \vol(F(\sigma, t_{0})) &\text{since $\pi_{p}$ is a projection map on
    $T_{p}\man$} \nonumber\\
  & \leq  D\, \Gamma_{0}\, R(\sigma)^{m}. &
\end{align}

Let $\sigma_{1} = x * {\sigma}$, and let $\omega_{\sigma_{1}}$
be the corresponding hitting map.  From the definition of
hitting sets and hitting maps, we have $R_{p}(\sigma^{m}) < \epsilon$,
and $R(\sigma_{1}, \omega_{\sigma_{1}}) < \beta R_{p}(\sigma^{m})$ and
$\sigma$ is $\Gamma_{0}^{k}$-thick.
Define $\omega_{\sigma} =
\omega_{\sigma_{1}}\mid_{\mathring{\sigma}}$.  Then, using
Lemma~\ref{claim-rel-edge-length-ortho-radius}~(3)  and the fact that
$R(\sigma, \omega_{\sigma}) \leq R(\sigma_{1}, \omega_{\sigma_{1}}) <
\beta R_{p}(\sigma^{m})$, we have
\begin{align}
  \label{eqn-relation-r-sigma-rp-sigmam}
  R(\sigma) &\leq R(\sigma, \omega_{\sigma})\left( 1-
    \frac{\delta^{2}_{0}}{\Upsilon(\sigma)} \right)^{-1} &\nonumber\\
  &\leq R(\sigma, \omega_{\sigma}) \left( 1 -
    \frac{\delta^{2}_{0}}{\Gamma^{m}_{0}} \right)^{-1} &\text{since
    $\Upsilon(\sigma) \geq \Gamma^{k}_{0} \geq
    \Gamma_{0}^{m}$}\nonumber\\
  &\leq 2R(\sigma, \omega_{\sigma}) &\text{since
    $\frac{\delta^{2}_{0}}{\Gamma^{m}_{0}} \leq \Gamma_{0} <
    \frac{1}{2}$ from
    Hyp.~${\mathcal H2}$, ${\mathcal H3}$}\nonumber\\
  &< 2\beta R_{p}(\sigma^{m}). &
\end{align}

The inequalities \eqref{eqn-volume-w-sigma} and
\eqref{eqn-relation-r-sigma-rp-sigmam} together yield
\begin{equation*}
  \vol(Y_{\sigma}) \leq  2^{m}\beta^{m} \, D \, \Gamma_{0} \, R_{p}(\sigma^{m})^{m},
\end{equation*}
and so using \Lemref{lem:num.hitting.sets} we have
%Using Eq.~\eqref{eqn-bound-b+-pp-E} and \eqref{eqn-volume-w-sigma} we have
\begin{align*}
  \vol(Y)
  &= \vol\left( \bigcup_{\sigma \in \mathbf{S}(\phi)} Y_{\sigma} \right) \nonumber\\
  &\leq \sum_{\sigma \in \mathbf{S}(\phi)} \vol(Y_{\sigma})  \nonumber\\
  &\leq E^{m+1}\beta^{m}\, D\, \Gamma_{0}\, R_{p}(\sigma^{m})^{m}.
\end{align*}
% The last inequality follows from the fact that $|{\bf B^{+}} \cap \mpts | \leq \frac{E}{2}$ (from Eq.~\eqref{eqn-bound-b+-pp-E})
% and therefore $|S_{\phi}| \leq \frac{E^{m+1}}{2^{m}}$.
\end{proof}

By the definition of the picking region, we have that
\begin{equation*}
	\vol(P(\phi, p)) = \ballvolm \alpha^{m}  R_{p}(\sigma^{m})^{m} .
\end{equation*}
By Hypothesis~${\mathcal H2}$, $E^{m+1} \beta^{m} D \, \Gamma_0\,
R_{p}(\sigma^{m})^{m}$ is less than $\vol(P(\phi, p))$, the volume of
the picking region of $\phi$. Thus with \Lemref{lem:bound.bad.points},
this proves the existence of points $y$ in the picking region $P(\phi,
p)$ of $\phi$ such that $\psi_{p}(y)$ is a good point.

The proof of \Thmref{thm:termination} is complete.

    % -*- LaTeX -*-
% quality.tex
% 20120808
% Output quality for the analysis section

\subsection{Output quality}
\label{ssec:output-quality}
\label{sec:quality} % redundant label for my convenience

We will now show that if Hypothesis $\mathcal{H}5$ is satisfied, in
addition to Hypotheses $\mathcal{H}0$ to $\mathcal{H}4$, then the
output to the refinement algorithm will meet the demands imposed by
\Thmref{thm:tan.cplx}, thus yielding \Thmref{thm:algo.guarantee}.

The main task is to ensure that every $m$-simplex in $\tancplxmpts$
has, for each vertex, a $\delta_0^2 \tsparseconst^2
\samconst^2$-power-protected Delaunay ball centred on the tangent
space of that vertex. This is achieved in two steps. First we
establish conditions to ensure that $\sstr{p}{\delta_{0}} = \emptyset$
for every $p \in \mpts$. As noted by
\Lemref{lem:empty.cosph.pwr.protects}, this ensures that every simplex
in $\str{p}$ has a $\delta_0^2 \tsparseconst^2
\samconst^2$-power-protected Delaunay ball centred on
$\tanspace{p}{\man}$. Next we show conditions such that if $\splxs^m
\in \tancplxmpts$, then $\splxs^m \in \str{p}$ for every vertex $p
\in \splxs^m$. In each step the required conditions impose an
additional constraint on the sampling radius, and this leads to
Hypothesis $\mathcal{H}5$.

As a starting point, we observe the following direct consequence of
the Termination \Thmref{lem-2-chapter-7-algorithm-analysis}:
\begin{cor}
  \label{cor:preliminary-prop-output}
  Under Hypotheses $\mathcal{H}0$ to $\mathcal{H}4$, for all $p \in
  \mpts$, the output of the algorithm satisfies the following:
  \begin{enumerate}
  \item $\sigma^{m} \in \str{p} \implies $ $R_{p}(\sigma^{m}) <
    \epsilon$ and $\sigma^{m}$ is a $\Gamma_{0}$-good simplex, and
  \item all $\sigma^{m+1} \in \sstr{p}{\delta_{0}}$ are
    $\Gamma_{0}$-good.
  \end{enumerate}
\end{cor}
We will show that for an appropriate sampling radius, there cannot be
a $\flakebnd$-good simplex in $\sstr{p}{\delta_{0}}$. We exploit the
following bound on the thickness of a small $(m+1)$-simplex:
\begin{lem}[Small $(m+1)$-simplices are not thick]
  \label{lem-sigma-m+1-fatness}
  Let $\sigma^{m+1}$ be an $(m+1)$-simplex with vertices in $\man$ and
  $\longedge{\sigma^{m+1}} < \reach$. For distinct vertices $p,q \in
  \sigma^{m+1}$ define $\theta = \angle (\aff(\sigma_q),
  \tanspace{p}{\man})$. Then
  \begin{equation*}
    \Upsilon(\sigma^{m+1}) \leq
    \left( \frac{\Delta(\sigma^{m+1})}{2\,\reach} + \sin \theta \right).
  \end{equation*}
\end{lem}
\begin{proof}
  We will bound the altitude $D(q, \sigma^{m+1})$. Let $\ell$ be the
  line through $p$ and $q$.  Using Lemma~\ref{lem:dist.to.tanspace}
  and the fact that $\angle (\aff(\sigma_q), \tanspace{p}{\man}) =
  \theta$, we get
  \begin{align*}
%    \label{eqn-sigma-m+1-fatness-1}
    D(q, \sigma^{m+1})
    &= \gdist_{\reel^{N}}(q, \aff(\sigma_q)) \\
    &=  \sin \angle (\ell, \aff(\sigma_q)) \times \gdist_{\reel^{N}}(p,q)  \\
    &\leq (\sin \angle (\ell, \tanspace{p}{\man}) + \sin \angle
    (\aff(\sigma_q), \tanspace{p}{\man}) )\times
    \gdist_{\reel^{N}}(p,q) \\
    &\leq \left( \frac{\gdist_{\reel^{N}}(p,q)}{2\, \reach} +
      \sin \theta \right) \times \gdist_{\reel^{N}}(p,q)\\
    &\leq \left( \frac{ \Delta(\sigma^{m+1})}{2\, \reach} + \sin
      \theta \right) \times \Delta(\sigma^{m+1}).
  \end{align*}
  Therefore we have
  \begin{equation*}
    \Upsilon(\sigma^{m+1}) \leq \left(
      \frac{\Delta(\sigma^{m+1})}{2\,\reach} + \sin \theta \right). 
  \end{equation*}
\end{proof}
Also, Whitney's \Lemref{lem:whitney.approx} implies that a
$\flakebnd$-good simplex in $\str{p}$ makes a small angle with
the tangent space at $p$: %$\tanspace{p}{\man}$:
\begin{lem}
  \label{lem:thick.small.angle}
  If $\sigma^m \in \str{p}$ is $\flakebnd$-good with $R_p(\sigma^m) <
  \samconst$, then
  \begin{equation*}
    \sin \theta < \frac{2\epsilon }{\Gamma^{m}_{0}\, \reach},
  \end{equation*}
  where $\theta = \angle (\aff(\sigma^{m}), \tanspace{p}{\man})$.
\end{lem}
\begin{proof}
  Let $\zeta = \max_{x \in \sigma^{m}} \gdist_{\reel^{N}}(x,
  \tanspace{p}{\man})$ where $x$ is a vertex of $\sigma^{m}$.  From
  Lemma~\ref{lem:dist.to.tanspace}, we have
  \begin{align*}
    %\label{eqn-protection-star-tangent-space-1}
    \zeta
    &= \max_{x \in \sigma^{m}} \gdist_{\reel^{N}}(x, \tanspace{p}{\man}) \nonumber \\
    &\leq \max_{x \in \sigma^{m}} \frac{\gdist_{\reel^{N}}(p, x)^{2}}{2\reach} \nonumber \\
    &\leq \frac{\Delta(\sigma^{m})^{2}}{2\, \reach}.
  \end{align*}

  %Let $\theta = \angle (\aff(\sigma^{m}), \tanspace{p}{\man})$.
  Using \Lemref{lem:whitney.approx} and the facts that $R(\sigma^{m})
  \leq R_{p}(\sigma^{m}) < \epsilon$ and $\Upsilon(\sigma^{m}) \geq
  \Gamma^{m}_{0}$ (since $\sigma^{m}$ is a $\Gamma_{0}$-good simplex),
  we have
  \begin{align*}
    \sin \theta
    &\leq  \frac{2\zeta}{\Upsilon(\sigma^{m}) \Delta(\sigma^{m})} &\\
    &\leq  \frac{\Delta(\sigma^{m})}{\Upsilon(\sigma^{m}) \reach} &\\
    &< \frac{2\epsilon }{\Gamma^{m}_{0}\, \reach} &\text{since
      $\Delta(\sigma^{m}) \leq 2 R(\sigma^{m}) < 2\epsilon$}.
  \end{align*}	
\end{proof}

Using Lemmas \ref{lem-sigma-m+1-fatness} and
\ref{lem:thick.small.angle} we get that no $(m+1)$-dimensional simplices
in $\sstr{p}{\delta_{0}}$ can be $\flakebnd$-good when $\epsilon$ is
sufficiently small:
\begin{lem}[$\sstr{p}{\delta_0}$ simplices are $\flakebnd$-bad]
  \label{cor-sigma-m+1-flake}
  Let $\sigma^{m+1} = p_{m+1} * \sigma^{m} \in \sstr{p}{\delta_{0}}$
  with $\sigma^{m} \in \str{p}$.  If
  \begin{equation*}
    \tilde{\epsilon}  \leq  \frac{\Gamma^{2m+1}_{0}}{4},
  \end{equation*}
  and $\delta_0^2 \leq \frac{1}{2}$, then $\Upsilon(\sigma^{m+1})
  < \Gamma^{m+1}_{0}$.
\end{lem}
\begin{proof}	
  By \Lemref{corollary-rel-edge-length-ortho-radius}
  \begin{equation*}
    \Delta(\sigma^{m+1}) \leq \frac{2}{1-\delta^{2}_{0}}
    R_{p}(\sigma^{m}) < 4\samconst,
  \end{equation*}
  since $R_{p}(\sigma^{m}) < \samconst$ and $\delta_0^2 \leq \frac{1}{2}$.
  
  % Using the fact that $\sigma^{m+1} = p_{m+1} * \sigma^{m} \in
  % \sstr{p}{\delta_{0}}$, $L(\sigma^{m+1}) \leq L(\sigma^{m})$ and
  % $R_{p}(\sigma^{m}) < \epsilon$,
  % \begin{eqnarray*}
  %   \Delta(\sigma^{m+1}) &\leq& R_{p}(\sigma^{m}) +
  %   \sqrt{R_{p}(\sigma^{m})^{2} + \delta_{0}^{2}
  %   L(\sigma^{m+1})^{2}}\\ 
  %   &\leq& R_{p}(\sigma^{m}) + \sqrt{R_{p}(\sigma^{m})^{2} +
  %     \delta_{0}^{2} L(\sigma^{m})^{2}}\\ 
  %   &\leq& (2+2\delta_{0}) R_{p}(\sigma^{m})\\ 
  %   &<& 4\epsilon
  % \end{eqnarray*}

  Then from Lemmas \ref{lem-sigma-m+1-fatness} and
  \ref{lem:thick.small.angle} we get
  \begin{align*}
    \Upsilon(\sigma^{m+1})
    &\leq \left( \frac{\Delta(\sigma^{m+1}) }{2\reach} + \sin \theta \right) \\
    &\leq \frac{2\epsilon}{\reach} \left( 1+ \frac{1}{\Gamma^{m}_{0}} \right) \\
    &< \frac{4\epsilon}{\Gamma_{0}^{m} \reach} \quad \text{since
      $\Gamma_{0} < 1$}\\
    &< \Gamma^{m+1}_{0},
  \end{align*}
  from the hypothesis on $\tilde{\epsilon}$.
\end{proof}

We emphasise the consequence of \Lemref{cor-sigma-m+1-flake}:
\begin{cor}
  \label{obs-1-chapter-7}
  If $\delta_0^2 < \frac{1}{2}$ and
  \begin{equation*}
%    \label{eq:sam.empty.cosph} 
    \tilde{\epsilon} \; \leq \; \frac{\Gamma_{0}^{2m+1}}{4},    
  \end{equation*}
  and all the simplices in $\sstr{p}{\delta_{0}}$ are
  $\Gamma_{0}$-good, then
  \begin{equation*}
    \sstr{p}{\delta_{0}} = \emptyset.    
  \end{equation*}
\end{cor}

Now we proceed to the second step of the analysis.  Assuming that
$\sstr{p}{\delta_{0}} = \emptyset$ for all $p$ in $\mpts$, the
following lemma says that if $\splxs \in \str{p}$, then also $\splxs
\in \str{q}$ for every vertex $q \in \splxs$, provided
%$\samconst$ satisfies an additional bound.
the appropriate constraints are met.
\begin{lem}
  \label{lem-protection-star-tangent-space}
  Let $\mpts$ be a $\tsparseconst\epsilon$-sparse $\epsilon$-sample of
  $\man$ with $\tsparseconst \leq 1$  independent of
  $\epsilon$.  We further assume $\delta_0 \leq 1$ and
  \begin{enumerate}
  \item[(1)] for all $p \in \mpts$, every $\sigma^{m} \in \str{p}$ is a
    $\Gamma_{0}$-good simplex with $R_{p}(\sigma^{m}) < \epsilon$, and
  \item[(2)] for all $p \in \mpts$, $\sstr{p}{\delta_{0}} =
    \emptyset$.
    % Hypothesis $\mathcal{H}3$: $\delta_0^2 \leq
    % \flakebnd^{m+1} \leq \frac{1}{2}$. 
  \end{enumerate}
  If
  \begin{equation*}
%    \label{eq:sam.all.stars}
    \tilde{\epsilon} \leq
    \frac{\delta^{2}_{0}\tilde{\mu}^{2}_{0}\Gamma^{m}_{0}}{36},
  \end{equation*}
  then $\str{p} = \str{p; \tancplxmpts}$ for all $p$ in $\mpts$.
\end{lem}
\begin{proof}
  For $p \in \mpts$, let $\sigma^{m} \in \str{p}$ and $q\; (\neq p)$
  be a vertex of $\sigma^{m}$.  We will show that $\sigma^{m}$ is also
  in $\str{q}$.
	
  Let $\theta = \max \angle (\aff(\sigma^{m}), \tanspace{x}{\man})$
  where the max is taken over the vertices $x$ of $\sigma^{m}$.
  % Using the same arguments as in the proof of
  % Lemma~\ref{cor-sigma-m+1-flake} and the fact that
  Since
  \begin{equation*}
  \tilde{\epsilon} \leq
  \frac{\delta_{0}^{2}\tilde{\mu}_{0}^{2}\Gamma^{m}_{0}}{36} < 
  \frac{\Gamma^{m}_{0}}{4},    
%  \frac{\Gamma^{2m+1}_{0}}{4},    
  \end{equation*}
 \Lemref{lem:thick.small.angle} yields
  \begin{equation*}
    \sin  \theta
    \leq  \frac{2\tilde{\epsilon} }{\Gamma^{m}_{0}}
    \stackrel{{\rm def}}{=} c_{1}\tilde{\epsilon} <  \frac{1}{2}.
  \end{equation*}
  It follows that $\cos \theta >{\sqrt{3}} / {2}$ and so
  \begin{equation*}
    \tan \theta \leq 2c_{1}\tilde{\epsilon}.
  \end{equation*}
	
  Recall that $\normhull{\sigma^m}$ denotes the affine space
  orthogonal to $\aff(\sigma^{m})$ and passing through
  $\circcentre{\sigma^{m}}$.
  Let %$p_{i} \in \{p_{1}, \, \dots , \, p_{m}\}$,
  $c$ be the unique point in $\normhull{\sigma^{m}} \cap
  \tanspace{q}{\man}$, and let $R = \gdist_{\reel^{N}}(c, p)$.

  Using the fact that $\angle (\aff(\sigma^{m}),
  \tanspace{q}{\man}) \leq \theta$, we have
  \begin{equation*}
    \gdist_{\reel^{N}}(\circcentre{\sigma^{m}}, c)
    \leq R(\sigma^{m}) \tan \theta \leq  2c_{1}\tilde{\epsilon}
    R(\sigma^{m}), 
  \end{equation*}
  and likewise
  \begin{equation*}
    \gdist_{\reel^{N}}(\circcentre{\sigma^{m}}, c_{p}(\sigma^{m}))
    \leq 2c_{1}\tilde{\epsilon}\, R(\sigma^{m}).
  \end{equation*}
  It follows that $R \leq (1+2c_{1}\tilde{\epsilon}) R(\sigma^{m})$,
  and $\gdist_{\reel^{N}}(c_{p}(\sigma^{m}), c) \leq
  4c_1\tilde{\epsilon} R(\sigma^{m})$.  From the above observations,
  and using the fact that $R(\sigma^{m}) \leq
  R_{p}(\sigma^{m})<\epsilon$, we get
  \begin{eqnarray*}
    \ballamb{c}{R} &\subseteq& \ballamb{c_{p}(\sigma^{m})}{
      (1+6c_{1}\tilde{\epsilon}) R(\sigma^{m})}\\ 
    &\subseteq& \ballamb{c_{p}(\sigma^{m})}{R_{p}(\sigma^{m}) +
      6c_{1}\tilde{\epsilon}\, \epsilon} . 
  \end{eqnarray*}	

  Since $\sstr{p}{\delta_{0}} = \emptyset$, and $\mpts$ is
  $\tsparseconst \samconst$-sparse, we have that $\sigma^{m}$ is
  $\delta_{0}^{2} \tsparseconst^{2}\epsilon^{2}$-power protected on
  $\tanspace{p}{\man}$ (\Lemref{lem:empty.cosph.pwr.protects}). This
  means that
  \begin{equation*}
    \ballamb{c_{p}(\sigma^{m})}{R_{p}(\sigma^{m})+\Delta}
    \cap (\mpts \setminus \sigma^m) = \emptyset,
  \end{equation*}
  where
  \begin{align*}
    \Delta
    &= \sqrt{R_{p}(\sigma^{m})^{2} + \delta_{0}^{2}\tsparseconst^{2}
      \epsilon^{2}} - R_{p}(\sigma^{m})\\ 
    &= \frac{\delta_{0}^{2}\tsparseconst^{2}\epsilon^{2}}
    {\sqrt{R_{p}(\sigma^{m})^{2} + \delta_{0}^{2} \tsparseconst^{2}
        \epsilon^{2}} + R_{p}(\sigma^{m})}\\
    &> \frac{\delta_{0}^{2} \tsparseconst^{2}\epsilon}
    {\sqrt{1 + \delta_{0}^{2} \tsparseconst^{2}} + 1}\\
    &> \frac{\delta_{0}^{2}\tilde{\mu}_{0}^{2}\epsilon}{3} \; \;
    \stackrel{{\rm def}}{=} \;\; c_{2}\epsilon.
  \end{align*}
  % The Delaunay ball $\ballamb{c_{p}(\sigma^{m})}{R_{p}(\sigma^{m})}$
  % is $(\delta_{0}^{2} \tsparseconst^{2} \epsilon^2)$-power protected,
  Since $6c_{1}\tilde{\epsilon} \epsilon \leq c_{2} \epsilon$, by our
  hypothesis on $\tilde{\epsilon}$, we have
  \begin{equation*}
    \ballamb{c}{R} \subset \ballamb{c_{p}(\sigma^{m})}{
      R_{p}(\sigma^{m})+\Delta}, 
  \end{equation*}
  and thus the $m$-simplex $\sigma^{m}$ belongs to $\str{q}$.
\end{proof}
The consequence of Lemma~\ref{lem-protection-star-tangent-space},
together with \Lemref{lem:empty.cosph.pwr.protects} is that every
$m$-simplex in $\tancplxmpts$ has, for each vertex, a $\delta_0^2
\tsparseconst^2 \samconst^2$-power-protected Delaunay ball centred on
the tangent space of that vertex:
\begin{cor}\label{obs-2-chapter-7}
  Let $\mpts$ be a $\tsparseconst \epsilon$-sparse $\epsilon$-sample
  of $\man$ with $\tsparseconst$ being independent of $\epsilon$.
  Under the hypotheses in
  Lemma~\ref{lem-protection-star-tangent-space}, for all $p \in
  \mpts$, all the $m$-simplices $\sigma^{m}$ in $\str{p;
    \tancplx{\mpts}}$ are $\delta_{0}^{2} \tsparseconst^{2}
  \epsilon^{2}$-power protected on $T_{p}\man$. I.e, for all $\splxs^m
  \in \str{p; \tancplxmpts}$ there exists a $c_p(\sigma^m) \in
  \normhull{\sigma^m} \cap \tanspace{p}{\man}$ such that for all $q
  \in \mpts \setminus \sigma^{m}$
  \begin{equation*}
    \gdist_{\reel^{N}}(q, c_{p}(\sigma^{m}))^{2} >
    \gdist_{\reel^{N}}(p, c_{p}(\sigma^{m}))^{2} 
    + \delta_{0}^{2} \tsparseconst^{2} \epsilon^{2}.
  \end{equation*}
\end{cor}
% \begin{proof}
%   Using the facts that $\sstr{p}{\delta_{0}} = \emptyset$ and $\mpts$
%   is $\tilde{\mu}_{0}\epsilon$-sparse we have all the $m$-simplices in
%   $\str{p}$ to be
%   $\delta_{0}^{2}\tilde{\mu}_{0}^{2}\epsilon^{2}$-power protected on
%   $T_{p}\man$.  From Lemma~\ref{lem-protection-star-tangent-space} we
%   have $\str{p; \tancplx{\mpts}} = \str{p}$ and therefore the result
%   follows.
% \end{proof}

%%%%%%%%%%%%%%%
%\ednote{added from right after the inter-point distance proof
%  (\Lemref{lem-2-chapter-7-algorithm-analysis}). We will presumably
%  need some kind of segway: }

We are now in a position to show that Hypothesis $\mathcal{H}5$, when
added to Hypotheses $\mathcal{H}0$ to $\mathcal{H}4$, results in the
output of the algorithm meeting the demands of \Thmref{thm:tan.cplx}.

Recalling that $\tsparseconst = \frac{1}{9}$, Hypotheses
$\mathcal{H}3$ yields the following consequence of ${\mathcal H5}$:
\begin{equation*}
%  \label{eqn:useful-final-bound-on-epsilon}
  \tilde{\epsilon} \leq \frac{\delta^{2}_{0}\,
    \Gamma^{2m}_{0}}{1.1\times 10^{9}} \leq 
    \min \left\{ \frac{\Gamma^{2m+1}_{0}}{4}, \, \frac{\delta_{0}^{2}
        \tsparseconst^{2} \Gamma^{m}_{0}}{36}, \, 
    \frac{\delta^{2}_{0} \tsparseconst^{3} \Gamma_{0}^{2m}}{1.5 \times
      10^{6}}  \right\}.
\end{equation*}
In other words, the sampling radius bounds demanded by
\Corref{obs-1-chapter-7}, \Lemref{lem-protection-star-tangent-space},
and \Thmref{thm:tan.cplx} are all simultaneously satisfied. 
\Corref{cor:preliminary-prop-output} together with
\Corref{obs-1-chapter-7} ensure that the hypotheses of
\Lemref{lem-protection-star-tangent-space} are satisfied, and so it
follows that the $m$-simplices of $\tancplxmpts$ are power-protected
as described by \Corref{obs-2-chapter-7}. Thus all the requirements of
\Thmref{thm:tan.cplx} are satisfied, and we obtain
\Thmref{thm:algo.guarantee}.

\section{Conclusions}

We have described an algorithm which meshes a manifold according to
extrinsic sampling conditions which guarantee that the intrinsic
Delaunay complex coincides with the restricted Delaunay complex, and
that it is homeomorphic to the manifold. The algorithm constructs the
tangential Delaunay complex, which is also shown to be equal to the
intrinsic Delaunay complex, and in this way we are able to exploit
existing structural results~\cite{boissonnat2011tancplx} to obtain the
homeomorphism guarantee.

This approach relies on an embedding of $\man$ in $\amb$. In future
work we aim to develop algorithms and structural results which enable
the construction of an intrinsic Delaunay triangulation in the absence
of an embedding in Euclidean space.

\subsection*{Acknowledgements}

This work was partially supported by the CG Learning project. The
project CG Learning acknowledges the financial support of the Future
and Emerging Technologies (FET) programme within the Seventh Framework
Programme for Research of the European Commission, under FET-Open
grant number: 255827.

\appendix
% -*- LaTeX -*-
% counterex.tex renamed from flake_problem.tex
% 20111201
%

\section{An obstruction to intrinsic Delaunay triangulations}
\label{sec:counter.ex}

When meshing Riemannian manifolds of dimension $3$ and higher using
Delaunay techniques, flake simplices pose problems which cannot be
escaped simply by increasing the sampling density. In particular,
developing an example on a $3$-manifold presented by Cheng et
al.~\cite{cheng2005}, Boissonnat et al.~\cite[Lemma
3.1]{boissonnat2009} show that the restricted Delaunay triangulation
need not be homeomorphic to the original manifold, even with dense
well separated sampling.

In this appendix we develop this example from the perspective of the
intrinsic metric of the manifold. It can be argued that this is an
easier way to visualize the problem, since we confine our viewpoint to
a three dimensional space and perturb the metric, without referring to
deformations into a fourth ambient dimension. This viewpoint also
provides an explicit counterexample to the results announced by Leibon
and Letscher~\cite{leibon2000}: In general the nerve of the intrinsic
Voronoi diagram is not homeomorphic to the manifold. The density
of the sample points alone cannot guarantee the existence of a
Delaunay triangulation.

We explicitly show how density assumptions based upon the strong
convexity radius cannot escape the problem.  The configuration
considered here may be recognised as essentially the same as that
which was described qualitatively in
\Secref{sec:qualitative.counterex}, but here we consider the Voronoi
diagram rather than Delaunay balls.  We work exclusively on a three
dimensional domain, and we are not concerned with ``boundary
conditions''; we are looking at a coordinate patch on a densely
sampled compact $3$-manifold.

\subsection{Sampling density alone is insufficient}

\newcommand{\llrad}[1]{\localconst(#1)}
\newcommand{\llradM}{\localconst(\man)}

We will now construct a more explicit example to demonstrate that the
problem of near-degenerate configurations cannot be escaped with the
kind of sampling criteria proposed by Leibon and
Letscher~\cite{leibon2000}. 

Leibon and Letscher~\cite[p. 343]{leibon2000} explicitly assume that
the points are \defn{generic} which they state as
\begin{de}
  \label{def:ll.delaunay.gen}
  The set $\mpts \subset \man$, is \defn{generic} if $\man$ is an
  $m$-manifold and $m+2$ points never lie on the boundary of a round
  ball. 
\end{de}
Here a round ball refers to a geodesic ball. This definition of
genericity is natural, and corresponds to Delaunay's original
definition~\cite{delaunay1934}, except Delaunay only imposed the
constraint on empty balls. A question that Delaunay addressed
explicitly, but which was not addressed by Leibon and Letscher, is
whether or not such an assumption is a reasonable one to
make. Delaunay showed that any (finite or periodic) point set in
Euclidean space can be made generic through an arbitrarily small
affine perturbation. That a similar construction of a perturbation can
be made for points on a compact Riemannian manifold has not been
explicitly demonstrated. However, in light of the construction we now
present, it seems that the question is moot when $m>2$, because an
arbitrarily small perturbation from degeneracy will not be sufficient
to ensure a triangulation.

Leibon and Letscher proposed adaptive density requirements based upon
the \defn{strong convexity radius}. These requirements are somewhat
complicated, but they will be satisfied if a simple constant sampling
density requirement is satisfied. Exploiting a
theorem~\cite[Thm. IX.6.1]{chavel2006}, that relates the strong
convexity radius to the injectivity radius, $\injradM$, and a positive
bound on the sectional curvatures, they arrive at the following:
\begin{claim}[\protect{\cite[Lemma 3.3]{leibon2000}}]
  \label{claim:ll}
  Suppose $\scurvbnd$ is a positive upper bound on the sectional
  curvatures of $\man$, and
  \begin{equation}
    \label{eq:ll.cond}
   \llradM = \min \left\{ \frac{\injradM}{10},
     \frac{\pi}{10\sqrt{\scurvbnd}} \right\}.
  \end{equation}
  If $\mpts$ is an $\llradM$-sample set for $\man$ with respect
  to $\gdistM$, then  $\carrier{\delMmpts} \cong \man$.
\end{claim}
In fact, we will show that no sampling conditions based on density
alone will be sufficient to guarantee a homeomorphic Delaunay complex
in general, even when a sparsity assumption is also demanded. An
\defn{$\tsamconst$-net} is an $\tsamconst$-sparse, $\tsamconst$-sample
set. We will show:
\begin{thm}
  \label{thm:leibon.wrong}
  With $\llradM$ as defined in \Eqnref{eq:ll.cond}, for any
  $\samconst > 0$, there exists a compact Riemannian manifold $\man$,
  and and a finite set $\mpts \subset \man$, such that $\mpts$ is an
  $(\samconst \llradM)$-net for $\man$, with respect to the metric
  $\gdistM$, but $\delMmpts$ is not homeomorphic to $\man$.
\end{thm}

\subsubsection{A counter-example}

We will construct the counter-example by considering a perturbation of
a Euclidean metric. This is a local operation, and the global
properties of the manifold are only relevant in so far as they affect
$\llradM$ of \Eqnref{eq:ll.cond}. We may assume, for example, that
the manifold is a $3$-dimensional torus $\man \cong \sphere{1} \times
\sphere{1} \times \sphere{1}$, initially with a flat metric. 

Thus assume there is some $\samconst_0$ such that any compact
Riemannian manifold may be triangulated by the intrinsic Delaunay
complex when $\mpts$ is an $ \samconst_0 \llradM$-net.  For
convenience, we choose a system of units so that $\samconst_0
\llradM = 1$. We will first construct a point configuration and
metric perturbation that leads to a problem, and then we will show
that the sampling assumptions are indeed met.

\begin{wrapfigure}{r}{0.35\linewidth}
  \begin{center}
    \includegraphics[width=.8\linewidth]{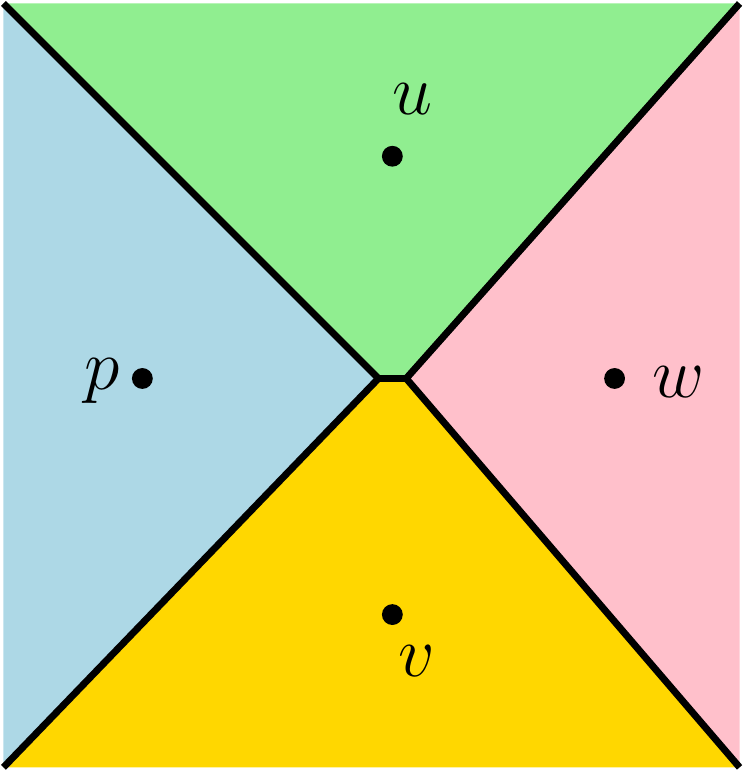} 
  \end{center}
  \caption{A vertical slice: the $xz$-plane of the initial Voronoi
    diagram, seen from the negative $y$ axis.}
  \label{fig:xzvor}
\end{wrapfigure}
We introduce a number of parameters which we will manipulate to
produce the  counter-example. We are exploiting the fact that
the genericity assumption allows configurations that are arbitrarily
close to being degenerate. The assumed $\samconst_0$ has been fixed.

We will work within a coordinate chart on $\man$, where the metric is
Euclidean. We will perturb this metric by constructing a metric tensor
$\tilde{g}$, and we will denote by $\tman$ the manifold with with this
new metric.

Consider points $u, v, w, p$ in the $xz$-plane arranged with $u$ and
$v$ at $\pm a$ on the $z$ axis, and $w$ and $p$ at $\pm (a+\xi)$ on
the $x$ axis, with $a = \frac{3}{4}$, and $0 < \xi < r_0 \gamma$,
where $r_0$ and $\gamma$ will be specified below. The Voronoi diagram
of these points in the $xz$-plane is shown in \Figref{fig:xzvor}.  The
main point here is that the Voronoi boundary between $\vorcellman{u}$ and
$\vorcellman{v}$ may be arbitrarily small with respect to the distance
between the sites, i.e., $\xi$ will be very very small.

The three dimensional Voronoi diagram is the extension of this in the
horizontal $y$-direction, so that every cross-section looks the
same. Note that since the points are not co-circular, they do not
represent a degeneracy by Delaunay's criteria~\cite{delaunay1934}, but
this is irrelevant; we will also argue that the points will not
represent a degenerate configuration with respect to the new metric.

We now introduce a small localized metric perturbation so as to change
the Voronoi diagram near the origin. For example, we can demand that
the matrix of the metric tensor in our coordinate system has the form
\begin{equation*}
\tilde{g}(p) =
  \begin{pmatrix}
    1 - f(\abs{p}) & 0 & 0 \\
    0 & 1 & 0 \\
    0 & 0 & 1
  \end{pmatrix},
\end{equation*}
where $\abs{p}$ is the parametric distance from $p$ to the origin. The
radial function $f$ is non-negative, and it and its first two
derivatives are bounded, e.g.,
\begin{equation}
  \label{eq:bound.derivs}
  f(r),\abs{f'(r)}, \abs{f''(r)} \leq \beta. 
\end{equation}
We also demand that there exists a positive $\gamma \leq \beta$
such that $f(r) \geq \gamma$ when $r \leq r_0$, and that $f(r) = 0$ if
$r \geq 2r_0$. The parameter $r_0$, defines the radius of the ball
bounding the perturbed region. Now we have $\dist{w}{p} < \dist{u}{v}$
when $\xi < r_0\gamma$.

Since $\gamma$ may be arbitrarily small compared to $\beta$,
standard arguments supply a function $f$ meeting these conditions. For
example, the $C^\infty$ construction described by
Munkres~\cite[p. 6]{munkres1968} may be multiplied by a scalar sufficiently
small to meet our needs. 

The vertical $y=0$ cross-section of the perturbed Voronoi diagram will
look something like \Figref{fig:after.xzvor}: $\vorcelltman{p}$ and
$\vorcelltman{w}$ now meet in the $xz$-plane, and $\vorcelltman{u}$ and
$\vorcelltman{v}$ do not. However, since geodesics which do not intersect
the ball $\ballE{0}{2r_0}$ will remain straight lines in the parameter
space, the Voronoi diagram is unchanged outside of a neighbourhood of
the origin. Thus looking from above at the slice of the Voronoi
diagram in the $xy$-plane, we will see something like
\Figref{fig:side.views}\subref{sfig:from.above}.
\Figref{fig:side.views}\subref{sfig:front.view} shows the
$yz$-plane. 

\begin{wrapfigure}{l}{0.35\linewidth}
  \begin{center}
    \includegraphics[width=.8\linewidth]{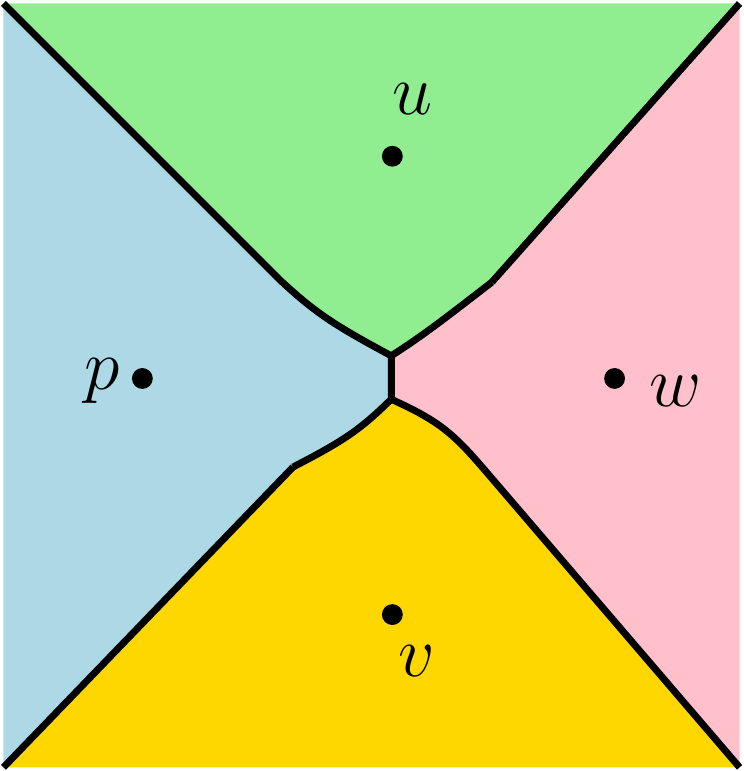} 
  \end{center}
  \caption{The $y=0$ slice of the perturbed Voronoi diagram.}
  \label{fig:after.xzvor}
\end{wrapfigure}
Two Voronoi vertices have been introduced, the red and blue points in
\Figref{fig:side.views}. These are the centres of distinct empty
geodesic circumballs for $\asimplex{p,u,v,w}$. Since they cannot lie
in the region unaffected by the perturbation, a quick calculation
shows that the parametric distance of these Voronoi vertices from the
origin is bounded by $4r_0$, when $r_0 \leq \frac{1}{4}$, and it
follows from another small calculation that the parametric distance
from these Voronoi vertices to any of the four sample points is
bounded by $a( 1 + \frac{3\xi + 16r_0^2}{a^2})$. The distances between
these Voronoi vertices and the sample points in the new metric will
also be subjected to the same bound, since no distances
increase. Also, The sparsity condition will not be affected by the
perturbation. Thus, since we can make $r_0$ as small as we please, and
$\xi$ is chosen such that $\xi < r_0\gamma$, it follows that the
radius of these balls may be made arbitrarily close to $a =
\frac{3}{4} = \frac{3}{4} \samconst_0 \llradM$. We will argue next
that we can make $\abs{ \llradM - \llrad{\tman} }$ as small as desired
by reducing the size of $\beta$ in \Eqnref{eq:bound.derivs}. Then other
sample points may be placed on the manifold so that the density
criteria are met, and no degenerate configuration (violation of
\Defref{def:ll.delaunay.gen}) need be introduced.

This means that the Delaunay complex, defined as the nerve of the
Voronoi diagram, will not be a triangulation of the manifold
$\tman$. As observed by Boissonnat et al.~\cite{boissonnat2009}, the
triangle faces $\asimplex{p,w,u}$ and $\asimplex{p,w,v}$ will be
adjacent to only a single tetrahedron, namely
$\asimplex{p,u,v,w}$. Thus $\del{\tman}{\mpts}$ is not a manifold complex as
defined in \Secref{sec:background}. This is clearly a problem if the
original manifold has no boundary.

\begin{figure}[ht]
  \begin{center}
    \subfloat[$xy$-plane from above]{
      \label{sfig:from.above}
      \includegraphics[width=.8\columnwidth]{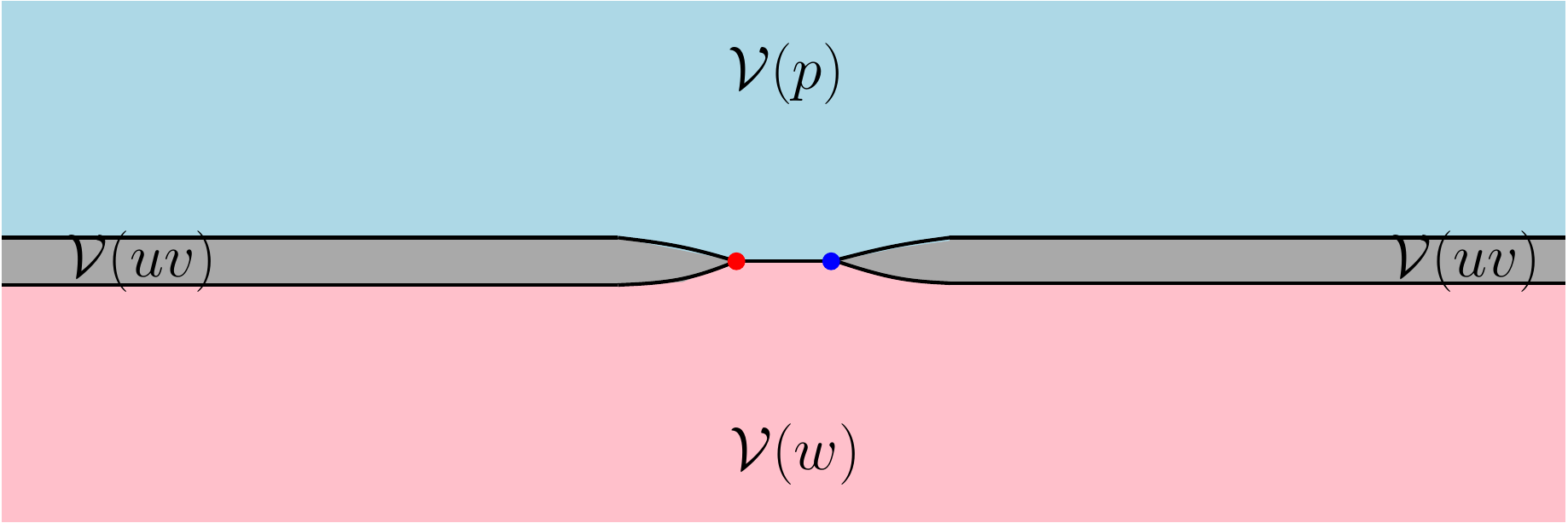} } \\
    \subfloat[$yz$-plane]{
      \label{sfig:front.view}
      \includegraphics[width=.8\columnwidth]{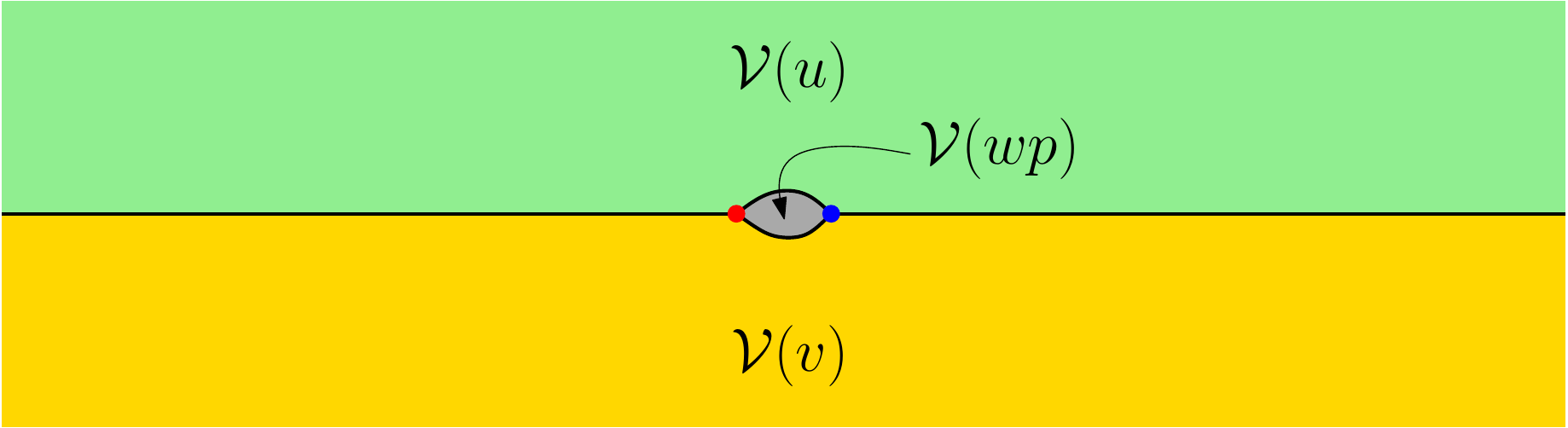} }
  \end{center}
  \caption{Looking at cross-sections; the positive $y$-direction is to
    the right. The four points, $p,u,v,w$, admit two small circumballs
    with distinct centres (the red and blue points).}
  \label{fig:side.views}
\end{figure}

% The problem configuration described here need not lead to a violation of
% \Defref{def:ll.delaunay.gen}. 
Although it is in some sense close to being degenerate, we emphasise
that this configuration represents a problem that cannot be escaped by
an arbitrarily small perturbation of the sample points. An argument
based on the triangle inequality shows that in order to effect a
change in the topology of the Voronoi diagram, a displacement of the
points by a distance of $\Omega(r_0\gamma - \xi)$ is required. 

More specifically, we observe that the configuration
$\asimplex{p,u,v,w}$ may be placed in an otherwise well behaved point
set $\mpts$ such that within a small ball centred at the origin in our
coordinate chart, all points will have $\asimplex{p,u,v,w}$ as the
four closest points in $\mpts$, and this would remain the case even if
the point positions were perturbed a small amount.  We may further
assume that the other Delaunay simplices are well shaped, so that
stability results \cite{boissonnat2012stab1} can be used to argue
that they cannot be destroyed with an \emph{arbitrarily} small
perturbation. Then we argue that in order to obtain a triangulation by
a perturbation $\mpts \to \mpts'$, we must ensure that the Voronoi
cell $\vorcelltman{\asimplex{p',w'}}$ must vanish: the edge
$\asimplex{p',w'}$ will never be incident to any tetrahedron other
than $\asimplex{p',u',v',w'}$. Then an argument based on the triangle
inequality shows that for a $\pertconst$-perturbation with $\pertconst
< \frac{r_0\gamma - \xi}{6}$, there will be a point in
$\vorcelltman{\asimplex{p',w'}}$ within a distance of $2\pertconst$ of
the origin.
% See CF12p6 and also p70.

\subsubsection{The sizing function under perturbation}

We need to establish that the metric manipulation that we performed in
order to construct the counter-example, does not have a dramatic
effect on the sizing function $\llrad{\tman}$. This follows from the
fact that we have bounded $g - \tilde{g}$ together with its first and
second deriviatives.

Since the sectional curvature may be described as a continuous
function of $g$ and it's first and second derivatives \cite[pp. 56 \&
93]{docarmo1992}, the effect of our perturbation on the sectional
curvatures can be made arbitrarily small by reducing $\beta$ in
\Eqnref{eq:bound.derivs}.

Since we started with a flat metric anyway, the bound $\scurvbnd$ can
be made arbitrarily small, and so the second term in
\Eqnref{eq:ll.cond} will not be the smallest. We need to bound the
change in the injectivity radius as well.

This follows from results in the literature
\cite{ehrlich1974,sakai1983}, which state that for a compact manifold,
$\injradM$ depends continuously on the metric and its first and second
derivatives. Specifically,
\begin{lem}[Ehrlich]
  Let $\mathfrak{M}$ be the space of $C^3$ Riemannian metric
  structures $g$ on a compact manifold $\man$, and endow
  $\mathfrak{M}$ with the $C^2$ topology. The function $g \mapsto
  \injr_g(\man)$ is continuous in this topology.
\end{lem}
This means that for any desired bound on $\abs{\llradM -
  \llrad{\tman}}$, there will be a $\beta$ that will satisfy the
bound.

The construction of the counter-example is complete.

\subsection{Discussion}

We have shown that for constructing a Delaunay triangulation for an
arbitrary Riemannian manifold, a sampling density requirement is not
sufficient in general. The solution we propose in the body of this
paper, is to constrain the kind of sample sets that we consider.
Another approach would be to constrain the kind of metrics that are
assumed. However, even with a purely Euclidean metric, allowing
configurations to be arbitrarily close to degeneracy means that
arbitrarily poorly shaped simplices are to be expected.  When the
metric is no longer Euclidean, the ``shape'' of a simplex no longer
has an obvious meaning, but the problems associated with point
configurations near degeneracy will certainly be present.

Our analysis relied on the ability to make the support of the
perturbation small. This is unlikely to be a necessary feature of the
construction, but it facilitates our simplistic analysis.

Clarkson~\cite{clarkson2006} remarked that an implication of Leibon
and Letscher's claim \cite{leibon2000} is that for four points close
enough together, there is a unique circumsphere with small radius. Our
counter-example shows that circumcentres need not be unique under
these conditions. In fact the existence of unique circumcentres does
not follow from the triangulation result: In our work we do not claim
that the $m$-simplices have a unique circumcentre in the intrinsic
metric. However, the argument sketched out by Leibon and Letscher
claimed that the intrinsic Voronoi diagram is a cell complex (i.e., it
satisfies the \defn{closed ball property} \cite{edelsbrunner1997rdt}),
and this does imply unique circumcentres for the top dimensional
simplices.

It is worth emphasising that the problems discussed here only arise
when the dimension is greater than $2$. The same sampling criteria for
two dimensional manifolds has been fully
validated~\cite{leibon1999,dyer2008sgp}, however these works both
assume genericity in the sample set, without demonstrating that it is
a reasonable assumption.

% In a manifold of dimension $m$, the intersection of sufficiently small
% geodesic $(m-1)$-spheres that intersect transversely is necessarily a
% topological $(m-2)$-sphere \cite[Lemma 5]{dyer2008sgp}. In the case of
% Delaunay spheres which circumscribe $m$-simplices with a common face

% -*- LaTeX -*-
% man_bg.tex
% 20120829
% 
% Background results for manifolds -- appendix

%\subsection{Background results for manifolds}
\section{Background results for manifolds}
\label{sec:manifold.background}

The tangent space at $p \in \man$ is denoted $\tanspace{p}{\man}$, and
we identify it with an $m$-flat in the ambient space. The normal
space, $\normspace{p}{\man}$, is the orthogonal complement of
$\tanspace{p}{\man}$ in $\tanspace{p}{\amb}$, and we likewise treat it
as the affine subspace of dimension $m-k$ orthogonal to
$\tanspace{p}{\man} \subset \amb$.

A ball $B=\ballamb{c}{r}$ is a \defn{medial ball} at $p$ if $B \cap
\man = \emptyset$, it is tangent to $\man$ at $p$, and it is maximal
in the sense that any ball which 
% is centred on the line through $p$ and $c$, and 
contains $B$ either coincides with $B$ or intersects
$\man$.  The \defn{local reach} at $p$ is the infimum of the radii of
the medial balls at $p$, and the \defn{reach} of $\man$, denoted
$\reach$, is the infimum of the local reach over all points of
$\man$. In order to approximate the geometry and topology with a
simplical complex, manifolds with small reach require a higher
sampling density than those with a larger reach. As is typical, an
upper bound on our sampling radius will be proportional to $\reach$.
Since $\man \subset \amb$ is a smooth, compact embedded submanifold,
it has positive reach.

An estimate of how the tangent space locally deviates from the
manifold is given by an observation of Federer~\cite[Theorem
4.8(7)]{federer1959} (see also Giesen and Wagner~\cite[Lemma
6]{giesen2004}):
\begin{lem}[Distance to tangent space]
  \label{lem:dist.to.tanspace}
  If $x,y \in \man \subset \amb$ and $\distamb{x}{y} \leq r < \reach$,
  then $\distamb{y}{\tanspace{x}{\man}} \leq \frac{r^2}{2\reach}$, and
  thus $\sin \alpha \leq \frac{r}{2\reach}$, where $\alpha$ is the
  angle between $\seg{x}{y}$ and $\tanspace{x}{\man}$.
\end{lem}
A complementary result bounds the distance to the manifold from a
point on a tangent space 
%\cite[Lemma 7]{giesen2004}:
\cite[Lemma~4.3]{boissonnat2010meshing}:
% \begin{lem}[Distance to manifold]
%   \label{lem:dist.tan.to.man}
%   Suppose $v \in \tanspace{x}{\man}$ with $\norm{v-x} = r \leq
%   \frac{\reach}{4}$. If $\hat{v} \in \man$ is the closest point in
%   $\man$ to $v$, then $\distamb{v}{\hat{v}} < \frac{2r^2}{\reach}$.
% \end{lem}
% \rd{SNAP -- to merge with above:} The proof of this lemma
% \rd{(\Lemref{lem-dist-psi-p-x-pp})} employs the following
% extension~\cite[Lemma~4.3]{boissonnat2010meshing} of
% \Lemref{lem:dist.tan.to.man}:
\begin{lem}[Distance to manifold]
  \label{lem:dist.tan.to.man}
  \label{lem:dist.tan.to.psip}
  Suppose $v \in \tanspace{x}{\man}$ with $\norm{v-x} = r \leq
  \frac{\reach}{4}$.  Let $y = \psi_x(v) \in \man$, where $\psi_x$ is
  the inverse projection~\eqref{eq:defn.psip}.  Then, $
  \distamb{v}{y} \leq \frac{2r^2}{\reach} $.
\end{lem}

The previous two lemmas lead to a convenient bound on the angle
between nearby tangent spaces. We prove here a variation on previous
results \cite[Prop. 6.2]{niyogi2008} \cite[Lemma
5.5]{boissonnat2011tancplx}: 
\begin{lem}[Tangent space variation]
  \label{lem:tangent.variation}
  Let $x,y \in \man$ be such that $\distamb{x}{y} = r \leq
  \frac{\reach}{4}$, and let $\alpha$ be the angle between
  $\tanspace{x}{\man}$ and $\tanspace{y}{\man}$. Then,
%  \begin{equation*}
$
    \sin \alpha < \frac{6r}{\reach}. 
$
% \end{equation*}
\end{lem}
\begin{proof}
  Let $v \in \tanspace{y}{\man} \subset \amb$ with $\norm{v-y} =
  r$. We will bound the angle between $v-y$ and $\tanspace{x}{\man}$.
%  Let $\tilde{v}$ be $v$ considered as a point in $\amb$.
  We have
  \begin{equation}
    \label{eq:bnd.alpha}
    \begin{split}
      \sin \alpha &\leq
      \frac{1}{\norm{v-y}} \left( \distamb{y}{\tanspace{x}{\man}} +
      \distamb{v}{\tanspace{x}{\man}} \right)\\ 
      &\leq
      \frac{1}{\norm{v-y}} \left(\distamb{y}{\tanspace{x}{\man}} +
      \distamb{v}{\hat{v}} +
      \distamb{\hat{v}}{\tanspace{x}{\man}} \right), 
    \end{split}
  \end{equation}
  where $\hat{v} \in \man$ is the closest point to $v$ in $\man$.

  By \Lemref{lem:dist.to.tanspace}, we have
  $\distamb{y}{\tanspace{x}{\man}} \leq \frac{r^2}{2\reach}$, and by
  \Lemref{lem:dist.tan.to.man} we get $\distamb{v}{\hat{v}}
  \leq \frac{2r^2}{\reach}$. For the third term in
  \Eqnref{eq:bnd.alpha}, we find 
  \begin{equation*}
    \begin{split}
      \distamb{x}{\hat{v}} &\leq
      \distamb{x}{y} + \norm{v-y} + \distamb{v}{\hat{v}}\\
      &\leq 2r + \frac{2r^2}{\reach} \leq \frac{5r}{2}\\
      &< \reach,
    \end{split}
  \end{equation*}
  and so we may apply \Lemref{lem:dist.to.tanspace} to obtain
  $\distamb{\hat{v}}{\tanspace{x}{\man}} \leq
  \frac{25r^2}{8\reach}$.

  Putting these observations back into \Eqnref{eq:bnd.alpha} we find
  \begin{equation*}
    \sin \alpha \leq \frac{1}{\norm{v-y}} \left( \frac{r^2}{2\reach} +
      \frac{2r^2}{\reach} + \frac{25r^2}{8\reach} \right) =
    \frac{45r}{8\reach} < \frac{6r}{\reach}. 
  \end{equation*}
\end{proof}

The following observation is a direct consequence of results
established by Niyogi et al. \cite[Lemma 5.4]{niyogi2008}:
\begin{lem}
  \label{lem:tanspace.proj.diffeo}
  Let $W = \ballRM{p}{r}$, for some $p \in \man$ and $r <
  \reach/2$. When restricted to $W$, the orthogonal projection $\pi_p
  |_{W}: W \to \tanspace{p}{\man}$ is a diffeomorphism onto its image.
\end{lem}
\begin{proof}
  Let $f = \pi_p|_{W}$. Niyogi et al. showed \cite[Lemma
  5.4]{niyogi2008} that the Jacobian of $f$ is nonsingular on $W$, so
  that $W$ is a covering space for $U = f(W) \subset
  \tanspace{p}{\man}$. The Morse-theory argument of Boissonnat and
  Chazals~\cite[Proposition 12 ]{boissonnat2001} can be applied to
  demonstrate that $W$ is a topological ball. It follows that $U$ is
  connected, since any path in $W$ projects to a path in $U$. Thus $W$
  must be a single-sheeted cover of $U$, since $f^{-1}(0) =
  \{p\}$. Indeed, if $q \in W$ with $q \neq p$ and $f(q) = 0$, then
  $\seg{p}{q}$ would be perpendicular to $\tanspace{p}{\man}$,
  contradicting \Lemref{lem:dist.to.tanspace}.
 %
 % the ambient ball of radius $\distamb{p}{q}/2 < \reach/2$ centred at
 %  the midpoint of $\seg{p}{q}$ would be tangent to $\man$ at $p$ and
 %  would contain two distinct points of $\man$ on its boundary,
 %  contradicting the definition of $\reach$.
  Thus $f:W \to U$ is a diffeomorphism.
\end{proof}

Niyogi et al~\cite[Prop 6.3]{niyogi2008} demonstrate a bound on the
geodesic distance between nearby points, with respect to the ambient
distance. We will use a modified statement of this result:
\begin{lem}[Geodesic distance bound]
  \label{lem:geodesic.distance}
  Let $x,y \in \man$ be such that $\distamb{x}{y} \leq
  \frac{\reach}{2}$. Then
  \begin{equation*}
    \distM{x}{y} \leq \distamb{x}{y} \left( 1 +
      \frac{2\distamb{x}{y}}{\reach} \right).
  \end{equation*}
\end{lem}
\begin{proof}
  The announced result states
  \begin{equation*}
    \distM{x}{y} \leq \reach \left( 1 - \sqrt{1 -
        \frac{2\distamb{x}{y}}{\reach}} \right). 
  \end{equation*} 
  under the same hypothesis on $x$ and $y$. Rearranging, we have
  \begin{equation*}
    \distM{x}{y} \leq \frac{2\distamb{x}{y}}{1 + \sqrt{1 -
        \frac{2\distamb{x}{y}}{\reach}}}
      \leq
      \frac{\distamb{x}{y}}{1 - \frac{\distamb{x}{y}}{\reach}}
      \leq
      \distamb{x}{y} \left( 1 + \frac{2\distamb{x}{y}}{\reach} \right),
  \end{equation*}
  where the second inequality is obtained by squaring away the radical.
\end{proof}

% -*- LaTeX -*-
% forbidden_vol.tex
% 20120808
%
% was new-appendix-vol-forbidden-region.tex

\section{Forbidden volume calculation}
%\section{Proof of Lemma~\ref{lem-bound-flake-forbidden-volume}}
\label{app-proof-lem-bound-flake-forbidden-volume}

In this appendix we demonstrate:

\noindent
\textbf{\Lemref{lem-bound-flake-forbidden-volume}(Volume of forbidden
  region)}
  Let $\sigma$ be a $k$-simplex with vertices on $\man$ and $k \leq
  m$.  If
  \begin{enumerate}
  \item $ \Gamma_{0} \leq \frac{1}{B+1}$,
  \item $\tilde{\epsilon} \leq \min \{ \frac{\xi}{4\beta\, \reach}, \,
    \frac{\Gamma^{m+1}_{0}}{8\beta} \}$ and
  \item $\delta_{0}^{2} \leq \min \{ \Gamma_{0}^{m+1}, \frac{1}{4}\}$,
  \end{enumerate}
  then
  \begin{equation*}
    \vol(F(\sigma, t)) \leq D\, \Gamma_{0}\, R(\sigma)^{m},
  \end{equation*}
  where $D$ depends on $m$ and $\beta$. %\ednote{and not $\reach$?}
\vspace{\baselineskip}

We will use the following lemmas in the proof of
\Lemref{lem-bound-flake-forbidden-volume}:
\begin{lem}[Triangle altitude bound]\label{lem-prop-triangle}
    For any non-degenerate triangle $\sigma = [p, q, r]$, we have
    $$
        D(p, \sigma) = \frac{\|p-q\| \|p-r\|}{2R(\sigma)} .
    $$
\end{lem}
\begin{proof}
    Let $\alpha = \angle prq$ and observe that
    $$
        \sin \alpha = \frac{\|p-q\|}{2R(\sigma)} .
    $$
    Since
    $D(p, \sigma) = \|p-r\| \sin \alpha$, the result follows.
\end{proof}

\begin{lem}
  \label{lem-modified-torus-lemma}
  Let $\sigma = [p_{0}\, \dots \, p_{k}] \subset \amb$ be a
  $k$-simplex with $1 \leq k \leq m < \ambdim$. Suppose $p_{k+1}
  \in \amb$ is such that $\sigma_{1} = p_{k+1} * \sigma$ admits an
  elementary weight function $\omega_{\sigma_{1}}:
  \mathring{\sigma_{1}} \rightarrow [0, \infty)$, and the following
  conditions are satisfied:
  \begin{enumerate}[label={(\arabic*)}]
  \item \label{tor:shortedge.bnd} $L(\sigma_{1}) > \frac{t}{9}$,
  \item \label{tor:rad.bnd} $R(\sigma_{1}, \omega_{\sigma_{1}}) < \beta t$,
  \item \label{tor:is.flake} $\sigma_{1}$ is a $\Gamma_{0}$-flake, and
  \item \label{tor:dno.bnd} $\delta^{2}_{0} \leq \min \{ \Gamma^{m+1}_{0}, \, \frac{1}{4}\}$.
  \end{enumerate}
  Then
  \begin{equation*}
    d_{\amb}(p_{k+1}; \partial S') \leq B\, \Gamma_{0} R(\sigma)
  \end{equation*}
  where $S' = B_{\amb}(C(\sigma), R(\sigma)) \cap \aff(\sigma)$ and
  \begin{equation*}
  B \; \stackrel{{\rm def}}{=} 4 +  96\beta( 1 + 2^{7}3^{2}\beta^{2}).
  \end{equation*}
\end{lem}
\begin{proof}
  Let $\omega_{\sigma} = \omega_{\sigma_{1}}
  \mid_{\mathring{\sigma}}$. Note that $\omega_{\sigma}:
  \mathring{\sigma} \rightarrow [0, \infty)$ is an elementary weight
  function, and $C(\sigma, \omega_{\sigma})$ is the orthogonal
  projection of $C(\sigma_{1}, \omega_{\sigma_{1}})$ onto
  $\aff(\sigma)$.
	
%  We define
%  \begin{equation*}
%    \cos t \stackrel{{\rm def}}{=}
%    \frac{R(\sigma, \omega_{\sigma})}{R(\sigma_{1}, \omega_{\sigma_{1}})}
%  \end{equation*}
	
  From Lemma~\ref{claim-rel-edge-length-ortho-radius}~(2) and the fact
  that $\delta^{2}_{0} \leq \frac{1}{4}$, we have
	\begin{equation}\label{eqn-55-lem-modified-torus-lemma}
		\Delta(\sigma_{1}) \leq \frac{2}{1-\delta^{2}_{0}} R(\sigma_{1}, \omega_{\sigma_{1}})
		< \frac{8}{3} R(\sigma_{1}, \omega_{\sigma_{1}}) .
	\end{equation}
	and
	\begin{align}\label{eqn-10-lem-modified-torus-lemma}
		\frac{R(\sigma, \omega_{\sigma})}{R(\sigma_{1}, \omega_{\sigma_{1}})}
		&\geq \frac{(1-\delta^{2}_{0}) \Delta(\sigma)}{2R(\sigma_{1}, \omega_{\sigma_{1}})}&
        \text{from Lemma~\ref{claim-rel-edge-length-ortho-radius}~(2)}\nonumber\\
		&\geq \frac{3 L(\sigma)}{8 R(\sigma_{1}, \omega_{\sigma_{1}})}&
        \text{as $\delta_{0}^{2}\leq \frac{1}{4}$ and $L(\sigma) \leq \Delta(\sigma)$}\nonumber\\
		&\geq \frac{1}{24\beta}&
	\end{align}
%	
%	Therefore from Eq.~\eqref{eqn-10-lem-modified-torus-lemma} we have
%	\begin{equation}\label{eqn-111-lem-modified-torus-lemma}
%		\frac{1}{1-\sin t} \leq \frac{1}{1-\sqrt{1-\frac{1}{24^{2}\beta^{2}}}} < 2\times 24^{2} \beta^{2}
%	\end{equation}
%	
	Therefore, from Lemma~\ref{lem:thin.flake.alt.bnd}, we have
	\begin{align}\label{eqn-6-lem-modified-torus-lemma}
		\frac{D(p_{k+1}, \sigma_{1})}{\Delta(\sigma)} &< \left( 1 + \frac{1}{k} \right) \Gamma_{0} \times
		\frac{\Delta(\sigma_{1})^{2}}{L(\sigma_{1}) \Delta(\sigma)}&\nonumber\\
		&\leq 2 \Gamma_{0} \times \frac{\Delta(\sigma_{1})^{2}}{L(\sigma_{1})^{2}}&
		\text{from $k \geq 1$ and $L(\sigma_{1}) \leq \Delta(\sigma)$}\nonumber\\
		&< \frac{128\Gamma_{0}}{9} \times \frac{R(\sigma_{1}, \omega_{\sigma_{1}})^{2}}{L(\sigma_{1})^{2}}&
		\text{from Eq.~\eqref{eqn-55-lem-modified-torus-lemma}}
		\nonumber\\
		&< 2^{7} 3^{2} \beta^{2} \times \Gamma_{0}&
		\text{from hyp.~\ref{tor:shortedge.bnd}~\&~\ref{tor:rad.bnd}}
	\end{align}
	
	Let $p$ be the point closest to $p_{k+1}$ in $\partial B_{\amb}(C; R)$ where $C= C(\sigma_{1}, \omega_{\sigma_{1}})$
	and $R = R(\sigma_{1}, \omega_{\sigma_{1}})$. We have
	\begin{equation}\label{eqn-69-lem-modified-torus-lemma}
		\| p  -  p_{k+1} \| = \sqrt{R^{2} + \omega_{\sigma_{1}}(p_{k+1})^{2}} - R
		\leq \omega_{\sigma_{1}}(p_{k+1}) \leq \delta_{0} L(\sigma_{1})
	\end{equation}
	
	Let $q$ be the point closest to $p$ on $\partial B_{\amb}(C; R) \cap \aff(\sigma)$, $p'$ be the
	projection of $p$ onto $\aff(\sigma)$, and and let $r$ denotes the
	intersection of the line $\aff([q\, C(\sigma, \omega_{\sigma})])$ with $\partial B_{\amb}(C; R)$.
	Note that $C(\sigma_{1}, \omega_{\sigma_{1}})$, $C(\sigma, \omega_{\sigma})$, $p_{k+1}$, $p$, $p'$,
	$q$ and $r$ lie on the same $2$-dimensional affine space.
	
	Using the fact that $\|p - p_{k+1}\| \leq \delta_{0} L(\sigma_{1})$, we get
	\begin{equation}\label{eqn-4-lem-modified-torus-lemma}
		\|p-p'\| \leq D(p_{k+1}, \sigma_{1}) + \delta_{0} L(\sigma_{1})
	\end{equation}
	
	We will now consider the triangle $\sigma_{2} = [p\, q\, r]$. Note that $C(\sigma_{1}, \omega_{\sigma_{1}})$,
	$R(\sigma_{1}, \omega_{\sigma_{1}})$ are the circumcenter and radius of $\sigma_{2}$ respectively. Also,
	$C(\sigma, \omega_{\sigma})$ is the midpoint of the line segment $[q\, r]$ with
	$2R(\sigma, \omega_{\sigma}) = \|q-r\|$ and $D(p, \sigma_{2}) = \|p - p'\|$.
    From the definition of $q$, we have $\|p-r\| \geq \|p-q\|$. Using the
    fact $\|q-r\| = 2R(\sigma, \omega_{\sigma})$, we have
    $$
        \|p-r\| \geq \frac{\|q-r\|}{2} = R(\sigma, \omega_{\sigma}) .
    $$
    This implies from Lemma~\ref{lem-prop-triangle}
    \begin{align}\label{eqn-67-lem-modified-torus-lemma}
        \|p-q\| &= \frac{2R(\sigma_{2}) D(p, \sigma_{2})}{\|p-r\|}& \nonumber\\
        &\leq \frac{2R(\sigma_{1}, \omega_{\sigma_{1}}) D(p, \sigma_{2})}{R(\sigma, \omega_{\sigma})}&
        \mbox{as $R(\sigma_{2}) \leq R(\sigma_{1}, \omega_{1})$ and $\|p-r\| \geq R(\sigma, \omega_{\sigma})$} \nonumber \\
        &\leq 48\beta D(p, \sigma_{2}) = 48\beta \|p-p'\|& \mbox{as Eq.~\eqref{eqn-10-lem-modified-torus-lemma}}
    \end{align}

%	Applying Lemma~\ref{lem-prop-triangle}, we get
%	\begin{equation}\label{eqn-67-lem-modified-torus-lemma}
%		\| p- q\| \leq \left(1+ \frac{1}{1-\sin t } \right) D(p, \sigma_{2})
%        = \left(1+ \frac{1}{1-\sin t } \right) \|p - p'\|
%	\end{equation}
%	where
%	\begin{equation*}%\label{eqn-11-lem-modified-torus-lemma}
%	\cos t = \frac{R(\sigma, \omega_{\sigma})}{ R(\sigma_{1}, \omega_{\sigma_{1}})}
%	\end{equation*}
%	and from Eq.~\eqref{eqn-111-lem-modified-torus-lemma} we have
%	\begin{equation}\label{eqn-11-lem-modified-torus-lemma}
%		1 + \frac{1}{1-\sin t} < 1+ 2^{7}3^{2}\beta^{2}
%	\end{equation}
	
	From Eq.~\eqref{eqn-69-lem-modified-torus-lemma}, \eqref{eqn-4-lem-modified-torus-lemma}
	and \eqref{eqn-67-lem-modified-torus-lemma}
	\begin{align}\label{eqn-5-lem-modified-torus-lemma}
		\|p_{k+1} - q\| &\leq \|p_{k+1}-p\| + \|p-q\|& \nonumber \\
%		\Delta =
		&\leq  \delta_{0} L(\sigma_{1}) + 48\beta
		(D(p_{k+1}, \sigma_{1}) + \delta_{0}L(\sigma_{1}))& %\left(1+ \frac{1}{1-\sin t } \right)&
        \nonumber\\
		&\stackrel{{\rm def}}{=} \eta_{1}&
	\end{align}
	
	Using the fact that $\sigma$ is $\Gamma^{k}_{0}$-thick (since
        $\sigma_{1}$ is a $\Gamma_{0}$-flake), and the bound
        $\delta_0^2\shortedge{\sigma_1}^2$ on the differences of the
        squared distances between $C(\sigma, \omega_{\sigma})$ and the
        vertices of $\sigma$, we obtain a bound \cite[Lemma
        4.1]{boissonnat2012stab1} on the distance from $C(\sigma,
        \omega_{\sigma})$ to $C(\sigma)$:
	\begin{align}\label{eqn-1-lem-modified-torus-lemma}
          \|C(\sigma) - C(\sigma, \omega_{\sigma})\| &\leq
          \frac{\delta_{0}^{2}L(\sigma_{1})^{2}}{2\Upsilon(\sigma)
            \Delta(\sigma)}&
          \nonumber\\
          &\leq \frac{\delta^{2}_{0} R(\sigma)}{\Upsilon(\sigma)}&
          \text{as $L(\sigma_{1}) \leq \Delta(\sigma) \leq 2 R(\sigma)$} \nonumber \\
          &\leq \frac{\delta^{2}_{0} R(\sigma)}{\Gamma_{0}^{k}}&
          \text{since $\sigma$ is $\Gamma^{k}_{0}$-thick,
            $\Upsilon(\sigma) \geq \Gamma^{k}_{0}$}\nonumber \\
          &\leq \frac{\delta^{2}_{0} R(\sigma)}{\Gamma_{0}^{m}}&
          \text{as $\Gamma_{0} \leq 1$} \nonumber\\
          &\stackrel{{\rm def}}{=} \eta_{2}&
	\end{align}
	
	Since $k \geq 1$, there exists
	$ p_{i} \in \mathring{\sigma}$ such that
	$$
    p_{i} \in
	B_{\amb}(C(\sigma), R(\sigma)) \cap
	B_{\amb}(C(\sigma, \omega_{\sigma}), R(\sigma, \omega_{\sigma})) \cap \aff(\sigma).
	$$
	Also, $\| C(\sigma) - p_{i} \| = R(\sigma)$ and
	$\|C(\sigma, \omega_{\sigma}) - p_{i}\| = R(\sigma, \omega_{\sigma})$.
	
	Using the facts that $R(\sigma) = \|C(\sigma) - p_{i}\|$ and
	$R(\sigma, \omega_{\sigma}) = \| C(\sigma, \omega_{\sigma}) - p_{i}\|$,
	and the Triangle inequality, we get
	\begin{eqnarray}\label{eqn-2-lem-modified-torus-lemma}
		R(\sigma)  - \|C(\sigma) - C(\sigma, \omega_{\sigma}) \| \; \; \leq
		& \| C(\sigma, \omega_{\sigma}) - p_{i}\| & \leq \; \; R(\sigma) + \|C(\sigma) - C(\sigma, \omega_{\sigma}) \| \nonumber \\
		R(\sigma) -\eta_{2} \; \; \leq
		&R(\sigma, \omega_{\sigma})& \leq \; \;  R(\sigma) +\eta_{2}
	\end{eqnarray}
	The last equation follows from Eq.~\eqref{eqn-1-lem-modified-torus-lemma}.
	
	Let $S'$ and $S$ denote $B_{\amb}(C(\sigma), R(\sigma)) \cap \aff(\sigma)$ and
	$B_{\amb}(C(\sigma, \omega_{\sigma}), R(\sigma, \omega_{\sigma}))$ respectively.
	From Eq.~\eqref{eqn-1-lem-modified-torus-lemma} and \eqref{eqn-2-lem-modified-torus-lemma},
	we have $d_{\amb}(\partial S', \partial S) \leq \eta_{1} + 2 \eta_{2}$.
	This implies that
	there exists $q' \in \partial S'$ such that
	 \begin{equation}\label{eqn-3-lem-modified-torus-lemma}
	 	\|q' - q  \| \leq 2\eta_{2}.
	 \end{equation}

	Therefore from Eq.~\eqref{eqn-5-lem-modified-torus-lemma} and
	\eqref{eqn-3-lem-modified-torus-lemma}, we get
	\begin{eqnarray*}
		\|p_{k+1} - q'\| \leq \|p_{k+1} - q\| + \|q' - q\| \leq \eta_{1} + 2 \eta_{2}
	\end{eqnarray*}
	
	Using the facts that $\delta^{2}_{0} \leq \Gamma_{0}^{m+1}\leq \Gamma^{2}_{0}$
	(from hyp.~\ref{tor:dno.bnd} of the lemma and $\Gamma_{0}\leq 1$),
	$L(\sigma_{1}) \leq L(\sigma) \leq \Delta(\sigma) \leq 2R(\sigma)$ and
    $\frac{D(p_{k+1}, \sigma_{1})}{\Delta(\sigma)} \leq 2^{7}3^{2}\beta^{2} \,\Gamma_{0}$
    (from Eq.~\eqref{eqn-6-lem-modified-torus-lemma}),
%    and
%	$\left( 1+\frac{1}{1-\sin t} \right) \geq 1+2^{7}3^{2}\beta^{2}$
%	(from Eq.~\eqref{eqn-11-lem-modified-torus-lemma}),
    and Eq.~\eqref{eqn-5-lem-modified-torus-lemma} and \eqref{eqn-1-lem-modified-torus-lemma}, we get
%    can show
%	$\Phi_{0}$ such that
	\begin{eqnarray*}
		d_{\amb}(p_{k+1}; \partial S') & \leq &  \|p_{k+1} - q'\| \\%\; \leq \; B\, \Gamma_{0} R(\sigma)
        &\leq& \eta_{1} + 2 \eta_{2} \\
        &\leq&  \delta_{0}L(\sigma_{1})+ 48 \beta \left( D(p_{k+1}, \sigma_{1}) + \delta_{0}L(\sigma_{1})  \right) %\left(1+\frac{1}{1-\sin t}\right)
        + \frac{2\delta^{2}_{0} R(\sigma)}{\Gamma^{m}_{0}}\\
        &\leq& B\Gamma_{0} R(\sigma)
	\end{eqnarray*}
	where
	$$
		B = 4 +  96\beta( 1 + 2^{7}3^{2}\beta^{2}) .
	$$
\end{proof}

We will use the following lemma from~\cite{boissonnat2010meshing} to
bound the volume of $F(\sigma)$.
\begin{lem}\label{lem-volume-lemma}
  Let $p$ be a point on $\man$. There exists $\xi$ that depends on
  $\reach$ and $m$, and $A$ that depends only on $m$ such that, for
  all $r = t \leq \xi$, we have
  \begin{equation*}
    0 <1-A \,\tilde{t} \leq
    \frac{\vol(B_{\reel^{N}}(p, r)\cap
      \man)}{\ballvolm r^{k}} \leq 1+A\, \tilde{t}
  \end{equation*}
  where $\ballvolm$ is the volume of the $m$-dimensional unit
  Euclidean ball.
\end{lem}

\begin{proof}[of Lemma~\ref{lem-bound-flake-forbidden-volume}]
  For the rest of the proof we define
  \begin{equation*}
    \tilde{t} = \frac{t}{\reach}
  \end{equation*}
	
  Consider the following elementary weight function: $\omega_{\sigma}
  = \omega_{\sigma_{1}} \mid_{\mathring{\sigma}}$.  Using the facts
  that $R(\sigma, \omega_{\sigma}) \leq R(\sigma_{1},
  \omega_{\sigma_{1}})$, $R(\sigma, \omega_{\sigma}) < \beta t\,
  \reach$, and Lemma~\ref{claim-rel-edge-length-ortho-radius}~(3)
  \begin{align}\label{eqn-radius-sigma-beta}
    R(\sigma) & \leq  R(\sigma, \omega_{\sigma}) \left( 1 -
      \frac{\delta^{2}_{0}}{\Upsilon(\sigma)} \right)^{-1}&\nonumber
    \\
    & \leq R(\sigma, \omega_{\sigma}) \left( 1 -
      \frac{\delta^{2}_{0}}{\Gamma_{0}^{m}} \right)^{-1}&
    \text{since $\Upsilon(\sigma) \geq \Gamma^{k}_{0} \geq
      \Gamma^{m}_{0}$}\nonumber\\
    & \leq 2\beta \tilde{t} \, \reach&
  \end{align}

  Let $p$ be a vertex of $\sigma$. Let $c$ be the point closest to
  $C(\sigma)$ on $T_{p}\man$ and $c^{*}$ be the point closest to $c$
  on $\man$ (see Fig.~\ref{fig:forbidden-region-new}).

\begin{figure}
  \begin{center}
    \quad\quad
    \quad\quad
    \quad\quad\includegraphics[width=8.0cm]{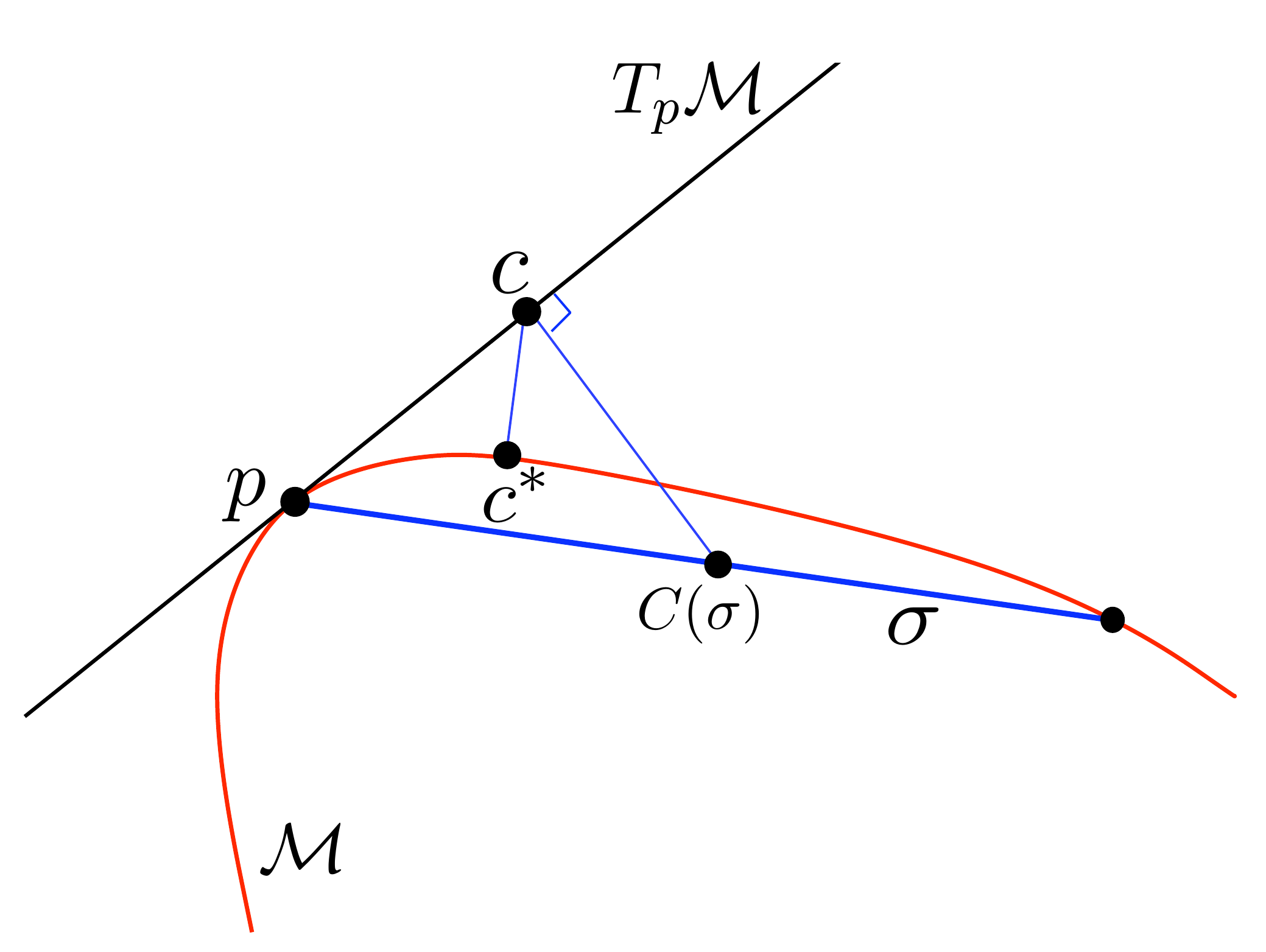}
  \end{center}
  \caption{Proof of Lemma~\ref{lem-bound-flake-forbidden-volume}.}
  \label{fig:forbidden-region-new}
\end{figure}

  From Lemma~\ref{lem:dist.to.tanspace}, we have for all $q \in
  \mathring{\sigma}$
  \begin{equation*}
    \distEm{q}{T_{p}\man} \leq \frac{\|p-q\|^{2}}{2 \reach} \leq
    \frac{\Delta(\sigma)^{2}}{2\, \reach}
    \stackrel{{\rm def}}{=} \eta
  \end{equation*}
	
	From Lemma~\ref{lem:whitney.approx}, and the facts that $\Upsilon(\sigma) \geq \Gamma^{m}_{0}$
	and $\Delta(\sigma) \leq 2R(\sigma) \leq 4\beta t$, we have
	\begin{eqnarray*}
		\sin \angle (T_{p}\man, \aff(\sigma)) \leq \frac{2\eta}{\Upsilon(\sigma) \Delta(\sigma)}
		\leq \frac{\Delta(\sigma)}{\Upsilon(\sigma) \reach} \leq \frac{4\beta\, \tilde{t}}{\Gamma_{0}^{m}}
	\end{eqnarray*}
	
	Therefore
	\begin{equation}\label{eqn-3-lem-forbidden-region}
		\| c - C(\sigma) \| \leq \sin \angle (T_{p}\man, \aff(\sigma))\times R(\sigma)
		\leq \left( \frac{4\beta\, \tilde{t}}{\Gamma^{m}_{0}}\right) R(\sigma) ,
	\end{equation}
	and from \Lemref{lem:Rp.small}% Lemma~\ref{lem-new-addition-1}~(1)
	\begin{equation}\label{eqn-2-lem-forbidden-region}
		\| c - c^{*} \| \leq \frac{2\|c-p\|^{2}}{\reach} \leq 4\beta \tilde{t}\, R(\sigma).
	\end{equation}
	
	Let $x \in F(\sigma, t)$ and $x^{*}$ be the point closest to $x$ on
	$\partial \, B_{\reel^{N}}(C(\sigma), R(\sigma)) \cap \aff(\sigma)$. Then from Lemma~\ref{lem-modified-torus-lemma},
	we have
	\begin{equation}\label{eqn-1-lem-forbidden-region}
		\|x-x^{*}\| < B\Gamma_{0} R(\sigma)
	\end{equation}

	Using the fact that $\|C(\sigma) - x^{*}\| = R(\sigma)$, we get
	\begin{align*}
          \|c^{*} - x \| &\leq \| c^{*} - c \| + \| c - C(\sigma)\| + \| C(\sigma) -x^{*} \| + \| x^{*} -  x\|& \\
          &< R(\sigma) \left( 1+ B\Gamma_{0} + 4\beta t \,
            \left(\frac{1}{\Gamma^{m}_{0}}+1 \right) \right)&
          \text{from Eq.~\eqref{eqn-3-lem-forbidden-region},
            \eqref{eqn-2-lem-forbidden-region},
            \eqref{eqn-1-lem-forbidden-region}}\\
          &\leq R(\sigma) \left( 1+ B\Gamma_{0} + \frac{8\beta \,
              \tilde{t}}{\Gamma^{m}_{0}} \, \right)&
          \text{since $\Gamma_{0} \leq 1$}\\
          &\leq R(\sigma) (1+(B+1) \Gamma_{0})& \text{from
            hyp.~\ref{vol:rad.bnd} of the lemma.}
	\end{align*}
	Similarly we can show that
	\begin{equation*}
		\|c^{*} -x\| < R(\sigma)(1-(B+1)\Gamma_{0})
	\end{equation*}
	
	Therefore
	$$
	F(\sigma, t) \subseteq \left( B_{\reel^{N}}(c^{*}, (1+\zeta) R(\sigma)) \setminus
    B_{\reel^{N}}(c^{*}, (1-\zeta)R(\sigma)) \right) \cap \man
	$$
	where $\zeta = (B+1) \Gamma_{0}$.

	Observe that Lemma~\ref{lem-volume-lemma} can be applied since
    \begin{align*}
        R(\sigma) (1+ \zeta) &\leq 2R(\sigma)& \text{since $\zeta \leq 1$
          from hyp.~\ref{vol:flake.bnd}}\\
        &\leq 4\beta t& \text{from Eq.~\eqref{eqn-radius-sigma-beta}}\\
        &\leq \xi& \text{because $t < \epsilon$.}
    \end{align*}

%    $R(\sigma)(1+\zeta)
%	\leq 2R(\sigma) \leq 4\beta t \leq \xi$ (from hyp.~(c) of the lemma).
    Therefore
	\begin{eqnarray}\label{eqn-20-lem-forbidden-region}
          \frac{\vol(F(\sigma, t))}{\ballvolm} &\leq&
          \frac{\vol(
            B_{\reel^{N}}(c^{*}, R(\sigma)(1+\zeta))\cap \man \setminus
            B_{\reel^{N}}(c^{*}, R(\sigma)(1-\zeta))\cap \man )}{\ballvolm}
          \nonumber\\
          &\leq& (1+A(1+\zeta) \tilde{t}\, )R(\sigma)^{m}(1+\zeta)^{m}
          - (1-A(1-\zeta)\tilde{t} \,
          )R(\sigma)^{m}(1+\zeta)^{m}\nonumber\\
          &\leq&
          R(\sigma)^{m} ((1+\zeta)^{m} - (1-\zeta)^{m}) + A \tilde{t}
          R(\sigma)^{m}((1+\zeta)^{m} + (1-\zeta)^{m}) \nonumber\\
          &\leq& 2^{m} \zeta \, R(\sigma)^{m}+A(2^{m+1}+1) \tilde{t} \, R(\sigma)^{m}
	\end{eqnarray}
	The last inequality follows from the fact that
        $(1+x)^{m}-(1-x)^{m} \leq 2^{m}x$ for all $x \in [0, 1]$.
	
	From hyp.~\ref{vol:rad.bnd} and the fact that $\Gamma_{0}< 1$, we have
	\begin{equation}\label{eqn-21-lem-forbidden-region}
		\tilde{t} \leq \tilde{\epsilon} \leq
                \frac{\Gamma_{0}^{m+1}}{8\beta} < \Gamma_{0}.
	\end{equation}
	
	The lemma now follows from Eq.~\eqref{eqn-20-lem-forbidden-region} and
	\eqref{eqn-21-lem-forbidden-region}.
\end{proof}

\phantomsection
\bibliographystyle{alpha}
\addcontentsline{toc}{section}{Bibliography}
\bibliography{delrefs}

\end{document}